\documentclass[journal,onecolumn]{IEEEtran}

\usepackage{graphicx}
\usepackage{amssymb}
\usepackage{epsfig}
\usepackage{amsmath}
\usepackage{amsthm}
\usepackage{amsfonts}
\usepackage{setspace}
\usepackage{url}
\usepackage{ifthen}
\usepackage[active]{srcltx}
\usepackage{subfigure}
\usepackage{ifpdf}

\newtheorem{theorem}{Theorem}
\newtheorem{cor}{Corollary}
\newtheorem{lem}{Lemma}

\newcommand{\underb}{\underbrace}

\newcommand{\udl}{\underline}
\newcommand{\lv}{\lVert}
\newcommand{\rv}{\rVert}
\newcommand{\lf}{\left}
\newcommand{\ri}{\right}
\newcommand{\ol}{\overline}
\newcommand{\mrm}{\mathrm}
\newcommand{\cube}{\mathrm{cube}}
\newcommand{\eff}{\mathrm{eff}}
\newcommand{\rc}{\mathrm{rc}}
\newcommand{\tr}{\mathrm{tr}}
\newcommand{\ex}{\mathrm{ex}}
\newcommand{\lattice}{\mathrm{lattice}}
\newcommand{\packing}{\mathrm{packing}}
\newcommand{\beq}{\begin{equation}}
\newcommand{\eeq}{\end{equation}}
\newcommand{\SNR}{\mathsf{SNR}}
\newcommand{\singlecolumntype}{1}

\begin{document}

%
\title{Fundamental Limits of Infinite Constellations in MIMO Fading Channels}

\author{Yair Yona, and Meir Feder,~\IEEEmembership{Fellow,~IEEE}
\ifthenelse{\equal{\singlecolumntype}{1}} {\thanks{The material in
this paper was presented in part at the IEEE International Symposium
on Information\newline Theory (ISIT) 2011.}
\thanks{This research was
supported by THE ISRAEL SCIENCE FOUNDATION grant No. 634/09.}
\thanks{The authors are with the Department of Electrical Engineering --
Systems, Tel-Aviv University, Ramat-Aviv 69978, Israel (e-mails:
\{yairyo,meir\}@eng.tau.ac.il).}} {\thanks{The material in this
paper was presented in part at the IEEE International Symposium on
Information Theory (ISIT) 2011.}
\thanks{This research was supported by THE ISRAEL SCIENCE FOUNDATION
grant No. 634/09.}
\thanks{The authors are with the Department of Electrical Engineering --
Systems, Tel-Aviv University, Ramat-Aviv 69978, Israel (e-mails:
\{yairyo,meir\}@eng.tau.ac.il).}}}

\maketitle

\begin{abstract}
The fundamental and natural connection between the infinite
constellation (IC) dimension and the best diversity order it can
achieve is investigated in this paper. In the first part of this
work we develop an upper bound on the diversity order of IC's for
any dimension and any number of transmit and receive antennas. By
choosing the right dimensions, we prove in the second part of this
work that IC's in general and lattices in particular can achieve the
optimal diversity-multiplexing tradeoff of finite constellations.
This work gives a framework for designing lattices for
multiple-antenna channels using lattice decoding.
\end{abstract}

\IEEEpeerreviewmaketitle

\section{Introduction}
The use of multiple antennas in wireless communication has certain
inherent advantages. On one hand, using multiple antennas in fading
channels allows to increase the transmitted signal reliability, i.e.
diversity. For instance, diversity can be attained by transmitting
the same information on different paths between
transmitting-receiving antenna pairs with i.i.d Rayleigh fading
distribution. The number of independent paths used is the diversity
order of the transmitted scheme. On the other hand, the use of
multiple antennas increases the number of degrees of freedom
available by the channel. In
\cite{TelatarCapacityMIMO},\cite{FoschiniCapacityMIMO} the ergodic
channel capacity was obtained for multiple-input multiple-output
(MIMO) systems with $M$ transmit and $N$ receive antennas, where the
paths have i.i.d Rayleigh fading distribution. It was shown that for
large signal to noise ratios ($\SNR$), the capacity behaves as
$C(\SNR)\approx \min(M,N)\log(\SNR)$. The multiplexing gain is the
number of degrees of freedom utilized by the transmitted scheme.

For the quasi-static Rayleigh flat-fading channel, Zheng and Tse
\cite{TseDivMult2003} characterized the dependence between the
diversity order and the multiplexing gain, by deriving the optimal
tradeoff between diversity and multiplexing, i.e. for each
multiplexing gain the maximal diversity order was found. They showed
that the optimal diversity-multiplexing tradeoff (DMT) can be
attained by ensemble of i.i.d Gaussian codes, given that the block
length is greater or equal to $N+M-1$. For this case, the tradeoff
curve takes the form of the piecewise linear function that connects
the points $(N-l)(M-l)$, $l=0,1,\dots, \min(M,N)$.

Space-time codes are coding schemes designed for MIMO systems e.g.
see
\cite{CalderbankSpaceTimePerformanceCriterion},\cite{CalderbankSpaceTimeOrthogonal}
\cite{EliaExplicitSTC} and references therein. The design of
space-time codes in these works pursue various goals such as
maximizing the diversity order, maximizing the multiplexing gain, or
achieving the optimal DMT. El Gamal et al \cite{ElGamalLAST2004}
were the first to show that lattice coding and decoding achieve the
optimal DMT. They presented lattice space-time (LAST) codes. These
space time codes are subsets of an infinite lattice, where the
lattice dimensionality equals to the number of degrees of freedom
available by the channel, i.e. $\min(M,N)$, multiplied by the number
of channel uses. By using a random ensemble of nested lattices,
common randomness, minimum mean square error (MMSE) estimation
followed by lattice decoding and modulo lattice operation, they
showed that LAST codes can achieve the optimal DMT. It is worth
mentioning that the MMSE estimation and the modulo operation take in
a certain sense into account the finite code book.

There has been an extensive research on explicit coding schemes,
based on lattices, which are DMT optimal. Such an explicit coding
schemes that attain the optimal DMT for any number of transmit and
receive antennas were presented in \cite{EliaExplicitSTC}. In
addition it was shown in \cite{EliaExplicitSTC} that $M$ channel
uses are sufficient to obtain the optimal DMT. Another step towards
finding explicit space-time coding schemes that attain the optimal
DMT with low computational complexity was made by Jalden and Elia
\cite{EliaJaldenDMTOptLRLinearLaticeDecoding}. They considered
explicit coding schemes based on the intersection between an
underlying lattice and a shaping region. They showed that for the
cases where these coding schemes attain the optimal DMT using
maximum-likelihood (ML) decoding, they also attain it when using
MMSE estimation in the receiver, followed by lattice decoding. The
MMSE estimation relies on the power constraint, i.e. the shaping
region boundaries. In addition, it was shown in
\cite{EliaJaldenDMTOptLRLinearLaticeDecoding} that by applying
lattice reduction methods, the optimal DMT is attained when using
suboptimal linear lattice decoders that require linear complexity as
a function of the rate. This result applies to wide range of
explicit space-time codes such as golden-codes
\cite{BelfioreGoldenCodes}, perfect space-time codes
\cite{BelfiorePerfectCodes} and in general cyclic division algebra
based space-time codes \cite{EliaExplicitSTC}, and as this codes are
approximately universal
\cite{TavildarViswanathApproximatelyUniversal} it also applies to
every statistical characterization of the fading channel. Note that
these schemes take into consideration the finiteness of the codebook
in the decoder. In our work we refer to \emph{regular} lattice
decoding as decoding over the infinite lattice without taking into
consideration the finiteness of the codebook.

The work in \cite{ElGamalLAST2004} also includes for the case $N\ge
M$ a lower bound on the diversity order of LAST codes shaped into a
sphere when regular lattice decoder is employed in the receiver. For
sufficiently large block length it is shown that $d(r)\ge
(N-M+1)(M-r)$ where $r$ is the multiplexing gain and the lattice
dimension per channel use is $M$. Taherzadeh and Khandani showed in
\cite{TaherzadehLimitationsNaiveLatticeDecoding} that this is also
an upper bound on the diversity order of any LAST code shaped into a
sphere and decoded with \emph{regular} lattice decoding. These
results show that LAST codes together with regular lattice decoding
are suboptimal compared to the optimal DMT of power constrained
constellations.

Infinite constellations (IC's) are structures in the Euclidean space
that have no power constraint. In \cite{PoltirevJournal}, Poltyrev
analyzed the performance of IC's over the additive white Gaussian
noise (AWGN) channel. In this work we first extend the definitions
of diversity order and multiplexing gain to the case where there is
no power constraint. We also introduce a new term: the average
number of dimensions per channel use, which is essentially the IC
dimension divided by the number of channel uses. Then we extend the
methods used in \cite{PoltirevJournal} in order to derive an upper
bound on the diversity of any IC with certain average number of
dimensions per channel use, as a function of the multiplexing gain.
It turns out that for a given number of dimensions per channel use
the diversity is a straight line as a function of the multiplexing
gain, that depends on the number of transmit and receive antennas.
This analysis holds for $any$ $M$ and $N$, and also applies for
lattices with regular lattice decoding. We also find the average
number of dimensions per channel use for which the upper bounds
coincide with the optimal DMT of finite constellations. Finally, we
show that each segment in the optimal DMT is attained by a sequence
of lattices with a corresponding average number of dimensions per
channel use, when using regular lattice decoder, i.e. for each point
in the DMT of \cite{TseDivMult2003} there exists a lattice sequence
of certain dimension that achieves it with regular lattice decoding.
Hence, this work characterizes the best DMT IC's may attain for any
average number of dimensions per channel use, and also proves that
lattices can achieve the optimal DMT when \emph{regular} lattice
decoder is employed in the receiver, by adapting their
dimensionality. It is important to note that when the IC is a
lattice, we show that the multiplexing gain of infinite lattices and
finite constellations coincide.

This work gives a framework for designing lattices for
multiple-antenna channels using regular lattice decoding. It also
shows the fundamental and natural connection between the IC
dimension and its optimal diversity order. For instance, it is shown
that for the case $M=N=2$, the maximal diversity order of $4$ can be
achieved (with regular lattice decoding) by a lattice that has at
most $\frac{4}{3}$ average number of dimensions per channel use. On
the other hand the Alamouti scheme \cite{AlamoutiScheme}, that also
has maximal diversity order of $4$, utilizes only a single dimension
per channel use in this set up. Hence, there is still a room to
improve by a $\frac{1}{3}$ of a dimension per channel use. In
addition, while in \cite{ElGamalLAST2004},
\cite{EliaJaldenDMTOptLRLinearLaticeDecoding}, the MMSE estimation
improves the channel in such a manner that enables the lattice
decoder to attain the optimal DMT, this work shows that when
considering regular lattice decoding, reducing the lattice
dimensionality takes the role of MMSE estimation in the sense of
improving the channel such that the optimal DMT is obtained.
Finally, the analysis in this work gives another geometrical
interpretation to the optimal DMT.

The outline of the paper is as follows. In section
\ref{sec:BasicDefinitions} basic definitions for the fading channel
and IC's are given. Section \ref{sec:LowerBoundErrorProb} presents
for each channel realization a lower bound on the average decoding
error probability of any IC, and an upper bound on the DMT of any
IC. An upper bound on the error probability of ensemble of IC's for
each channel realization, a transmission scheme that attains the
optimal DMT, and some averaging arguments on how the optimal DMT is
attained by IC's, are all presented in section
\ref{sec:LowerBoundDiversityOrder}. Discussion on the results, that
addresses the difference between lattice constellations and full
dimension lattice based finite constellations, followed by a
geometrical interpretation to the optimal DMT, and a discussion on
the relation between the multiplexing gains of an IC and a finite
constellation, is presented in section
\ref{sec:LatticeVsLatticeBasedFC}. This discussion presents an
intuitive interpretation to our results and relies mainly on the
basic definitions given in section \ref{sec:BasicDefinitions}.

\section{Basic Definitions}\label{sec:BasicDefinitions}
We refer to the countable set $S=\{s_{1},s_{2},\dots\}$ in
$\mathbb{C}^{n}$ as infinite constellation (IC). Let
$\cube_{l}(a)\subset\mathbb{C}^{n}$ be a (probably rotated)
$l$-complex dimensional cube ($l\le n$) with edge of length $a$
centered around zero. An IC $S_{l}$ is $l$-complex dimensional if
there exists rotated $l$-complex dimensional cube $\cube_{l}(a)$
such that $S_{l}\subset\lim_{a\to\infty}\cube_{l}(a)$ and \emph{l}
is minimal. $M(S_{l},a)=|S_{l}\bigcap \cube_{l}(a)|$ is the number
of points of the IC $S_{l}$ inside $\cube_{l}(a)$. In
\cite{PoltirevJournal}, the $n$-complex dimensional IC density for
the AWGN channel was defined as the upper limit (the limit supremum)
of the ratio
$\gamma_{\mathrm{G}}=\limsup_{a\to\infty}\frac{M(S,a)}{a^{2n}}$ and
the volume to noise ratio (VNR) was given as
$\mu_{\mathrm{G}}=\frac{\gamma_{\mathrm{G}}^{-\frac{1}{n}}}{2\pi
e\sigma^{2}}$.

The Voronoi region of a point $x\in S_{l}$, denoted as $V(x)$, is
the set of points in $\lim_{a\to\infty}\cube_{l}(a)$ closer to $x$
than to any other point in the IC. The effective radius of the point
$x\in S_{l}$, denoted as $r_{\mathrm{eff}}(x)$, is the radius of the
$l$-complex dimensional ball that has the same volume as the Voronoi
region, i.e. $r_{\mathrm{eff}}(x)$ satisfies
\begin{equation}\label{eq:EffectiveRadius}
|V(x)|=\frac{\pi^{l} r_{\mathrm{eff}}^{2\cdot l}(x)}{\Gamma(l+1)}.
\end{equation}

A complex lattice $\Lambda$ is an IC that constitutes a discrete set
in $\mathbb{C}^{n}$, closed under addition. The Voronoi regions of
all lattice points are identical and satisfy
\begin{equation}
|V \lf(\udl{x}\ri)|=\gamma_{G}^{-1}\quad\forall \udl{x}\in\Lambda.
\end{equation}
Hence, for large dimension the VNR of a lattice, $\mu_{G}$,
approaches the ratio $\frac{r_{\eff}^{2}}{\sigma^{2}}$ where
$r_{\eff}$ is the lattice effective radius. Regular lattice decoder
finds the closest lattice point to an observation
$\udl{y}\in\mathbb{C}^{n}$, i.e. \emph{regular} lattice decoder
finds the solution to the optimization problem
\begin{equation}\label{eq:BasicDefinitionsRegularLatticeDecoding}
\arg \min_{\udl{x}\in\Lambda}\lVert\udl{y}-\udl{x}\rVert.
\end{equation}
Note that these definitions can be also extended in a straight
forward manner to an IC that constitutes a real lattice in
$\mathbb{R}^{2n}$. For instance when the first $n$ entries of each
lattice point are transmitted on the real part of the IC, and the
second $n$ entries of each lattice point are transmitted on the
imaginary part of the IC.

We consider a quasi static flat-fading channel with $M$ transmit and
$N$ receive antennas. We assume for this MIMO channel perfect
channel knowledge at the receiver and no channel knowledge at the
transmitter. The channel model is as follows:
\begin{equation} \label{eq:Channel Fading}
\underline{y}_{t}=H\cdot\underline{x}_{t}+\rho^{-\frac{1}{2}}\underline{n}_{t}\qquad
t=1,\dots, T
\end{equation}
where $\udl{x}_{t}$, $t=1,\dots, T$ is the transmitted signal,
$\underline{n}_{t}\sim CN(\underline{0},\frac{2}{2\pi e}I_{N})$ is
the additive noise where $CN$ denotes complex-normal, $I_{N}$ is the
$N$-dimensional unit matrix, and $\udl{y}_{t}\in\mathbb{C}^{N}$. $H$
is the fading matrix with $N$ rows and $M$ columns where
$h_{i,j}\sim CN(0,1)$, $1\le i \le N$, $1\le j \le M$, and
$\rho^{-\frac{1}{2}}$ is a scalar that multiplies each element of
$\underline{n}_{t}$, where $\rho$ plays the role of average $\SNR$
in the receive antenna for power constrained constellations that
satisfy $\frac{1}{T} \sum_{t=1}^{T}
E\{\lVert\udl{x}_{t}\rVert^{2}\}\le \frac{2}{2\pi e}$.

We also define the extended vector
$\udl{x}=\{\udl{x}_{1}^{\dagger},\dots,\udl{x}_{T}^{\dagger}\}^{\dagger}$.
Suppose $\udl{x}\in S_{l}\subset\mathbb{C}^{MT}$, where $S_{l}$ is
an IC with density
$\gamma_{tr}=\limsup_{a\to\infty}\frac{M(S_{l},a)}{a^{2\cdot l}}$
$\big(a^{2\cdot l}$ is the volume of $cube_{l}(a)\big)$. By defining
$H_{ex}$ as an $NT\times MT$ block diagonal matrix, where each block
on the diagonal equals $H$,
$\underline{n}_{\mathrm{ex}}=\rho^{-\frac{1}{2}}\cdot\{\underline{n}_{1}^{\dagger},
\dots,\underline{n}_{T}^{\dagger}\}^{\dagger}\in\mathbb{C}^{NT}$ and
$\udl{y}_{\mathrm{ex}}\in\mathbb{C}^{NT}$ we can rewrite the channel
model in \eqref{eq:Channel Fading} as
\begin{equation}\label{eq:ExtendedChannelModel}
\underline{y}_{\mathrm{ex}}=H_{\mrm{ex}}\cdot\underline{x}+\underline{n}_{\mathrm{ex}}.
\end{equation}

In the sequel we use $L$ to denote $min(M,N)$. We define as
$\sqrt{\lambda}_{i}$, $1\le i\le L$ the real valued, non-negative
singular values of $H$. We assume $\sqrt{\lambda}_{L}\ge
\dots\ge\sqrt{\lambda}_{1}>0$. Our analysis is done for large values
of $\rho$ (large VNR at the transmitter). We state that
$f(\rho)\dot{\ge}g(\rho)$ when $\lim_{\rho\to\infty}-\frac{\ln\left(
f\left(\rho\right)\right)}{\ln(\rho)}\le
-\frac{\ln\left(g\left(\rho\right)\right)}{\ln(\rho)}$, and also
define $\dot{\le}$, $\dot{=}$ in a similar manner by substituting
$\le$ with $\ge$, $=$ respectively.

We now turn to the IC definitions in the transmitter. We define the
average number of dimensions per channel use as the IC dimension
divided by the number of channel uses. We denote the average number
of dimensions per channel use by $K$. Let us consider a $KT$-complex
dimensional sequence of IC's $S_{KT}(\rho)$, where $K\le L$, and $T$
is the number of channel uses. First we define
$\gamma_{tr}=\rho^{rT}$ as the density of $S_{KT}(\rho)$ in the
transmitter. The IC multiplexing gain is defined as
\begin{equation}\label{eq:ICMGDefinition}
MG(r)=\lim_{\rho\to\infty}\frac{1}{T}\log_{\rho}(\gamma_{\mathrm{tr}}+1)=
\lim_{\rho\to\infty}\frac{1}{T}\log_{\rho}(\rho^{rT}+1).
\end{equation}
Note that $MG(r)=max(0,r)$, i.e. for $0\le r\le K$ the multiplexing
gain is $r$. Roughly speaking, $\gamma_{tr}=\rho^{rT}$ gives us the
number of points of $S_{KT}(\rho)$ within the $KT$-complex
dimensional region $cube_{KT}(1)$. In order to get the multiplexing
gain, we normalize the exponent of the number of points within
$cube_{KT}(1)$, $rT$, by the number of channel uses - $T$. Note that
the IC multiplexing gain, $r$, can be directly translated to finite
constellation multiplexing gain $r$ by considering the IC points
within a shaping region. For more details see
\ref{subseec:LatticeVsLatticeBasedFCICMGtoFCMG}. The VNR in the
transmitter is
\begin{equation}
\mu_{\mathrm{tr}}=\frac{\gamma_{\mathrm{tr}}^{-\frac{1}{KT}}}{2\pi
e\sigma^{2}}=\rho^{1-\frac{r}{K}}
\end{equation}
where $\sigma^{2}=\frac{\rho^{-1}}{2\pi e}$ is each dimension noise
variance. Now we can understand the role of the multiplexing gain
for IC's. The AWGN variance decreases as $\rho^{-1}$, where the IC
density increases as $\rho^{rT}$. When $r=0$ we get constant IC
density as a function of $\rho$, where the noise variance decreases,
i.e. we get the best error exponent. In this case the number of
points within $cube_{KT}(1)$ remains constant as a function of
$\rho$. On the other hand, when $r=K$, we get VNR $\mu_{tr}=1$, and
from \cite{PoltirevJournal} we know that it inflicts average error
probability that is bounded away from zero. In this case, the
increase in the number of IC points within $cube_{KT}(1)$ occurs at
maximal rate.

Now we turn to the IC definitions in the receiver. First we define
the set $H_{ex}\cdot cube_{KT}(a)$ as the multiplication of each
point in $cube_{KT}(a)$ with the matrix $H_{ex}$. In a similar
manner $S_{KT}^{'}=H_{ex}\cdot S_{KT}$. The set $H_{ex}\cdot
cube_{KT}(a)$ is almost surely $KT$-complex dimensional (where $K\le
L$) and in this case $M(S_{KT},a)=|S_{KT}\bigcap
\cube_{KT}(a)|=|S_{KT}^{'}\bigcap (H_{ex}\cdot\cube_{KT}(a))|$. We
define the receiver density as
$$\gamma_{\mathrm{rc}}=\limsup_{a\to\infty}\frac{M(S_{KT},a)}{\mathbf{Vol}(H_{ex\cdot}\cube_{KT}(a))}$$
i.e., the upper limit of the ratio of the number of IC points in
$H_{ex\cdot}\cube_{KT}(a)$, and the volume of
$H_{ex\cdot}\cube_{KT}(a)$. Based on the majorization property of a
matrix singular values \cite{JiangGMD}, we get that the volume of
the set $H_{ex}\cdot cube_{KT}(a)$ is smaller than
$a^{2KT}\cdot\lambda_{L}^{T}\dots\lambda_{L-B+1}^{T}\cdot\lambda_{L-B}^{\beta
T}$, assuming $K=B+\beta$ where $B\in\mathbb{N}$ and $0<\beta\le 1$,
i.e. the volume is smaller than the multiplication of the $B+1$
strongest singular values, raised to the power of the maximal amount
of channel uses each can take place in. Hence we get
\begin{equation}\label{eq:GammaRcUpperBound}
\gamma_{\mathrm{rc}}\ge\rho^{rT}\lambda_{L}^{-T}\dots\lambda_{L-B+1}^{-T}\cdot\lambda_{L-B}^{-\beta
T}
\end{equation}
and the receiver VNR is
\begin{equation}\label{eq:MuRc}
\mu_{\mathrm{rc}}\le\rho^{1-\frac{r}{K}}\cdot\lambda_{L}^{\frac{1}{K}}\dots\lambda_{L-B+1}^{\frac{1}{K}}\cdot\lambda_{L-B}^{\frac{\beta}{K}}.
\end{equation}
Note that for $N\ge M$ and $K=M$ we get
$\gamma_{\mathrm{rc}}=\rho^{rT}\cdot\prod_{i=1}^{M}\lambda_{i}^{-T}$
and
$\mu_{\mathrm{rc}}=\rho^{1-\frac{r}{M}}\cdot\prod_{i=1}^{M}\lambda_{i}^{\frac{1}{M}}$.
The average decoding error probability over the IC points of
$S_{KT}(\rho)$, for a certain channel realization $H$, is defined as
\begin{equation}\label{eq:AverageDecodingErrorProbability}
\overline{P_{e}}(H,\rho)=\limsup_{a\to\infty}\frac{\sum_{\underline{x}^{'}\in
S_{KT}^{'}\bigcap
(H_{ex}\cdot\cube_{KT}(a))}P_{e}(\underline{x}^{'},H,\rho)}{M(S_{KT},a)}
\end{equation} where $P_{e}(\underline{x}^{'},H,\rho)$ is the error
probability associated with $\udl{x}^{'}$. The average decoding
error probability of $S_{KT}(\rho)$ over all channel realizations is
$\overline{P_{e}}(\rho)=E_{H}\{\overline{P_{e}}(H,\rho)\}$. Hence
the $diversity$ $order$ equals
\begin{equation}\label{eq:DiversityOrder}
d=-\lim_{\rho\to\infty}\log_{\rho}(\overline{P_{e}}(\rho))
\end{equation}

\section{Upper Bound on the Diversity
Order}\label{sec:LowerBoundErrorProb} In this section we derive an
upper bound on the diversity order of any IC with average number of
dimensions per channel use $K$ and any value of $T$, $M$ and $N$. In
Theorem \ref{Th:LowerBoundChannelReal} we derive for each channel
realization a lower bound on the error probability of any IC with
$K$ average number of dimensions per channel use. In Theorem
\ref{Th:UpperBoundDiversityOrder} we derive an upper bound on the
DMT of any sequence of IC's with $K$ average number of dimensions
per channel use. Finally in Corollary \ref{Cor:MaximalDiverityOrder}
we show that by choosing the correct average number of dimensions
per channel use, the upper bound coincides with the optimal DMT of
finite constellations.

As in \cite{TseDivMult2003} and \cite{ElGamalLAST2004}, we also
define $\lambda_{i}=\rho^{-\alpha_{i}}$, $1\le i \le L$. When the
entries of the channel matrix $H$ are all i.i.d with PDF $CN(0,1)$,
the PDF of its singular values is of the form
$\rho^{-\sum_{i=1}^{L}(|N-M|+2i-1)\alpha_{i}}$ for large $\rho$
\cite{TseDivMult2003}, where following the definitions above
$0\le\alpha_{L}\le\dots \le\alpha_{1}$. \footnote{A generalization
of the Rayleigh fading channel is the Jacobi fading channel. The
optimal DMT for this channel was derived in
\cite{DarFederJacobiChannel}.} By assigning in
\eqref{eq:GammaRcUpperBound}, \eqref{eq:MuRc} respectively, we can
write
$$\gamma_{\mathrm{rc}}\ge\rho^{T(r+\sum_{i=0}^{B-1}\alpha_{L-i}+\beta\alpha_{L-B})}$$
and
$$\mu_{\mathrm{rc}}\le\rho^{1-\frac{1}{K}(r+\sum_{i=0}^{B-1}\alpha_{L-i}+\beta\alpha_{L-B})}.$$

\begin{theorem}\label{Th:LowerBoundChannelReal}
For any $KT$-complex dimensional IC $S_{KT}(\rho)$ with transmitter
density $\gamma_{\mathrm{tr}}=\rho^{rT}$ and channel realization
$\underline{\alpha}=(\alpha_{1},\dots,\alpha_{L})$, we have the
following lower bound on the average decoding error probability for
$0\le r\le K$
$$\overline{P_{e}}(H,\rho)> \frac{C(KT)}{4}e^{-\mu_{\mathrm{rc}}\cdot A(KT)+(KT-1)\ln
(\mu_{\mathrm{rc}})}$$ where
$A(KT)=e\cdot\Gamma(KT+1)^{\frac{1}{KT}}$ and
$C(KT)=\frac{e^{KT-\frac{3}{2}}\Gamma(KT+1)^{\frac{KT-1}{KT}}}{2\cdot\Gamma(KT)}$.
\end{theorem}
\begin{proof}
We divide the proof into two parts. In the first part we prove the
result for lattices, that constitute a symmetric structure for which
the Voronoi regions of different lattice points are identical. In
the second part we prove the result for general IC's with receiver
density $\gamma_{rc}$. As the second part of the proof is somewhat
more involved, we defer it to appendix
\ref{append:ProofFirstTheorem}. Note that we could have used the
tighter bounds of \cite{IngberZamirFederDensityIC}, but these bounds
are not needed for DMT. Instead we derive coarser and more
simplified upper bounds, which are sufficient for our purposes.

We begin by proving the result for lattices. Lattices constitute a
discrete subgroup of the Euclidean space, with the ordinary vector
addition operation. Consider a $KT$-complex dimensional lattice,
$S^{'}_{KT}(\rho)$, in the receiver with density $\gamma_{rc}$. The
lattice points have identical Voronoi regions up to a translation.
Hence, the volume of each Voronoi region equals
$$|V(x)|=\frac{1}{\gamma_{rc}}\quad\forall x\in S^{'}_{KT}(\rho).$$
According to the definition of the effective radius in
\eqref{eq:EffectiveRadius}, we get that $r_{\eff}(x)=
r_{\eff}(\gamma_{rc})=(\frac{\Gamma(KT+1)}{\gamma_{rc}\pi^{KT}})^\frac{1}{2KT}$,
$\forall x\in S^{'}_{KT}(\rho)$. Note that in lattices the
maximum-likelihood (ML) decoding error probability is identical for
all lattice points, i.e. the average and maximal error probabilities
are identical. It has been proven in \cite{PoltirevJournal},
\cite{TarokhUniversalBoundPerformanceLattices} that the error
probability of any lattice point in the receiver fulfils
$$P_{e}^{S^{'}_{KT}}>Pr(\lVert
\underline{\tilde{n}}_{\mathrm{ex}}\rVert \ge
r_{\mathrm{eff}}(\gamma_{rc}))$$ where $P_{e}^{S^{'}_{KT}}$ is the
ML decoding error probability of any lattice point, and
$\underline{\tilde{n}}_{\mathrm{ex}}$ is the effective noise in the
$KT$-complex dimensional hyperplane where $S^{'}_{KT}(\rho)$
resides. We find an explicit expression for the lower bound
\ifthenelse{\equal{\singlecolumntype}{1}}
{\begin{equation}\label{eq:LowerBoundCertainChannelRealizationForLattices}
\Pr\big(\lVert\underline{\tilde{n}}_{\mathrm{ex}}\rVert\ge
r_{\mathrm{eff}}(\gamma_{\mathrm{rc}})\big)>\Pr\big(\lVert\underline{\tilde{n}}_{\mathrm{ex}}\rVert\ge
r_{\mathrm{eff}}(\frac{\gamma_{\mathrm{rc}}}{2})\big)>
\int_{r_{\mathrm{eff}}^{2}}^{r_{\mathrm{eff}}^{2}+\sigma^{2}}\frac{r^{KT-1}e^{-\frac{r}{2\sigma^{2}}}}{\sigma^{2KT}2^{KT}
\Gamma(KT)}dr
\ge\frac{r_{\mathrm{eff}}^{2KT-2}e^{-\frac{r_{\mathrm{eff}}^{2}}{2\sigma^{2}}}}{\sigma^{2KT-2}2^{KT}
\Gamma(KT)\sqrt{e}}.
\end{equation}}
{\begin{align}\label{eq:LowerBoundCertainChannelRealizationForLattices}
\Pr\big(\lVert\underline{\tilde{n}}_{\mathrm{ex}}\rVert\ge
r_{\mathrm{eff}}(\gamma_{\mathrm{rc}})\big)&>\Pr\big(\lVert\underline{\tilde{n}}_{\mathrm{ex}}\rVert\ge
r_{\mathrm{eff}}(\frac{\gamma_{\mathrm{rc}}}{2})\big)>
\nonumber\\
\int_{r_{\mathrm{eff}}^{2}}^{r_{\mathrm{eff}}^{2}+\sigma^{2}}\frac{r^{KT-1}e^{-\frac{r}{2\sigma^{2}}}}{\sigma^{2KT}2^{KT}
\Gamma(KT)}&dr
\ge\frac{r_{\mathrm{eff}}^{2KT-2}e^{-\frac{r_{\mathrm{eff}}^{2}}{2\sigma^{2}}}}{\sigma^{2KT-2}2^{KT}
\Gamma(KT)\sqrt{e}}.
\end{align}}
By assigning
$r_{\eff}^{2}=(\frac{2\cdot\Gamma(KT+1)}{\gamma_{rc}\pi^{KT}})^\frac{1}{KT}$
we get
$$P_{e}^{S^{'}_{KT}}> C(KT)\cdot
e^{-\frac{\gamma_{\mathrm{rc}}^{-\frac{1}{KT}}}{2\pi
e\sigma^{2}}A(KT)+(KT-1)\ln
(\frac{\gamma_{\mathrm{rc}}^{-\frac{1}{KT}}}{2\pi e\sigma^{2}})}$$
and by assigning
$\mu_{rc}=\frac{\gamma_{\mathrm{rc}}^{-\frac{1}{KT}}}{2\pi
e\sigma^{2}}$ we get
\begin{equation}\label{eq:LowerBoundErrorProbWithVNRForLattices}
P_{e}^{S^{'}_{KT}}> \frac{C(KT)}{4}\cdot e^{-\mu_{rc}A(KT)+(KT-1)\ln
(\mu_{rc})}.
\end{equation}
Note that in
\eqref{eq:LowerBoundCertainChannelRealizationForLattices} we lower
bounded the error probability with $r_{\eff}(\frac{\gamma_{rc}}{2})$
instead of $r_{\eff}(\gamma_{rc})$, and also in
\eqref{eq:LowerBoundErrorProbWithVNRForLattices} we multiplied by
$\frac{1}{4}$, in order to be consistent with the general lower
bound for IC's shown in appendix \ref{append:ProofFirstTheorem}. For
lattices we have $\overline{P_{e}}(H,\rho)=P_{e}^{S^{'}_{KT}}$.
Essentially what we have shown here is a scaled sphere packing
bound.\footnote{Note that while Theorem
\ref{Th:LowerBoundChannelReal} refers to $KT$-complex dimensional
IC's, the lower bound derived in this theorem applies for any
$2KT$-real dimensional IC.}
\end{proof}

Next, we would like to use this lower bound to average over the
channel realizations and get an upper bound on the diversity order.

\begin{theorem}\label{Th:UpperBoundDiversityOrder}
The diversity order of any $KT$-complex dimensional sequence of IC's
$S_{KT}(\rho)$, with $K$ average number of dimensions per channel
use, is upper bounded by
$$d_{KT}(r)\le  d^{\ast}_{K}(r)=M\cdot N(1-\frac{r}{K})$$
for $0 < K\le\frac{M\cdot N}{N+M-1}$, and
$$d_{KT}(r)\le  d^{\ast}_{K}(r)=(M-l)(N-l)\frac{K}{K-l}(1-\frac{r}{K})$$
for $\frac{(M-l+1)(N-l+1)}{N+M-1-2(l-1)}+l-1< K\le
\frac{(M-l)(N-l)}{N+M-1-2\cdot l}+l$ and $l=1,\dots,L-1$. In all of these cases $0\le r \le K$.
\end{theorem}
\begin{proof}
For any IC with VNR $\mu_{rc}$, assigning $\mu_{rc}^{'}>\mu_{rc}$ in
the lower bound from Theorem \ref{Th:LowerBoundChannelReal} also
gives a lower bound on the error probability
$$\overline{P_{e}}(H,\rho)> \frac{C(KT)}{4}e^{-\mu_{\mathrm{rc}}^{'}\cdot
A(KT)+(KT-1)\ln (\mu_{\mathrm{rc}}^{'})}.$$ It results from the fact
that inflating the IC into an IC with VNR $\mu_{rc}^{'}$ must
decrease the error probability, where
\ifthenelse{\equal{\singlecolumntype}{1}}
{$$\frac{C(KT)}{4}e^{-\mu_{\mathrm{rc}}^{'}\cdot A(KT)+(KT-1)\ln
(\mu_{\mathrm{rc}}^{'})}$$}
{$\frac{C(KT)}{4}e^{-\mu_{\mathrm{rc}}^{'}\cdot A(KT)+(KT-1)\ln
(\mu_{\mathrm{rc}}^{'})}$} is a lower bound on the error probability
of any IC with VNR $\mu_{rc}^{'}$. Hence, for the case $\mu_{rc}\le
1$ we can lower bound the error probability by assigning 1 in the
lower bound and get $\frac{C(KT)}{4}e^{-A(KT)}$, i.e. for
$\mu_{rc}\le 1$ the average decoding error probability is bounded
away from 0 for any value of $\rho$. We can give the event
$\mu_{rc}\le 1$ the interpretation of an outage event.

We would like to set a lower bound for the error probability for
each channel realization $\udl{\alpha}$, which we denote by
$P_{e}^{LB}(\rho,\udl{\alpha})$. We know that
$\mu_{\mathrm{rc}}\le\rho^{1-\frac{1}{K}(r+\sum_{i=0}^{B-1}\alpha_{L-i}+\beta\alpha_{L-B})}$.
For the case $\sum_{i=0}^{B-1}\alpha_{L-i}+\beta\alpha_{L-B}<K-r$,
we take
$$P_{e}^{LB}(\rho,\udl{\alpha})=\frac{C(KT)}{4}e^{-L(\rho,\udl{\alpha})\cdot
A(KT)+(KT-1)\ln (L(\rho,\udl{\alpha}))}$$ where
$L(\rho,\udl{\alpha})=\rho^{1-\frac{1}{K}(r+\sum_{i=0}^{B-1}\alpha_{L-i}+\beta\alpha_{L-B})}>1$.
For the case $\sum_{i=0}^{B-1}\alpha_{L-i}+\beta\alpha_{L-B}\ge K-r$
we get that $\mu_{rc}\le 1$, and we take
$$P_{e}^{LB}(\rho,\udl{\alpha})=\frac{C(KT)}{4}e^{-A(KT)}.$$

In order to find an upper bound on the diversity order, we would
like to average $P_{e}^{LB}(\rho,\udl{\alpha})$ over
the channel realizations. In our analysis we consider large values
of $\rho$, and so we calculate
\begin{equation}\label{eq:LowerBoundInt}
\ol{P_{e}}(\rho)\dot{>}\int_{\udl{\alpha}\ge 0}
P_{e}^{LB}(\rho,\udl{\alpha})\cdot
\rho^{-\sum_{i=1}^{L}(|N-M|+2i-1)\alpha_{i}}d\udl{\alpha}
\end{equation}
where $\udl{\alpha}\ge 0$ signifies the fact that
$\alpha_{1}\ge\dots \ge\alpha_{L}\ge 0$. By defining
$\mathcal{A}=\{\udl{\alpha}|\sum_{i=0}^{B-1}\alpha_{L-i}+\beta\alpha_{L-B}<K-r;\udl{\alpha}\ge
0\}$ and
$\ol{\mathcal{A}}=\{\udl{\alpha}|\sum_{i=0}^{B-1}\alpha_{L-i}+\beta\alpha_{L-B}\ge
K-r;\udl{\alpha}\ge 0\}$ we can split \eqref{eq:LowerBoundInt} into
2 terms
\ifthenelse{\equal{\singlecolumntype}{1}}
{\begin{equation}
\ol{P_{e}}(\rho)\dot{>}\int_{\udl{\alpha}\in \mathcal{A}}
P_{e}^{LB}(\rho,\udl{\alpha})\cdot
\rho^{-\sum_{i=1}^{L}(|N-M|+2i-1)\alpha_{i}}d\udl{\alpha}
+\int_{\udl{\alpha}\in \ol{\mathcal{A}}}
P_{e}^{LB}(\rho,\udl{\alpha})\cdot
\rho^{-\sum_{i=1}^{L}(|N-M|+2i-1)\alpha_{i}}d\udl{\alpha}.
\end{equation}}
{\begin{align}
\ol{P_{e}}(\rho)\dot{>}\int_{\udl{\alpha}\in \mathcal{A}}P_{e}^{LB}(\rho,\udl{\alpha})\cdot
\rho^{-\sum_{i=1}^{L}(|N-M|+2i-1)\alpha_{i}}d\udl{\alpha}
\nonumber\\
+\int_{\udl{\alpha}\in \ol{\mathcal{A}}}P_{e}^{LB}(\rho,\udl{\alpha})\cdot
\rho^{-\sum_{i=1}^{L}(|N-M|+2i-1)\alpha_{i}}d\udl{\alpha}.
\end{align}}
Hence
\begin{equation}
\ol{P_{e}}(\rho)\dot{>}\int_{\udl{\alpha}\in
\ol{\mathcal{A}}}P_{e}^{LB}(\rho,\udl{\alpha})\cdot
\rho^{-\sum_{i=1}^{L}(|N-M|+2i-1)\alpha_{i}}d\udl{\alpha}.
\end{equation}

In a similar manner to \cite{TseDivMult2003},
\cite{ElGamalLAST2004}, for very large $\rho$, we approximate the
average value by finding the most dominant exponential term in the
integral. For this we would like to find the minimal value of
$$\lim_{\rho\to\infty}-\log_{\rho}(P_{e}^{LB}(\rho,\udl{\alpha})\cdot
\rho^{-\sum_{i=1}^{L}(|N-M|+2i-1)\alpha_{i}})$$ for the case
$\udl{\alpha}\in \ol{\mathcal{A}}$. For $\udl{\alpha}\in
\ol{\mathcal{A}}$, we get that $P_{e}^{LB}(\rho,\udl{\alpha})$ is
bounded away from 0 for any value of $\rho$. Hence, in order to find
the most dominant error event we would like to find
$\min_{\udl{\alpha}}\sum_{i=1}^{L}(|N-M|+2i-1)\alpha_{i}$ given that
$\udl{\alpha}\in\ol{\mathcal{A}}$. The minimal value is achieved at
the boundary, i.e. for $\udl{\alpha}$ satisfying
$\sum_{i=0}^{B-1}\alpha_{L-i}+\beta\alpha_{L-B}=K-r$,
$\udl{\alpha}\ge 0$. Hence, for any $K\le L$ we state that
\begin{equation}\label{eq:DiversityOrderUpperBoundOptimizationProblem}
d_{KT}(r)\le
\min_{\udl{\alpha}}\sum_{i=1}^{L}(|N-M|+2i-1)\alpha_{i},\qquad 0\le
r\le K
\end{equation}
where $\sum_{i=0}^{B-1}\alpha_{L-i}+\beta\alpha_{L-B}=K-r$ and
$\alpha_{1}\ge\dots\ge\alpha_{L}\ge 0$. Basically this optimization
problem is a linear programming problem whose solution is as
follows. For $0 < K\le\frac{M\cdot N}{N+M-1}$ the solution is
$\alpha_{i}=1-\frac{r}{K}$, $i=1,\dots ,L$. For
$\frac{(M-l+1)(N-l+1)}{N+M-1-2(l-1)}+l-1< K\le
\frac{(M-l)(N-l)}{N+M-1-2\cdot l}+l$ and $l=1,\dots,L-1$  the
solution is $\alpha_{L}=\dots=\alpha_{L-l+1}=0$ and
$\alpha_{L-l}=\dots=\alpha_{1}=\frac{K-r}{K-l}$. The desired upper
is attained by substituting the optimal values of $\udl{\alpha}$ in
\eqref{eq:DiversityOrderUpperBoundOptimizationProblem}. The detailed
solution for the optimization problem is presented in appendix
\ref{Append:OptimizationProblemSolutionOfUpperBoundDivOrder}.
\end{proof}

From Theorem \ref{Th:UpperBoundDiversityOrder} we get an upper bound
on the diversity order by assuming transmission of the $KT$ complex
dimensions over the $B+1$ strongest singular values. This assumption
is equivalent to assuming \emph{beamforming} which may improve the
coding gain, but does not increase the diversity order. This
assumption allows us to derive a lower bound on the average decoding
error probability. However, we still get maximal diversity order of
$MN$ in this case.

Let us consider as an illustrative example the case of $M=N=2$. In
this case, for $0<K\le\frac{4}{3}$ we get $d_{K}^{\ast}(r)=
4(1-\frac{r}{K})$. For $\frac{4}{3}<K\le 2$ we get $d_{k}^{\ast}(r)=
\frac{K}{K-1}(1-\frac{r}{K})$. In both cases $0\le r\le K$. For this
set up we have two singular values and so
$\alpha_{1}\ge\alpha_{2}\ge 0$. The optimization problem is of the
form $\min_{\udl\alpha\ge 0}\alpha_{1}+3\alpha_{2}$, where for
$0<K\le 1$ the constraint is $\beta\alpha_{2}=K-r$, and for $1<K\le
2$ the constraint is $\alpha_{2}+\beta\alpha_{1}=K-r$. For the case
$0<K<\frac{4}{3}$ the optimization problem solution is
$\alpha_{1}=\alpha_{2}=1-\frac{r}{K}$, i.e. in this case the most
dominant error event occurs when both singular values are very
small. For the case $K=\frac{4}{3}$ the constraint is of the form
$\alpha_{2}+\frac{\alpha_{1}}{3}=\frac{4}{3}-r$, and the
optimization problem solution is achieved for both
$\alpha_{1}=\alpha_{2}=1-\frac{3r}{4}$ and $\alpha_{2}=0$,
$\alpha_{1}=4-3r$. For the case $\frac{4}{3}<K\le 2$ the
optimization problem solution is achieved for $\alpha_{2}=0$,
$\alpha_{1}=\frac{K-r}{K-1}$, i.e. one strong singular value and
another very weak singular value.

\begin{figure}[h!]
\centering \epsfig{figure=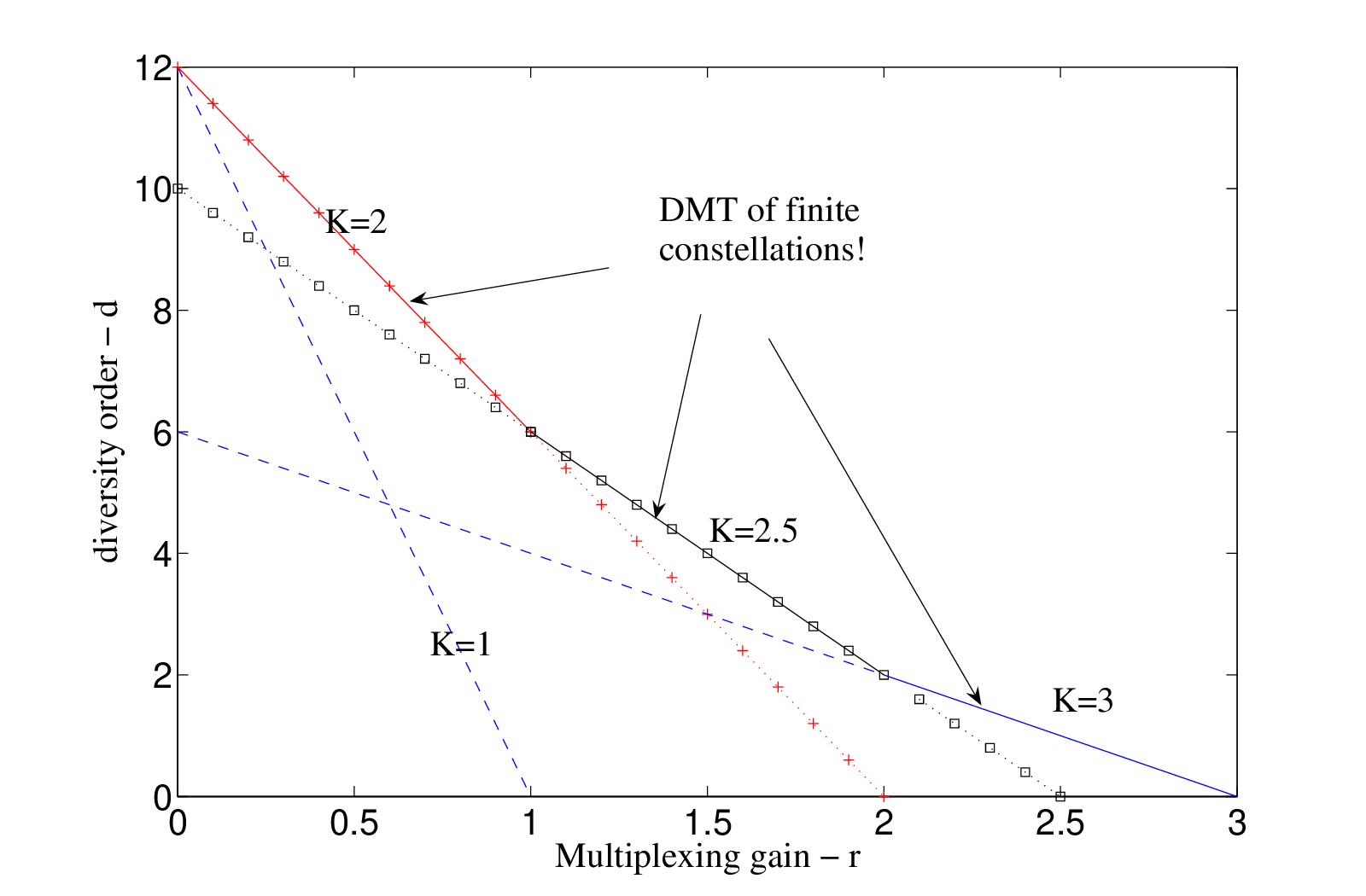,height=4.8cm} \caption{The
diversity order as a linear function of the multiplexing gain $r$
for $M=4$, $N=3$ and $K=1$, $2$, $2.5$ and $3$.} \label{fig:IC_DMT}
\end{figure}

\begin{cor}\label{Cor:EqualPointsUpperBoundDivOrder}
For $0 < K\le\frac{M\cdot N}{N+M-1}$ we get $d^{\ast}_{K}(0)=MN$.
For $\frac{(M-l+1)(N-l+1)}{N+M-1-2(l-1)}+l-1< K\le
\frac{(M-l)(N-l)}{N+M-1-2\cdot l}+l$, $l=1,\dots,L-1$ we get
$d^{\ast}_{K}(l)=(M-l)(N-l)$.
\end{cor}
\begin{proof}
The proof is straight forward from $d_{K}^{\ast}(r)$ properties.
\end{proof}

From Corollary \ref{Cor:EqualPointsUpperBoundDivOrder} we get that
the range of $K$ can be divided into segments, where for each
segment we have a set of straight lines, that are all equal at a
certain integer point. Note that at these points, we get the same
values as the optimal DMT for finite constellations.

\begin{cor}\label{Cor:MaximalDiverityOrder}
In the range $l\le r \le l+1$, the maximal possible diversity order
is achieved at dimension $K_{l}=\frac{(M-l)(N-l)}{N+M-1-2\cdot l}+l$
and equals
$$d^{\ast}_{K_{l}}(r)=(M-l)(N-l)\frac{K_{l}}{K_{l}-l}(1-\frac{r}{K_{l}})$$
$$=(M-l)(N-l)-(r-l)(N+M-2\cdot l-1)$$
where $l=0,\dots,L-1$. This expression equals to the optimal DMT of
finite constellations in this range.
\end{cor}
\begin{proof}
The proof is straight forward from $d_{K}^{\ast}(r)$ properties.
\end{proof}

From Corollary \ref{Cor:MaximalDiverityOrder} we can see that
$d_{K_{l}}^{\ast}(l)=(M-l)(N-l)$ and
$d_{K_{l}}^{\ast}(l+1)=(M-l-1)(N-l-1)$. We also know that
$d_{K_{l}}^{\ast}(r)$ is a straight line. Also, the optimal DMT for
finite constellations consists of a straight line in the range $l\le
r\le l+1$, that equals $(N-l)(M-l)$ when $r=l$ and $(M-l-1)(N-l-1)$
when $r=l+1$. Hence, in the range $l\le r\le l+1$ for
$K_{l}=\frac{(M-l)(N-l)}{N+M-1-2\cdot l}+l$, we get an upper bound
that equals to the optimal DMT of finite constellations presented in
\cite{TseDivMult2003}. Since for each $l=0,\dots, L-1$, we have such
$K_{l}$, the solution of
$$\max_{0\le K\le L}d_{K}^{\ast}(r)\quad 0\le r\le L$$
equals to the optimal DMT of finite constellations.

Figure \ref{fig:IC_DMT} illustrates the properties of
$d_{K}^{\ast}(r)$ following Corrolaries
\ref{Cor:EqualPointsUpperBoundDivOrder},
\ref{Cor:MaximalDiverityOrder}. We take the example of $M=4$, $N=3$.
For $0\le K\le 2$ we get upper bounds that have diversity order $12$
for $r=0$. We can see that in the range $0\le r\le 1$, the upper
bound of $K=2$ is maximal and equals to the optimal DMT of finite
constellations. In the range $2<K\le 2.5$ we can see that the upper
bounds have the same diversity order $6$ at $r=1$. In the range
$1\le r \le 2$, the upper bound of $K=2.5$ is maximal and equals to
the optimal DMT of finite constellations in this range. For
$2.5<K\le 3$, the upper bounds equal to $2$ at $r=2$. In the range
$2<r\le 3$, the upper bound of $K=3$ is maximal and again equals to
the optimal DMT of finite constellations in this range.

\begin{figure}[h!]
\centering \epsfig{figure=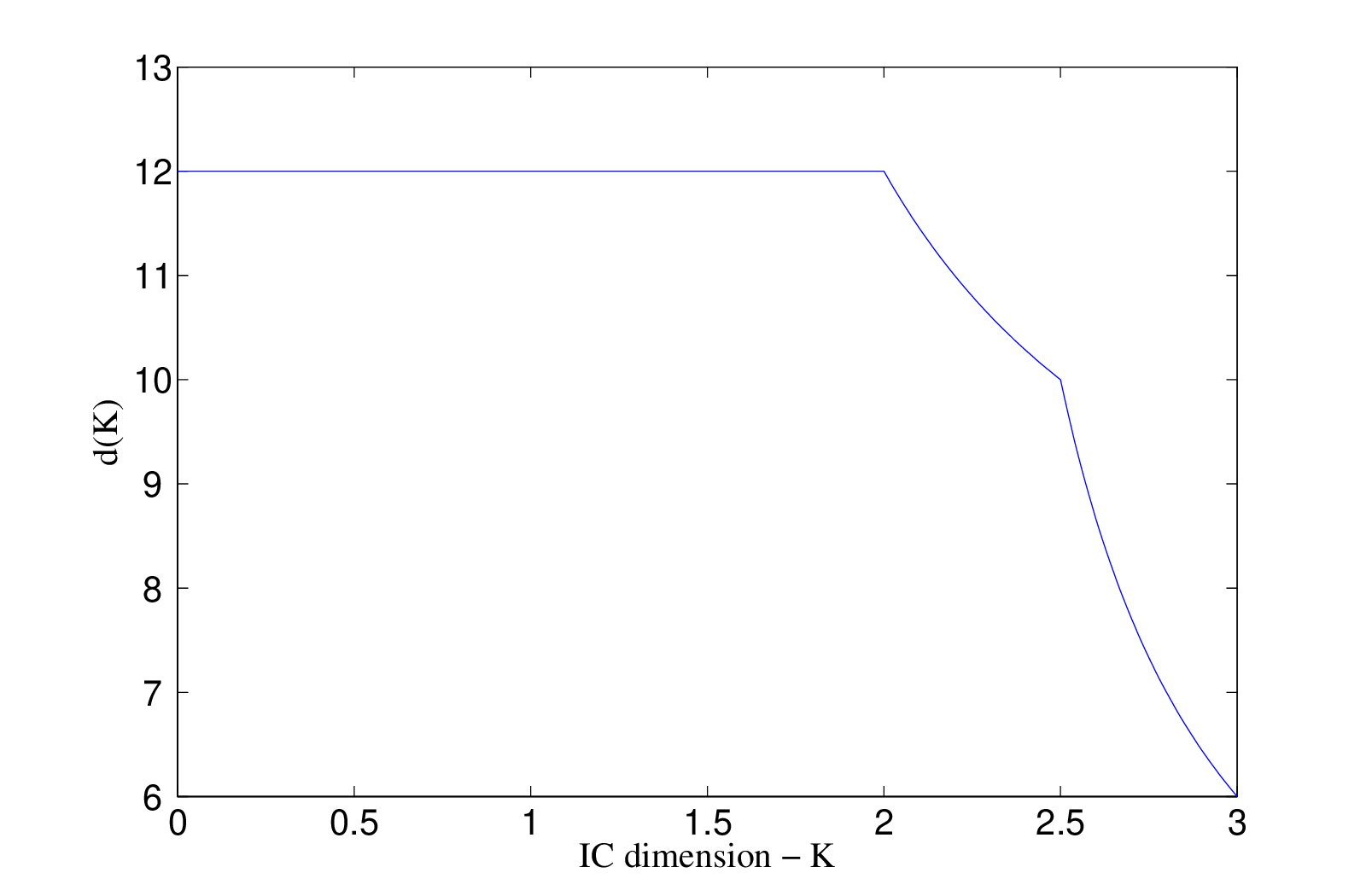,height=4.8cm}
\caption{$d_{K}^{\ast}(0)$ as a function of the IC dimensions per channel use $K$,
for $M=4$, $N=3$.} \label{fig:Diversity_MG_0}
\end{figure}

Figure \ref{fig:Diversity_MG_0} presents the maximal diversity order
that can be attained for different average number of dimensions per
channel use, for the case $M=4$ and $N=3$, i.e. the upper bound on
the diversity order for $r=0$, $d^{\ast}_{K}(0)$, where $0\le K\le
3$. In the range $0\le K\le 2$ we get $d^{\ast}_{K}(0)=12$. It
coincides with the result presented in Figure \ref{fig:IC_DMT},
where we showed that in this range the straight lines have the same
value for $r=0$. Hence, for IC's, one can use up to 2 average number
of dimensions per channel use without compromising the diversity
order. Starting from $K\ge 2$, the tradeoff starts to kick-in and
the maximal diversity order starts to reduce as we increase the
average number of dimensions per channel use. Also note that for
$K=3$ the diversity order is $6$ when $r=0$.

\section{Attaining the Best Diversity
Order}\label{sec:LowerBoundDiversityOrder} In this section we show
that the optimal DMT of finite constellations is achievable by a
sequence of IC's in general and lattices using regular lattice
decoding in particular. In subsection
\ref{subsec:TheTransmissionScheme} we present a transmission scheme
for any $M$ and $N$ that transmits an IC with
$K_{l}=\frac{(M-l)(N-l)}{N+M-1-2\cdot l}+l$ and $T_{l}=N+M-1-2\cdot
l$, $l=0,\dots,L-1$, where as previously defined $L=\min(M,N)$ and
$K_{l}$ is chosen based on the results in section
\ref{sec:LowerBoundErrorProb}. In subsection
\ref{subsec:TheEffectiveChannel} we present the effective channel
induced by this transmission scheme. Following that we extend the
methods presented in \cite{PoltirevJournal} and derive in Theorem
\ref{Th:UpperBoundErrorProb} for each channel realization an upper
bound on the average decoding error probability of ensemble of IC's.
By averaging the upper bound over the channel realizations, we show
in Theorem \ref{Th:LowerBoundDiversityOrder} that the proposed
transmission scheme attains the optimal DMT. In Theorem
\ref{Th:LowerBoundDiversityOrderLattices} we extend this result also
to lattices when employing regular lattice decoder. Finally, we
discuss power spreading technique over the transmit antennas for the
transmission scheme in subsection \ref{subsec:PeaktoAverageRatio},
and give some averaging arguments on the existence of sequence of
IC's that attain the optimal DMT in subsection
\ref{subsec:AveragingArguments}.

\subsection{The Transmission Scheme}\label{subsec:TheTransmissionScheme}
The transmission matrix $G_{l}$, $l=0,\dots,L-1$, has $M$ rows that
represent the transmission antennas, and $T_{l}=N+M-1-2\cdot l$
columns that represent the number of channel uses.

We begin by describing the transmission matrix structure in general for any $M$ and $N$.
\begin{enumerate}
\item For $N\ge M$ and
$K_{M-1}=\frac{M(N-M+1)}{N-M+1}=M$: the matrix $G_{M-1}$ has $N-M+1$
columns (channel uses). In the first column transmit symbols
$x_{1},\dots, x_{M}$ on the $M$ antennas, and in the $N-M+1$ column
transmit symbols $x_{M(N-M)+1},\dots, x_{M(N-M+1)}$ on the $M$
antennas.
\item For $M > N$ and
$K_{N-1}=\frac{N(M-N+1)}{M-N+1}=N$: the matrix $G_{N-1}$ has $M-N+1$
columns. In the first column transmit symbols $x_{1},\dots, x_{N}$
on antennas $1,\dots,N$ and in the $M-N+1$ column transmit symbols
$x_{N(M-N)+1},\dots, x_{N(M-N+1)}$ on antennas $M-N+1,\dots,M$.
\item For $K_{l}$, $l=0,\dots,L-2$: the matrix $G_{l}$ has $M+N-1-2\cdot l$ columns.
We add to $G_{l+1}$, the transmission scheme of $K_{l+1}$, two
columns in order to get $G_{l}$. In the first added column transmit
$l+1$ symbols on antennas $1,\dots,l+1$. In the second added column
transmit different $l+1$ symbols on antennas $M-l,\dots,M$.
\end{enumerate}

\emph{Example}: $M=4$, $N=3$. In this case the transmission
scheme for $K=3$, $2.5$ and $2$ ($G_{2}$, $G_{1}$ and $G_{0}$ respectively) is as follows:
\begin{equation}\label{eq:TheTransmissionSchemeExample3x4}
\underb{\underb{\underb{\left(\begin{array}{cc}
x_{1} & 0 \\
x_{2} & x_{4}\\
x_{3} & x_{5}\\
0   &   x_{6}
\end{array}\right.}_{K_{2}=\frac{6}{2}}
\left.\begin{array}{cc}
x_{7} & 0\\
x_{8} & 0\\
0 & x_{9}\\
0 & x_{10}
\end{array} \right.}_{K_{1}=\frac{10}{4}}
\left.\begin{array}{cc}
x_{11} & 0\\
0 & 0\\
0 & 0\\
0 & x_{12}
\end{array} \right)}_{K_{0}=\frac{12}{6}}.\end{equation}

\subsection{The Effective Channel}\label{subsec:TheEffectiveChannel}
Next we define the effective channel matrix induced by the
transmission scheme. In accordance with the channel model from
\eqref{eq:Channel Fading}, the multiplication $H\cdot G_{l}$ yields
a matrix with $N$ rows and $T_{l}$ columns, where each column equals
to $H\cdot\udl{x}_{t}$, $t=1\dots T_{l}$, as in \eqref{eq:Channel
Fading}. We are interested in transmitting $K_{l}T_{l}$-complex
dimensional IC with $K_{l}T_{l}$ complex symbols. Hence, in the
proposed transmission scheme, $G_{l}$ has exactly $K_{l}T_{l}$
non-zero complex entries that represent the $K_{l}T_{l}$-complex
dimensional IC within $\mathbb{C}^{MT_{l}}$. For each column of
$G_{l}$, denoted by $\udl{g}_{i}$, $i=1\dots T_{l}$, we define the
effective channel that $\udl{g}_{i}$ sees as $\widehat{H}_{i}$. It
consists of the columns of $H$ that correspond to the non-zero
entries of $\udl{g}_{i}$, i.e.
$H\cdot\udl{g}_{i}=\widehat{H}_{i}\cdot\udl{\widehat{g}}_{i}$, where
$\udl{\widehat{g}}_{i}$ equals the non-zero entries of
$\udl{g}_{i}$. As an example assume without loss of generality that
the first $l_{i}$ entries of $\udl{g}_{i}$ are not zero. In this
case $\widehat{H}_{i}$ is an $N\times l_{i}$ matrix equals to the
first $l_{i}$ columns of $H$. In accordance with
\eqref{eq:ExtendedChannelModel}, $H_{\eff}^{(l)}$ is an
$NT_{l}\times K_{l}T_{l}$ block diagonal matrix consisting of
$T_{l}$ blocks. Each block corresponds to the multiplication of $H$
with different column of $G_{l}$, i.e. $\widehat{H}_{i}$ is the
$i'th$ block of $H_{\eff}^{(l)}$. Note that in the effective matrix
$NT_{l}\ge K_{l}T_{l}$.

We would like to elaborate on the structure of the blocks of
$H_{\eff}^{(l)}$. For this reason we denote the columns of $H$ as
$\udl{h}_{i}$, $i=1,\dots, M$.
\begin{enumerate}
\item The case where $N\ge M$. For this case the transmission scheme has $N+M-1-2\cdot l$ columns. The first $N-M+1$ columns of $G_{l}$, $\udl{g}_{1},\dots,\udl{g}_{N-M+1}$, contain $M\cdot (N-M+1)$ different complex symbols, i.e. there are no zero entries in these columns. Hence, in this case the first $N-M+1$ blocks of $H_{\eff}^{(l)}$ are
\begin{equation}\label{eq:HeffBlocksForTheCaseNgeM1}
\widehat{H}_{i}=H\qquad i=1,\cdots, N-M+1.
\end{equation}
After the first $N-M+1$ columns we have $M-1-l$ pairs of columns. For each pair we have
\begin{equation}\label{eq:HeffBlocksForTheCaseNgeM2}
\widehat{H}_{N-M+2k}=\{\udl{h}_{1},\dots,\udl{h}_{M-k}\}
\end{equation}
and
\begin{equation}\label{eq:HeffBlocksForTheCaseNgeM3}
\widehat{H}_{N-M+2k+1}=\{\udl{h}_{k+1},\dots,\udl{h}_{M}\}
\end{equation}
where $k=1,\dots, M-1-l$.
\item The case where $M>N$. Again the transmission scheme has $N+M-1-2\cdot l$ columns. By the definition of the first $M-N+1$ columns of $G_{l}$, we get that
\begin{equation}\label{eq:HeffBlocksForTheCaseMgN1}
\widehat{H}_{i}=\{\udl{h}_{i},\dots, \udl{h}_{N+i-1}\}\qquad
i=1,\cdots, M-N+1.
\end{equation}
We have additional $N-1-l$ pairs of columns in $G_{l}$. For each of these pairs we get
\begin{equation}\label{eq:HeffBlocksForTheCaseMgN2}
\widehat{H}_{M-N+2k}=\{\udl{h}_{1},\dots,\udl{h}_{N-k}\}
\end{equation}
and
\begin{equation}\label{eq:HeffBlocksForTheCaseMgN3}
\widehat{H}_{M-N+2k+1}=\{\udl{h}_{M-N+k+1},\dots,\udl{h}_{M}\}
\end{equation}
where $k=1,\dots, N-1-l$.
\end{enumerate}
\emph{Example}: consider $M=4$, $N=3$ as presented in
\eqref{eq:TheTransmissionSchemeExample3x4}. In this case $l=0,1,2$
and we have $K_{2}=3$, $K_{1}=2.5$ and $K_{0}=2$ respectively.
\begin{enumerate}
\item $K_{2}=3$: $H_{\eff}^{(2)}$ is generated from the
multiplication of the $3\times 4$ matrix $H$ with the first two
columns of the transmission matrix. In this case $H_{\eff}^{(2)}$ is
a $6\times 6$ block diagonal matrix, consisting of two blocks. Each
block is a $3\times 3$ matrix. We get that
$\widehat{H}_{1}=\{\udl{h}_{1},\udl{h}_{2},\udl{h}_{3}\}$ and
$\widehat{H}_{2}=\{\udl{h}_{2},\udl{h}_{3},\udl{h}_{4}\}$.
\item $K_{1}=\frac{10}{4}=2.5$: $H_{\eff}^{(1)}$ is a $12\times 10$ block
diagonal matrix consisting of 4 blocks. The first two blocks are
identical to the blocks of $H_{\eff}^{(2)}$. The additional two
blocks (multiplication with columns 3-4) are $3\times 2$ matrices.
We get that $\widehat{H}_{3}=\{\udl{h}_{1},\udl{h}_{2}\}$ and
$\widehat{H}_{4}=\{\udl{h}_{3},\udl{h}_{4}\}$.
\item $K_{0}=2$: $H_{\eff}^{(0)}$ consists of six blocks. In
this case the last two blocks are $3\times 1$ vectors. We get that
$\widehat{H}_{5}=\udl{h}_{1}$ and $\widehat{H}_{6}=\udl{h}_{4}$.
\end{enumerate}
We present $H_{\eff}^{(0)}$ of our example in equation
\eqref{eq:Heff0TransmissionMatrix}. Note that
$\udl{h}_{i}\in\mathbb{C}^{3}$ for $1\le i\le 4$, and $\udl{0}$ is a
$3\times 1$ vector.

\begin{figure*}
\begin{equation}\label{eq:Heff0TransmissionMatrix}
H_{\eff}^{(0)}=\left(\begin{array}{cccccc}
\udl{h}_{1} & \udl{h}_{2} &\udl{h}_{3}  &\udl{0} &\udl{0} &\udl{0}\\
\udl{0} &\udl{0} &\udl{0} &\udl{h}_{2} & \udl{h}_{3} &\udl{h}_{4}\\
\udl{0} &\udl{0} &\udl{0} &\udl{0} &\udl{0} &\udl{0}\\
\udl{0} &\udl{0} &\udl{0} &\udl{0} &\udl{0} &\udl{0}\\
\udl{0} &\udl{0} &\udl{0} &\udl{0} &\udl{0} &\udl{0}\\
\udl{0} &\udl{0} &\udl{0} &\udl{0} &\udl{0} &\udl{0}
\end{array}\right.
\left.\begin{array}{cccc}
\udl{0} &\udl{0} &\udl{0} &\udl{0}\\
\udl{0} &\udl{0} &\udl{0} &\udl{0}\\
\udl{h}_{1} &\udl{h}_{2} &\udl{0} &\udl{0}\\
\udl{0} &\udl{0} &\udl{h}_{3} &\udl{h}_{4}\\
\udl{0} &\udl{0} &\udl{0} &\udl{0}\\
\udl{0} &\udl{0} &\udl{0} &\udl{0}
\end{array} \right.
\left.\begin{array}{cc}
\udl{0} &\udl{0}\\
\udl{0} &\udl{0}\\
\udl{0} &\udl{0}\\
\udl{0} &\udl{0}\\
\udl{h}_{1} &\udl{0}\\
\udl{0} &\udl{h}_{4}
\end{array} \right)
\end{equation}
\end{figure*}

From the sequential construction of the blocks of $H_{\eff}^{(l)}$
\eqref{eq:HeffBlocksForTheCaseNgeM1}-\eqref{eq:HeffBlocksForTheCaseNgeM3},
\eqref{eq:HeffBlocksForTheCaseMgN1}-\eqref{eq:HeffBlocksForTheCaseMgN3}
it is easy to see that when two columns of $H$ occur in a certain
block of $H_{\eff}^{(l)}$, the columns of $H$ between them must also
occur in the same block, i.e. if $\udl{h}_{1}$, $\udl{h}_{5}$ occur
in a certain block, then $\udl{h}_{2},\udl{h}_{3},\udl{h}_{4}$ also
occur in the same block. Next we prove a property of the
transmission scheme $G_{l}$, that relates to the number of
occurrences of the columns of $H$ in the blocks of $H_{\eff}^{(l)}$.
For each set of columns in $H$, we give an upper bound on the amount
of its appearances in different blocks.

\begin{lem}\label{Lem:LemmaOfTheOccurrencesOfEachColumn}
Consider the transmission scheme $G_{l}$, $l=0,\dots L-1$. In case
$0\le i-j<L$, the columns $\udl{h}_{j},\dots,\udl{h}_{i}$ may occur
together in at most $N-i+j$ blocks of $H_{\eff}^{(l)}$. In case
$i-j\ge L$ they can not occur together in any block of
$H_{\eff}^{(l)}$.
\end{lem}
\begin{proof}
See appendix \ref{Append:NumberEachColumnOccurrences}.
\end{proof}

\subsection{Upper Bound on The Error Probability}

Next we would like to derive an upper bound on the average decoding
error probability of ensemble of $K_{l}T_{l}$-complex dimensional
IC, for each channel realization. We define
$|H_{\eff}^{(l)\dagger}H_{\eff}^{(l)}|=\rho^{-\sum_{i=1}^{K_{l}T_{l}}\eta_{i}}$
, where $\rho^{-\frac{\eta_{i}}{2}}$ is the $i'th$ singular value of
$H_{\eff}^{(l)}$, $1\le i\le K_{l}T_{l}$. We also define
$\udl{\eta}=(\eta_{1},\dots,\eta_{K_{l}T_{l}})^{T}$. Note that
$NT_{l}\ge K_{l}T_{l}$.

\begin{theorem}\label{Th:UpperBoundErrorProb}
There exists a sequence of $K_{l}T_{l}$-complex dimensional IC's,
with channel realization $H_{\eff}^{(l)}$ and a receiver VNR
$\mu_{rc}=\rho^{1-\frac{r}{K_{l}}-\frac{\sum_{i=1}^{K_{l}T_{l}}{\eta_{i}}}{K_{l}T_{l}}}$,
that has an average decoding error probability
\ifthenelse{\equal{\singlecolumntype}{1}}
{$$\ol{P_{e}}(H_{\eff}^{(l)},\rho)=\ol{P_{e}}(\udl{\eta},\rho)\le
D(K_{l}T_{l})\rho^{-T_{l}(K_{l}-r)+\sum_{i=1}^{K_{l}T_{l}}\eta_{i}}
=D(K_{l}T_{l})\rho^{-T_{l}(K_{l}-r)}\cdot|H_{\eff}^{(l)\dagger}H_{\eff}^{(l)}|^{-1}$$}
{\begin{align}
\nonumber\ol{P_{e}}(H_{\eff}^{(l)},\rho)=\ol{P_{e}}(\udl{\eta},\rho)\le
D(K_{l}T_{l})\rho^{-T_{l}(K_{l}-r)+\sum_{i=1}^{K_{l}T_{l}}\eta_{i}}\nonumber\\
=D(K_{l}T_{l})\rho^{-T_{l}(K_{l}-r)}\cdot|H_{\eff}^{(l)\dagger}H_{\eff}^{(l)}|^{-1}\nonumber
\end{align}}
where $D(K_{l}T_{l})$ is a constant independent of $\rho$, and
$\eta_{i}\ge 0$ for every $1\le i\le K_{l}T_{l}$.
\end{theorem}
\begin{proof}
We base our proof on the techniques developed by Poltyrev
\cite{PoltirevJournal} for the AWGN channel. However, the channel
considered here is colored. In spite of that, we show that what
affects the average decoding error probability is the singular
values product, which is encapsulated by the receiver VNR,
$\mu_{rc}$. This observation enables us to facilitate this colored
channel analysis. The full proof in appendix
\ref{Append:UpperBoundErrorProb}.
\end{proof}
By averaging arguments we know that there exists a sequence of IC's that
satisfies these requirements.

\subsection{Achieving the Optimal DMT}
In this subsection we calculate the DMT of the proposed transmission
scheme. We upper bound the determinant of the effective channel
inverse, $|H_{\eff}^{(l)\dagger}H_{\eff}^{(l)}|^{-1}$, based on the
effective channel properties presented in subsection
\ref{subsec:TheEffectiveChannel}. In Theorem
\ref{Th:UpperBoundErrorProb} we showed that the upper bound on the
error probability depends on this determinant. Hence, the upper
bound on the determinant gives us a new upper bound on the average
decoding error probability. We average the new upper bound over all
channel realizations and get the DMT of the transmission scheme.

The channel matrix $H$ consists of $N\cdot M$ i.i.d entries, where
each entry has distribution $h_{i,j}\sim CN(0,1)$. Without loss of
generality we consider the case where the columns of $H$ are drawn
sequentially from left to right, i.e. $\udl{h}_{1}$ is drawn first,
then $\udl{h}_{2}$ is drawn et cetera. Column $\udl{h}_{j}$ is an
$N$-dimensional vector. Given
$\udl{h}_{\max(1,j-N+1)},\dots,\udl{h}_{j-1}$, we can write
$$\udl{h}_{j}=\Theta(\udl{h}_{\max(1,j-N+1)},\dots,\udl{h}_{j-1})\cdot\udl{\widetilde{h}}_{j}$$
where $\Theta(\cdot)$ is an $N\times N$ unitary matrix.
$\Theta(\cdot)$ is chosen such that:
\begin{enumerate}
\item The first element of $\udl{\widetilde{h}}_{j}$, $\widetilde{h}_{1,j}$,  is
in the direction of $\udl{h}_{j-1}$.
\item The second element, $\widetilde{h}_{2,j}$, is in the direction orthogonal to
$\udl{h}_{j-1}$, in the hyperplane spanned by
$\{\udl{h}_{j-1},\udl{h}_{j-2}\}$.
\item Element $\widetilde{h}_{\min(j,N)-1,j}$ is in the direction
orthogonal to the hyperplane spanned by
$\{\udl{h}_{\max(2,j-N+2)},\dots,\udl{h}_{j-1}\}$ inside the
hyperplane spanned by
$\{\udl{h}_{\max(1,j-N+1)},\dots,\udl{h}_{j-1}\}$.
\item The rest of the $N-\min(j,N)+1$
elements are in directions orthogonal to the hyperplane
$\{\udl{h}_{\max(1,j-N+1)},\dots,\udl{h}_{j-1}\}$.
\end{enumerate}
Note that $\widetilde{h}_{i,j}$, $1\le i\le N$, $1\le j\le M$ are
i.i.d random variables with distribution $CN(0,1)$. Let us denote by
$\udl{h}_{j\perp j-1,\dots,j-k}$ the component of $\udl{h}_{j}$
which resides in the $N-k$ subspace which is perpendicular to the
space spanned by $\{\udl{h}_{j-1},\dots,\udl{h}_{j-k}\}$. In this
case we get
\begin{equation}\label{eq:TheValueOfOrthogonalValuesToHyperplane}
\lv\udl{h}_{j\perp
j-1,\dots,j-k}\rv^{2}=\sum_{i=k+1}^{N}|\widetilde{h}_{i,j}|^2\quad
1\le k\le\min(j,N)-1.
\end{equation}

If we assign $|\widetilde{h}_{i,j}|^{2}=\rho^{-\xi_{i,j}}$, we get
that the probability density function (PDF) of $\xi_{i,f}$ is
\begin{equation}\label{eq:TheXiOriginalPDF}
f(\xi_{i,j})=C\cdot\log{\rho}\cdot\rho^{-\xi_{i,j}}\cdot
e^{-\rho^{-\xi_{i,j}}}
\end{equation}
where $C$ is a normalization factor. In our analysis we assume a
very large value for $\rho$. Hence we can neglect events where
$\xi_{i,j}<0$ since in this case the PDF \eqref{eq:TheXiOriginalPDF}
decreases exponentially as a function of $\rho$. For a very large
$\rho$, $\xi_{i,j}\ge 0$, $1\le i\le N$ and $1\le j\le M$, the PDF
takes the following form
\begin{equation}\label{eq:ThePDFOfMagnitudeExponent}
f(\xi_{i,j})\propto \rho^{-\xi_{i,j}} \qquad \xi_{i,j}\ge 0.
\end{equation}
In this case by assigning in
\eqref{eq:TheValueOfOrthogonalValuesToHyperplane} the vector
$\udl{\xi}_{j}=(\xi_{1,j},\dots,\xi_{N,j})^{T}$, whose PDF is
proportional to $\rho^{-\sum_{i=1}^{N}\xi_{i,j}}$, we get
\begin{equation}\label{eq:TheContributionOfhjInTheGeneralCase1}
\lv\udl{h}_{j\perp
j-1,\dots,j-k}\rv^{2}\dot{=}\rho^{-\min_{s\in\{k+1,\dots,N\}}\xi_{s,j}}=\rho^{-a(k,\udl{\xi}_{j})}
\end{equation}
where $1\le k\le \min(j,L)-1$ and
$a(k,\udl{\xi}_{j})=\min_{s\in\{k+1,\dots,N\}}\xi_{s,j}$. In
addition
\begin{equation}\label{eq:TheContributionOfhjInTheGeneralCase2}
\lv\udl{h}_{j}\rv^{2}\dot{=}\rho^{-\min_{s\in\{1,\dots,N\}}\xi_{s,j}}=\rho^{-a(0,\udl{\xi}_{j})}.
\end{equation}
Note that
\begin{equation}\label{eq:InequalityOfTheNumberOfContributions}
a(\min(j,L)-1,\udl{\xi}_{j})\ge\dots\ge a(0,\udl{\xi}_{j})\ge 0.
\end{equation}

Next we wish to quantify the contribution of a certain column in the
channel matrix, $\udl{h}_{j}$, to the determinant
$|H_{\eff}^{(l)\dagger}H_{\eff}^{(l)}|$. $H_{\eff}^{(l)}$ is a block
diagonal matrix. Hence the determinant of
$|H_{\eff}^{(l)\dagger}H_{\eff}^{(l)}|$ can be expressed as
\begin{equation}
|H_{\eff}^{(l)\dagger}H_{\eff}^{(l)}|=\prod_{i=1}^{T_{l}}|\widehat{H}_{i}^{\dagger}\widehat{H}_{i}|.
\end{equation}
Assume
$\widehat{H}_{i}=(\udl{\widehat{h}}_{1},\dots,\udl{\widehat{h}}_{m})$,
i.e. $\widehat{H}_{i}$ has $m$ columns. In this case we can state
that the determinant
$$|\widehat{H}_{i}^{\dagger}\widehat{H}_{i}|=\lv\udl{\widehat{h}}_{1}\rv^{2}\lv\udl{\widehat{h}}_{2\perp
1}\rv^{2}\dots\lv\udl{\widehat{h}}_{m\perp m-1,\dots,1}\rv^{2}.$$
Note that $\widehat{H}_{i}$ also has more rows than columns. The
columns of $\widehat{H}_{i}$ are subset of the columns of the
channel matrix $H$. Hence we are interested in the blocks where
$\udl{h}_{j}$ occurs. We know that the contribution of $\udl{h}_{j}$
to those determinants can be quantified by taking into account the
columns to its left in each block. We consider two cases:
\begin{itemize}
\item The case $N\ge M$. In this case we can see
from
\eqref{eq:HeffBlocksForTheCaseNgeM1}-\eqref{eq:HeffBlocksForTheCaseNgeM3}
that $\udl{h}_{j}$ may occur with
$\{\udl{h}_{1},\dots,\udl{h}_{j-1}\}$ to its left in different blocks.
\item The
case $M>N$. In this case we can see from
\eqref{eq:HeffBlocksForTheCaseMgN1}-\eqref{eq:HeffBlocksForTheCaseMgN3}
that $\udl{h}_{j}$ may occur only with
$\{\udl{h}_{\max(1,j-N+1)},\dots,\udl{h}_{j-1}\}$ to its left in different
blocks.
\end{itemize}

Based on \eqref{eq:TheContributionOfhjInTheGeneralCase1} and
\eqref{eq:TheContributionOfhjInTheGeneralCase2} we can quantify the
contribution of $\udl{h}_{j}$ to
$|H_{\eff}^{(l)\dagger}H_{\eff}^{(l)}|$ by
\ifthenelse{\equal{\singlecolumntype}{1}}
{\begin{equation}\label{eq:TheContributionOfhjIntermsofbetta}
\lv\udl{h}_{j}\rv^{2b_{j}(0)}\prod_{k=1}^{\min(j,L)-1}\lv\udl{h}_{j\perp j-1,\dots,j-k}\rv^{2b_{j}(k)}\dot{=}
\rho^{-\sum_{k=0}^{\min(j,L)-1}b_{j}(k)a(k,\udl{\xi}_{j})}
\end{equation}}
{\begin{align}\label{eq:TheContributionOfhjIntermsofbetta}
\lv\udl{h}_{j}\rv^{2b_{j}(0)}\prod_{k=1}^{\min(j,L)-1}\lv\udl{h}_{j\perp j-1,\dots,j-k}\rv^{2b_{j}(k)}\dot{=}\nonumber\\
\rho^{-\sum_{k=0}^{\min(j,L)-1}b_{j}(k)a(k,\udl{\xi}_{j})}
\end{align}}
where $b_{j}(k)$ is the number of occurrences of $\udl{h}_{j}$ in
the blocks of $H_{\eff}^{(l)}$, with only
$\{\udl{h}_{j-1},\dots,\udl{h}_{j-k}\}$ to its left. $b_{j}(0)$ is
the number of occurrences of $\udl{h}_{j}$ with no columns to its
left. Note that from the definition of the transmission scheme we
get that for $l=0$, $b_{j}(k)>0$ for $1\le k\le\min(j,L)-1$.

In the following theorem we calculate the DMT of the proposed
transmission scheme.

\begin{theorem}\label{Th:LowerBoundDiversityOrder}
There exists a sequence of $K_{l}T_{l}$-complex dimensional IC's
with transmitter density $\gamma_{tr}=\rho^{rT_{l}}$ and $T_{l}$
channel uses that has diversity order
$$d_{K_{l}T_{l}}(r)\ge (M-l)(N-l)-(r-l)(N+M-2\cdot l-1)$$
where $0\le r\le K_{l}$ and $l=0,\dots,L-1$. In the range $l\le r\le
l+1$ this lower bound coincides with the optimal DMT of finite
constellations.
\end{theorem}
\begin{proof}
The proof outline is as follows. The upper bound on the error
probability from Theorem \ref{Th:UpperBoundErrorProb} depends on
$|H_{\eff}^{(l)\dagger}H_{\eff}^{(l)}|^{-1}$. We upper bound this
determinant value and average over different realizations of
$H_{\eff}^{(l)}$ in order to find the diversity order of the
transmission matrix $G_{l}$. We begin by lower bounding
$|H_{\eff}^{(l)\dagger}H_{\eff}^{(l)}|$. Based on the sequential
structure of $G_{l}$, we lower bound the contribution of a certain
column of $H$, $\udl{h}_{j}$, $1\le j\le M$ to the determinant. This
gives us a new upper bound on the error probability for each channel
realization. We average the new upper bound on the error
probability, by averaging over $\udl{\widetilde{h}}_{1},\dots,
\udl{\widetilde{h}}_{M}$. From this averaging we get the required
DMT. The full proof is in appendix
\ref{Append:LowerBoundDiversityOrder}
\end{proof}

The diversity order attained in Theorem
\ref{Th:LowerBoundDiversityOrder} for $K_{l}$, $T_{l}$ coincides
with the optimal DMT of finite constellations in the range $l\le
r\le l+1$. Hence, by considering $0\le l\le L-1$, we can attain the
optimal DMT with $L$ sequences of IC's.

We present as an illustrative example the case of $M=N=2$. Let us
consider the case where $l=0$. In this case $K_{0}=\frac{4}{3}$, and
$T_{0}=3$, i.e. we transmit $4$-complex dimensional IC. The
transmission scheme diversity order in this case is $4-3r$, $0\le
r\le \frac{4}{3}$. In this case the effective channel matrix,
$H_{\eff}^{(0)}$, consists of three blocks:
$\widehat{H}_{1}=(\udl{h}_{1},\udl{h}_{2})$,
$\widehat{H}_{2}=\udl{h}_{1}$ and $\widehat{H}_{3}=\udl{h}_{2}$.
According to our definitions
$$|\widehat{H}_{1}^{\dagger}\widehat{H}_{1}|=\lv\udl{h}_{1}\rv^{2}\cdot\lv\udl{h}_{2\perp
1}\rv^{2}=\rho^{-\min(\xi_{1,1},\xi_{2,1})}\cdot\rho^{-\xi_{2,2}}$$
and also $\lv\udl{h}_{1}\rv^{2}=\rho^{-\min(\xi_{1,1},\xi_{2,1})}$,
$\lv\udl{h}_{2}\rv^{2}=\rho^{-\min(\xi_{1,2},\xi_{2,2})}$. In
accordance with  \eqref{eq:UpperBoundErrorProbDueExact2Ranges} we
divide the integral into two terms. In the first term we solve the
optimization problem \ifthenelse{\equal{\singlecolumntype}{1}}
{\begin{equation} \min_{\xi_{\udl{i},\udl{j}}\in
\mathcal{A}}(4-3r)-(\xi_{2,2}+2\cdot\min{\big(\xi_{1,1},\xi_{2,1})}
+\min{(\xi_{1,2},\xi_{2,2}})\big)+\sum_{i=1}^{2}\sum_{j=1}^{2}\xi_{i,j}.
\end{equation}}
{\begin{align}
\min_{\xi_{\udl{i},\udl{j}}\in \mathcal{A}}(4-3r)-(\xi_{2,2}+2\cdot\min{\big(\xi_{1,1},\xi_{2,1})}\nonumber\\
+\min{(\xi_{1,2},\xi_{2,2}})\big)+\sum_{i=1}^{2}\sum_{j=1}^{2}\xi_{i,j}.\nonumber
\end{align}}
One solution to this problem is $\xi_{i,j}=0$ for $1\le i\le 2$, $1\le j\le 2$. In this case we get an exponential term that equals $4-3r$. For the second integral we solve the optimization problem
$$\min_{\xi_{\udl{i},\udl{j}}\in \ol{\mathcal{A}}}\sum_{i=1}^{2}\sum_{j=1}^{2}\xi_{i,j}.$$
In this case the optimization problem solution is
$\sum_{i=1}^{2}\sum_{j=1}^{2}\xi_{i,j}=4-3r$. Hence, all together,
we get a diversity order that equals $4-3r$, that coincides with the
optimal DMT of finite constellations in the range $0\le r\le 1$.

In the next theorem we prove the existence of a sequence of lattices that
has the same lower bound as in Theorem
\ref{Th:LowerBoundDiversityOrder}.

\begin{theorem}\label{Th:LowerBoundDiversityOrderLattices}
There exists a sequence of $2K_{l}T_{l}$-real dimensional lattices
with transmitter density $\gamma_{tr}=\rho^{rT_{l}}$ and $T_{l}$
channel uses, that attains a diversity order
$$d_{K_{l}T_{l}}(r)\ge (M-l)(N-l)-(r-l)(N+M-2\cdot l-1)$$
where $0\le r\le K_{l}$ and $l=0,\dots,L-1$.
\end{theorem}
\begin{proof}
See appendix \ref{Append:LatticesDiversityOrder}
\end{proof}
Note that we considered a $2K_{l}T_{l}$-real dimensional lattice,
where the lattice first $K_{l}T_{l}$ dimensions are spread over the
real part of the non-zero entries of $G_{l}$, and the other
$K_{l}T_{l}$ dimensions of the lattice are spread on the imaginary
part of the non-zero entries of $G_{l}$. This does not necessarily
yields a $K_{l}T_{l}$-complex dimensional lattice in the transission
scheme. Considering the $2K_{l}T_{l}$-real dimensional lattice
enables us to use the \emph{Minkowski-Hlawaka-Siegel} Theorem
\cite{PoltirevJournal},\cite{LekkerkerkerGeomety}, and prove Theorem
\ref{Th:LowerBoundDiversityOrderLattices}.

\subsection{Power Spreading}\label{subsec:PeaktoAverageRatio}
For practical reasons, such as power peak to average ratio, one may
prefer to have a transmission scheme that spreads the transmitted
power equally over time and space. The transmitting matrix $G_{l}$
contains exactly $K_{l}T_{l}$ non-zero entries, where the rest of
the entries are zero. In order to spread the power more equally over
time and space we use the following unitary operations
$$U_{L} G_{l}U_{R}.$$
$U_{L}$ is an $M\times M$ unitary matrix that spreads each column of
$G_{l}$, i.e. spreads over space. $U_{R}$ is a $T_{l}\times T_{l}$
unitary matrix that spreads each raw of $G_{l}$, i.e. spreads over
time. As the distribution of $H$ and $H\cdot U_{L}$ are identical,
multiplying $U_{L}$ with $G_{l}$ gives exactly the same performance.
Based on the notations from $\eqref{eq:Channel Fading}$ we can state
that
$$G_{l}\cdot U_{R}=\big(\udl{x}_{1},\dots,\udl{x}_{T_{l}}\big)$$
where $\big(\udl{x}_{1},\dots,\udl{x}_{T_{l}}\big)$ are the channel inputs.
In the receiver we can state that the received signals are $\big(\udl{y}_{1},\dots,\udl{y}_{T_{l}}\big)$. By multiplying with $U_{R}^{\dagger}$ we get
$$\big(\udl{y}_{1},\dots,\udl{y}_{T_{l}}\big)\cdot U_{R}^{\dagger}=G_{l}+\big(\udl{n}_{1},\dots,\udl{n}_{T_{l}}\big)U_{R}^{\dagger}.$$
The distribution of $\big(\udl{n}_{1},\dots,\udl{n}_{T_{l}}\big)$ is
identical to the distribution of
$\big(\udl{n}_{1},\dots,\udl{n}_{T_{l}}\big)U_{R}^{\dagger}$. Hence,
multiplying $G_{l}$ with $U_{R}$ gives also exactly the same
performance. For instance, in order to achieve full diversity and
spread the power more uniformly, we take $G_{0}$ and duplicate its
structure $s$ times to create the transmission scheme $G_{0}^{(s)}$.
In this case the transmission matrix $G_{0}^{(s)}$ consists of
$sK_{0}T_{0}$ complex non-zero entries, i.e we transmit an
$sK_{0}T_{0}$ complex dimensional IC within the $sMT_{0}$ complex
space. $G_{0}^{(s)}$ is an $M\times sT_{0}$ dimensional matrix, that
has exactly the same diversity order as $G_{0}$ (it duplicates the
structure of $G_{0}$ $s$ times). Each row of $G_{0}^{(s)}$ has
exactly $sN$ non-zero entries. We define $U_{R}^{(s)}$ as
$sT_{0}\times sT_{0}$ unitary matrix. For large enough $s$, the
multiplication $G_{0}^{(s)}\cdot U_{R}^{(s)}$ spreads the power more
uniformly over space and time, and still achieves full diversity.
\footnote{It can be shown that replacing $U_{L}$ and $U_{R}$ with
any other two invertible matrices still yields transmission scheme
that attains the optimal DMT. It extends the set of subspaces in
$\mathbb{C}^{MT}$ that attain the optimal DMT. It also alludes that
alongside the proposed transmission matrix
\ref{subsec:TheTransmissionScheme}, there are many other options to
attain the optimal DMT.}

\subsection{Averaging Arguments}\label{subsec:AveragingArguments}
In this subsection we show that there exist $L$ sequences of
lattices that attain the optimal DMT, where each sequence of the $L$
sequences attains a different segment on the optimal DMT curve. In
addition we show that there exists a single IC that attains the
optimal DMT by diluting its points and adapting its dimensionality.

As a consequence of Theorem \ref{Th:UpperBoundErrorProb} and Theorem
\ref{Th:LowerBoundDiversityOrder} we can state the following
\begin{cor}\label{Cor:SequenceICAttainsTheEntireLine}
 Consider a sequence of $KT$-complex dimensional IC's $S_{KT}(\rho)$ with density $\gamma_{tr}=1$, that attains diversity order $d$. This sequence of IC's  also attains diversity order $d(1-\frac{r}{K})$ when the sequence density is scaled to $\gamma_{tr}=\rho^{rT}$.
\end{cor}
\begin{proof}
The proof is in appendix
\ref{append:SequenceICAttainsTheEntireLine}.
\end{proof}

\begin{cor}
The optimal DMT is attained by exactly $L$ sequences of
 $2K_{l}T_{l}$-real dimensional lattices, $l=0,\dots, L-1$, where
each sequence attains different segment of the optimal DMT.
\end{cor}
\begin{proof}
From Theorem \ref{Th:LowerBoundDiversityOrderLattices} we know that
there exists a $2K_{l}T_{l}$-real dimensional sequence of lattices
with density $\gamma_{tr}=1$ that attains diversity
$(M-l)(N-l)+l(N+M-2\cdot l-1)$. Hence, based on Corollary
\ref{Cor:SequenceICAttainsTheEntireLine} we can scale this
$2K_{l}T_{l}$-real dimensional sequence of lattices into a sequence
of lattices with density $\gamma_{tr}=\rho^{rT_{l}}$, and a
diversity order $(M-l)(N-l)-(r-l)(N+M-2\cdot l-1)$, i.e. the
sequence of lattices attains the optimal DMT line in the range $l\le
r\le l+1$. The optimal DMT is the maximal value of the $L$ lines,
for each $0\le r\le L$. Hence, there exist $L$ sequences of lattices
that attain the optimal DMT.
\end{proof}

Next, we show that there exists a single sequence of IC's that
attains the optimal DMT. The optimal DMT consists of $L$ segments of
straight lines. Each segment is attained by reducing the IC's
dimensionality to the correct dimension, and diluting their points
to get the desired density. Note that in Theorem
\ref{Th:LowerBoundDiversityOrder} we showed that for each
multiplexing gain, $r$, there exists a sequence of IC's that attains
the optimal DMT. On the other hand, in Corollary
\ref{cor:SingleICAttainsTheOptimalDMT} we show that a single
sequence of IC's attains the optimal DMT for any $r$, by adapting
its dimensionality and diluting its points. Also note that
$K_{0}T_{0}>K_{1}T_{1}>\cdots>K_{L-1}T_{L-1}$.
\begin{cor}\label{cor:SingleICAttainsTheOptimalDMT}
There exists a single sequence of $K_{0}T_{0}$-complex dimensional
IC's, that attains the $L$ segments of the optimal DMT:
$$(M-l)(N-l)-(r-l)(N+M-2\cdot l-1)\quad 0\le r\le K_{l}$$
where $l=0,\cdots,L-1$. The $l'th$ segment is attained by reducing
the IC's complex dimensionality to $K_{l}T_{l}$, and by diluting
their points to get density $\gamma_{tr}=\rho^{T_{l}r}$.
\end{cor}
\begin{proof}
See Appendix \ref{append:SingleICAttainsTheOptimalDMT}.
\end{proof}

\section{Discussion}\label{sec:LatticeVsLatticeBasedFC}
In this section we discuss the results presented in the paper. We
begin by explaining why full dimension lattice based coding schemes
such as Golden-codes \cite{BelfioreGoldenCodes}, perfect codes
\cite{BelfiorePerfectCodes} and other cyclic-division algebra based
space-time codes \cite{EliaExplicitSTC} which were shown to attain
the optimal DMT, are sub-optimal when regular lattice decoder
\eqref{eq:BasicDefinitionsRegularLatticeDecoding} is employed in the
receiver. In addition, we explain why using the MMSE estimation in
the receiver enables these schemes to attain the optimal DMT.
Afterwards, based on our results, we give another geometrical
interpretation to the optimal DMT. Finally, since in practice a
finite codebook is transmitted, we show that given a lattice with
multiplexing gain $r$ as defined for IC's in
\eqref{eq:ICMGDefinition}, a finite constellation with multiplexing
gain $r$ as defined in \cite{TseDivMult2003} can also be carved from
it.

\subsection{Lattice Constellations Vs. Full Dimension Lattice Based Finite Constellations}\label{subsec:LatticeVsLatticeBasedFC}
In order to demonstrate that full dimension lattice based coding
schemes with regular lattice decoding are sub-optimal let us
consider Golden-codes transmitted over a channel with $M=N=2$ where
$T=2$. For large $\rho$ the channel singular values PDF is
proportional to $\rho^{-\alpha_{1}-3\alpha_{2}}$, where
$\alpha_{1}\ge\alpha_{2}\ge 0$. A Golden-code of a certain rate is
carved from a $4$-complex dimensional lattice. We show that when
performing regular lattice decoding in the receiver the maximal
diversity order that can be attained for $r=0$ is $2$. This is in
contrast to ML decoding or alternatively MMSE estimation followed by
lattice decoding \cite{ElGamalLAST2004},
\cite{EliaJaldenDMTOptLRLinearLaticeDecoding} for which the maximal
diversity order equals 4.

We begin by showing why the maximal diversity order of a Golden-code
is 2 when performing regular lattice decoding. In the receiver, the
squared effective radius of the effective lattice induced by the
channel realization equals \eqref{eq:EffectiveRadius}
\begin{equation}
r_{\eff}^{2}
\dot{=}\rho^{-\frac{\alpha_{1}+\alpha_{2}}{2}}\dot{=}\gamma_{\rc}^{-\frac{1}{4}}.
\end{equation}
For lattices $r_{\eff}\ge
r_{\packing}=\frac{d_{\min}^{\lf(\lattice\ri)}}{2}$, where
$r_{\packing}$, $d_{\min}^{\lf(\lattice\ri)}$ are the packing radius
and the minimal distance of the lattice respectively. Hence, we get
\begin{equation}\label{eq:MinimalDistanceFadingLattices}
\lf(\frac{d_{\min}^{\lf(\lattice\ri)}}{2}\ri)^{2}\dot{\le}
\rho^{-\frac{\alpha_{1}+\alpha_{2}}{2}}.
\end{equation}
When the squared minimal distance is in the order of the additive
noise variance, $\rho^{-1}$, the error probability will not decrease
with $\rho$. This will happen for instance when $\alpha_{2}=0$ and
$\alpha_{1}=2$. This event occurs for large $\rho$ with probability
proportional to $\rho^{-2}$. Hence, in this case the diversity order
is 2. Note that for the 4-complex dimensional lattice we get
\eqref{eq:MuRc}
\begin{equation}\label{eq:VNRMinimalDistanceRatioFadingLattices}
\mu_{\rc}\dot{=}\frac{r_{\eff}^{2}}{\rho^{-1}}\dot{=}\rho^{1-\frac{\alpha_{1}+\alpha_{2}}{2}}.
\end{equation}
Therefore, the event where the squared effective radius is in the
order of the noise variance is equivalent to $\mu_{\rc}\dot{=}1$
which is the outage event for lattices, presented in Theorem
\ref{Th:UpperBoundDiversityOrder}.

From equation \eqref{eq:MinimalDistanceFadingLattices} we get that
the minimal distance for each channel realization of the
\emph{entire} lattice, induces diversity order 2. On the other hand,
when the decoder only considers the words within the \emph{finite}
codebook, the non-vanishing determinant (NVD) property combined with
the boundaries of the codebook leads to a lower bound on the minimal
distance of the Golden-code for each channel realization, that is
larger than the expression in
\eqref{eq:MinimalDistanceFadingLattices}, and enables to attain
diversity order 4 \cite{EliaExplicitSTC}.

The fact that considering the entire lattice leads to smaller
minimal distance is not surprising since the multiplication of the
transmitted lattice with the channel realization leads to scaling of
this lattice in the direction of the channel singular values. When
considering the infinite lattice, the scaling may reduce the
distance between points that were very far in the transmitted
lattice. These points are not necessarily part of the finite
codebook and therefore does not effect the minimal distance of the
finite Golden-code but do effect the minimal distance of the
lattice.

MMSE estimation followed by lattice decoding will also lead to
diversity order 4. Translating the arguments presented in
\cite{ElGamalLAST2004},
\cite{EliaJaldenDMTOptLRLinearLaticeDecoding} to our setting leads
to VNR
\begin{equation}
\tilde{\mu}_{\rc}\dot{=}\rho^{\frac{\lf(1-\alpha_{1}\ri)^{+}+\lf(1-\alpha_{2}\ri)^{+}}{2}}
\end{equation}
where $\lf(x\ri)^{+}=x$ for $x\ge 0$ and zero else. This expression
is larger than the expression in
\eqref{eq:VNRMinimalDistanceRatioFadingLattices} and implies that
the MMSE estimation, that takes into account the transmitted power,
also improves the minimal distance for each channel realization.
However, the improvement in VNR (and minimal distance) comes at the
expense of a self additive noise that depends on the transmitted
codeword. Under the assumption that the transmitted codewords are
not too far from the origin the variance of the effective noise is
small enough to allow attaining the optimal DMT. For instance
Golden-code codewords are from a bounded shaping region, which
enables to attain diversity order 4. Note that for the entire
lattice, the farther the lattice point is from the origin, the
larger the effective noise variance is. This eventually leads to
poor error performance for lattice points far enough from the
origin.

Our work shows that transmitting a lattice with average number of
dimensions per channel use $K=\frac{4}{3}$ and performing regular
lattice decoding in the receiver leads to VNR
\begin{equation}
\mu_{\rc}\dot{=}\rho^{1-\frac{\alpha_{1}}{4}-\frac{3\alpha_{2}}{4}}
\end{equation}
which is also larger than
\eqref{eq:VNRMinimalDistanceRatioFadingLattices} and enables to
attain diversity order 4 (in fact it attains the optimal DMT in the
range $0\le r\le 1$). Hence, from our work we can see that reducing
the lattice dimensionality increases the \emph{lattice} minimal
distance to such an extent that enables to attain the optimal DMT
when performing regular lattice decoding. In this sense reducing the
lattice dimensionality takes the role of MMSE estimation. It is also
interesting to note that MMSE estimation followed by lattice
decoding yields good error performance for lattice points close
enough to the origin (for instance lattice points within the shaping
region), and bad performance for lattice points very far from the
origin. On the other hand, regular lattice decoding yields the same
performance for all lattice points inside or outside the shaping
region. An illustrative example that shows how reduced dimension
assists in increasing the minimal distance compared to full
dimension lattice is presented in Figure
\ref{fig:FiniteVsInfiniteConstellations1}.

\begin{figure}
\centering \subfigure[Finite constellation: In this case even when
$h_{2}$ is small it is possible to decode.]{
\setlength{\unitlength}{1bp}%
\begin{picture}(392.48, 107.73)(0,0)
\put(0,0){\includegraphics{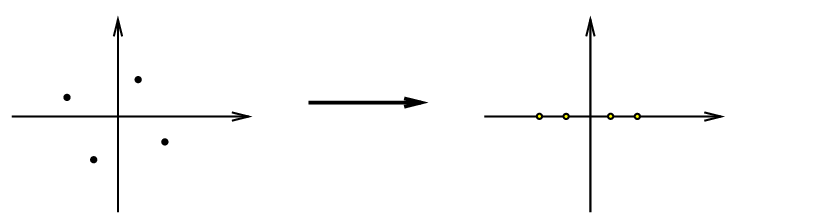}}
\put(285.63,95.79){\fontsize{8.03}{9.64}\selectfont $h_{2}x_{2}$}
  \put(339.45,59.05){\fontsize{8.03}{9.64}\selectfont $h_{1}x_{1}$}
  \put(58.85,95.79){\fontsize{8.03}{9.64}\selectfont $x_{2}$}
  \put(112.68,59.05){\fontsize{8.03}{9.64}\selectfont $x_{1}$}
  \end{picture}%
} \subfigure[Full dimensional infinite constellaion: In this case
due to the infiniteness of the constellation when $h_{2}$ is very
small it is impossible to decode.]{
  \setlength{\unitlength}{1bp}%
  \begin{picture}(395.32, 107.73)(0,0)
  \put(0,0){\includegraphics{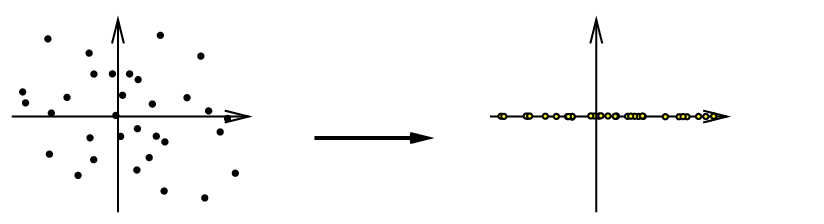}}
  \put(58.85,95.79){\fontsize{8.03}{9.64}\selectfont $x_{2}$}
  \put(112.68,59.05){\fontsize{8.03}{9.64}\selectfont $x_{1}$}
  \put(288.46,95.79){\fontsize{8.03}{9.64}\selectfont $h_{2}x_{2}$}
  \put(342.29,59.05){\fontsize{8.03}{9.64}\selectfont $h_{1}x_{1}$}
  \end{picture}%
} \subfigure[Infinite constellaion with reduced dimension: In this
case even when $h_{2}$ is very small it is possible to decode.]{
  \setlength{\unitlength}{1bp}%
  \begin{picture}(392.48, 107.73)(0,0)
  \put(0,0){\includegraphics{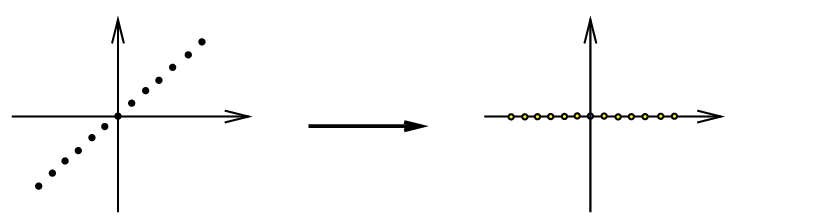}}
  \put(58.85,95.79){\fontsize{8.03}{9.64}\selectfont $x_{2}$}
  \put(112.68,59.05){\fontsize{8.03}{9.64}\selectfont $x_{1}$}
  \put(285.63,95.79){\fontsize{8.03}{9.64}\selectfont $h_{2}x_{2}$}
  \put(339.45,59.05){\fontsize{8.03}{9.64}\selectfont $h_{1}x_{1}$}
  \end{picture}%
}\caption{Illustrative example for the case $M=2$, $N=2$ of the
significance of reducing dimensions when considering regular lattice
decoding. For this example we assume that the realization of $H$ is
diagonal, where the diagonal elements are $h_{1}$ and
$h_{2}$.}\label{fig:FiniteVsInfiniteConstellations1}
\end{figure}
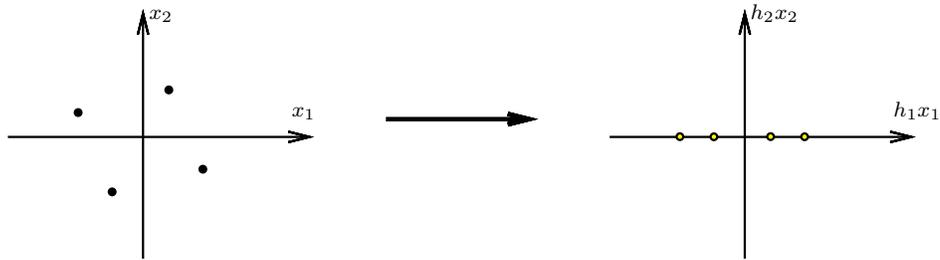
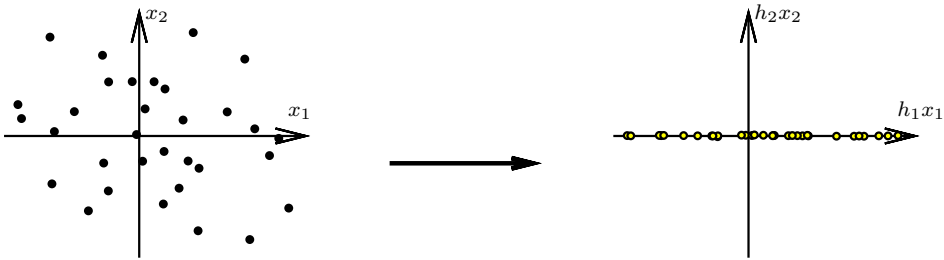
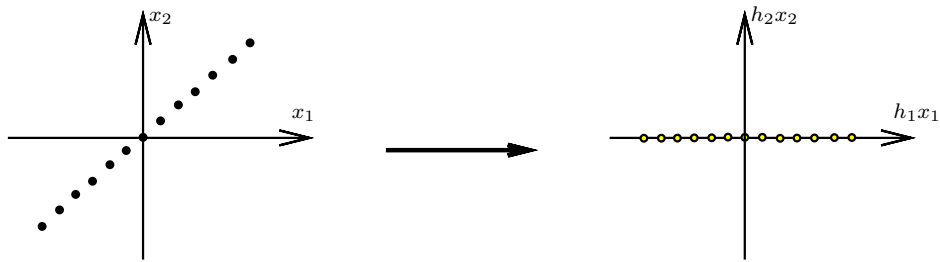

\subsection{Geometrical Interpretation of the Optimal DMT, for IC's}
In this subsection we give a geometrical interpretation of the
optimal DMT, based on allocation of lattice dimensions. This is a
qualitative discussion and the exact results appear in sections
\ref{sec:LowerBoundErrorProb}, \ref{sec:LowerBoundDiversityOrder}.

First from our results we can see that for a sequence of lattices
with certain number of dimensions per channel use the DMT is a
straight line as a function of the multiplexing gain (see Corollary
\ref{Cor:SequenceICAttainsTheEntireLine}). It results from the fact
that for lattices changing the multiplexing gain is equivalent to
scaling each dimension by $\rho^{-\frac{r}{2K}}$. Assume that the
sequence of lattices attains for multiplexing gain $r=0$ diversity
order $d$, i.e. the error probability decays as $\rho^{-d}$. In this
case scaling each dimension by $\rho^{-\frac{r}{2K}}$ leads to error
probability that decays as $\rho^{-d \lf(1-\frac{r}{K}\ri)}$. This
behavior results from the fact that the lattice decoder takes into
consideration all the lattice points. Hence, the scaling merely
replaces $\rho$ with $\rho^{1-\frac{r}{K}}$ in the error probability
expression. The optimal DMT is a piecewise linear function. We get
that each line corresponds to a sequence of lattices with certain
number of dimensions per channel use.

Next we wish to give the reasoning for the average number of
dimensions per channel use required to achieve each line in the
optimal DMT. For simplicity let us consider the case $M=N=3$. We
begin by considering the straight line in the range $0\le r\le 1$.
In this range the optimal DMT equals $9-5\cdot r$. We wish to show
why the average number of dimensions per channel use that enables to
attain this straight line equals $\frac{9}{5}$. For large $\rho$ the
channel singular values PDF is of the form of
$\rho^{-\alpha_{1}-3\alpha_{2}-5\alpha_{3}}$, where
$\alpha_{1}\ge\alpha_{2}\ge\alpha_{3}\ge 0$. When the transmission
scheme spreads over $T$ channel uses, the equivalent channel matrix,
$H_{\ex}$,  presented in \eqref{eq:ExtendedChannelModel} has $3T$
singular values. Each singular value of $H$ occurs $T$ times in the
singular values of $H_{\ex}$. Assume each complex dimension of the
lattice is transmitted on a certain singular value of $H_{\ex}$. Let
us denote by $T_{i}$ the number of dimensions transmitted on the
singular values that equal $\rho^{-\frac{\alpha_{i}}{2}}$, $1\le
i\le 3$. Note that $\sum_{i=1}^{3}T_{i}$ may be smaller than $3T$.
According to this assumption a $\lf(\sum_{i=1}^{3}T_{i}\ri)$-complex
dimensional lattice is transmitted over $T$ channel uses, and the
average number of dimensions per channel use is
$K=\frac{\sum_{i=1}^{3}T_{i}}{T}$. The effective radius in the
receiver equals
\begin{equation}\label{eq:FiniteVsInfiniteConstellations8}
r_{\eff}\dot{=}\rho^{-\frac{r\cdot
T}{\sum_{i=1}^{3}T_{i}}-\frac{T_{1}\alpha_{1}+T_{2}\alpha_{2}+T_{3}\alpha_{3}}{\sum_{i=1}^{3}T_{i}}}.
\end{equation}
and the VNR equals
\begin{equation}\label{eq:FiniteVsInfiniteConstellations7}
\mu_{\rc}\dot{=}\rho^{1-\frac{r\cdot
T}{\sum_{i=1}^{3}T_{i}}-\frac{T_{1}\alpha_{1}+T_{2}\alpha_{2}+T_{3}\alpha_{3}}{\sum_{i=1}^{3}T_{i}}}.
\end{equation}
We are interested in the probability of the outage event, i.e. the
probability that $\mu_{\rc}\dot{=}1$. Essentially, we show that when
$K<\frac{9}{5}$ it is possible to attain maximal diversity order of
$9$ for $r=0$, but it is impossible to attain the line $9-5\cdot r$
for any $0<r\le\frac{9}{5}$. It results from the fact that
multiplexing gain $r>0$ requires \emph{scaling} each dimension by
$\rho^{-\frac{r}{2K}}=\rho^{-\frac{r\cdot
T}{2\sum_{i=1}^{3}T_{i}}}$, which decreases $r_{\eff}$ (and as a
consequence also decreases the lattice minimal distance) to such an
extent that it does not enable to attain the optimal DMT. On the
other hand when $K>\frac{9}{5}$ the \emph{channel} decreases
$r_{\eff}$ to such an extent that it does not enable to attain the
optimal DMT for $0\le r<1$. Hence, $K=\frac{9}{5}$ balances the
effect of the scaling and the channel and allows to attain the
optimal DMT in the range $0\le r\le 1$.

In order to attain the maximal diversity order $9$ when $r=0$, the
outage event $\mu_{\rc}\dot{=}1$ implies that the following
conditions need to be fulfilled
\begin{equation}\label{eq:FiniteVsInfiniteConstellations6}
\frac{T_{1}}{\sum_{i=1}^{3}T_{i}}\le\frac{1}{9},\qquad
\frac{T_{1}+T_{2}}{\sum_{i=1}^{3}T_{i}}\le\frac{4}{9}
\end{equation}
i.e. each singular value can not occur in more dimensions than the
relative effect it has on the PDF of the singular values. The
largest average number of dimensions per channel use that fulfils
\eqref{eq:FiniteVsInfiniteConstellations6} is $\frac{9}{5}$. In this
case for $T=5$ a 9-complex dimensional lattice is transmitted, and
the conditions are fulfilled with equality when $T_{1}=1$, $T_{2}=3$
and $T_{3}=5$. When $K<\frac{9}{5}$ the conditions in
\eqref{eq:FiniteVsInfiniteConstellations6} are still fulfilled and
therefore diversity order $9$ is still attained for $r=0$. However,
based on \eqref{eq:FiniteVsInfiniteConstellations8} we get for $r>0$
that $r_{\eff}$ decreases faster than the case of $K=\frac{9}{5}$.
Hence, for $K<\frac{9}{5}$ the diversity order is smaller than
$9-5\cdot r$ when $0<r\le\frac{9}{5}$.

So far we have shown that choosing $K<\frac{9}{5}$ leads to sub-optimal DMT.
Now, we wish to show that in the range $0\le r\le 1$ the DMT is
smaller than $9-5\cdot r$ also when $K>\frac{9}{5}$. First, for
$K>\frac{9}{5}$ the conditions in
\eqref{eq:FiniteVsInfiniteConstellations6} are not met. Hence, in
this case the diversity order is smaller than $9$ when $r=0$. For
$r=1$ and $K=\frac{9}{5}$ the diversity order equals $4$. Assume the
best assignment of lattice dimensions would enable to choose
$T_{3}=T$. In this case $\mu_{\rc}$ in
\eqref{eq:FiniteVsInfiniteConstellations7} is effected equally if
$r=1$, $\alpha_{3}=0$ or $r=0$, $\alpha_{3}=1$, i.e. the scaling inflicted by $r=1$ decreases $r_{\eff}$ in \eqref{eq:FiniteVsInfiniteConstellations8} as if the singular value $\rho^{-\frac{\alpha_{3}}{2}}=\rho^{-\frac{1}{2}}$. In both cases we get
\begin{equation}
\mu_{\rc}=\rho^{\frac{T_{1}+T_{2}-T_{1}\alpha_{1}-T_{2}\alpha_{2}}{T_{1}+T_{2}+T}}.
\end{equation}
The difference is that when $r=1$, $\alpha_{3}=0$ the PDF of the
singular values equals $\rho^{-\alpha_{1}-3\alpha_{2}}$ which leads
to smaller diversity order than the case $r=0$, $\alpha_{3}=1$. For
large $\rho$ and $r=1$, $\alpha_{3}=0$ is included in the most
dominant error event when $K\ge\frac{9}{5}$. Hence, diversity order
of $4$ is attained for $r=1$ and $K>\frac{9}{5}$ when the following
condition is met
\begin{equation}
\frac{T_{1}}{T_{1}+T_{2}}\le\frac{1}{4}
\end{equation}
which is exactly the condition for attaining maximal diversity order
of $4$ when $r=0$ in a channel with $2$ transmit and $2$ receive
antennas. This condition is met as long as $K\le \frac{7}{3}$.
Hence, for $\frac{9}{5}<K\le\frac{7}{3}$ the best diversity order is
smaller than $9$ when $r=0$, and equals $4$ when $r=1$. Since for
each $K$ the largest DMT is a straight line, the DMT for each
$0<K\le\frac{7}{3}$ in the range $0\le r\le 1$ is smaller than
$9-5\cdot r$. We are left with the case $\frac{7}{3}<K\le 2$. By
applying similar arguments, only this time considering $r=2$, it can
be shown that in the range $0\le r < 2$ the largest DMT for any
$\frac{7}{3}<K\le 2$ is smaller than $7-3\cdot r$. These arguments
also show that in the range $2\le r\le 3$ the optimal DMT equals
$2-r$. Hence, we get for $0\le r <1$ that the optimal DMT equals
$9-5\cdot r$, where for $1\le r<2$, $2\le r\le 3$ the optimal DMT
equals $7-3\cdot r$ and $2-r$ respectively.

\subsection{The Relation Between the Multiplexing Gains of an IC and a Finite Constellation}\label{subseec:LatticeVsLatticeBasedFCICMGtoFCMG}
In this paper we defined the multiplexing gain of IC's sequence as
the rate the IC's density increases \eqref{eq:ICMGDefinition}, i.e.
when $\gamma_{\mathrm{tr}}=\rho^{rT}$ the multiplexing gain is $r$.
We characterized the optimal DMT of IC's based on this definition of
the multiplexing gain. In practice a finite constellation is
transmitted, even when performing regular lattice decoding in the
receiver. Hence, in this subsection we show that finite
constellation with multiplexing gain $r$ can be carved from a
lattice with multiplexing gain $r$ (according to the definition
given in \eqref{eq:ICMGDefinition}), while maintaining the same
performance when performing regular lattice decoding in the
receiver.

Consider a lattice $\Lambda$ with density
$\gamma_{\mathrm{tr}}=\rho^{rT}$. In this case for each lattice
point the Voronoi region volume equals
\begin{equation*}
|V \lf(x\ri)|=|V|=\gamma_{\mathrm{tr}}^{-1}=\rho^{-rT}\quad \forall
x\in\Lambda.
\end{equation*}
In \cite{LoeligerAveragingBounds} it has been shown that for any
Jordan measurable bounded set $S$  with volume $|V \lf(S\ri)|$ there
exists a translate $u$ such that
\begin{equation}\label{eq:SubsectionICMGFCMG}
|\lf(\Lambda+u\ri)\cap S|\ge \frac{|V \lf(S\ri)|}{|V|}
\end{equation}
where $\Lambda+u$ is the translate of each lattice point by the
constant $u$, and $|\lf(\Lambda+u\ri)\cap S|$ is the number of words
of the translated lattice within the region $S$. Hence, for each
lattice in a sequence with multiplexing gain $r$, there exists a
translate such that the number of codewords within a sphere with
volume 1 is larger or equal to $\rho^{rT}$, i.e. the rate is $r\log
\lf(\rho\ri)$ where in this setting $\rho$ takes the role of $\SNR$.
Hence, it is possible to carve from the translated lattices sequence
a finite constellations sequence with multiplexing gain $r$
according to the definitions of finite constellations. When
performing regular lattice decoding the translate does not effect
the performance. Hence, the results we presented in this work also
apply when carving finite constellations with the corresponding
multiplexing gain from the lattices sequence, and performing regular
lattice decoding in the receiver.

\section{Summary}
This work investigates the DMT of IC's. A new tradeoff between the
IC average number of dimensions per channel use and the best DMT it
may attain is presented. Based on this tradeoff a transmission
scheme that enables to attain the optimal DMT of finite
constellations, by lattices with regular lattice decoding, is
presented.

\begin{appendices}
\section{Proof of Theorem \ref{Th:LowerBoundChannelReal}}\label{append:ProofFirstTheorem}
We prove the result for any IC with density $\gamma_{rc}$. The proof
outline is as follows. We prove the theorem by contradiction. First,
for a given IC with receiver density $\gamma_{rc}$, we assume an
average decoding error probability that equals to the lower bound we
wish to prove. Then, we derive a ``regular'' IC from the given IC
with the same density $\gamma_{rc}$ and the same average decoding
error probability. Regularizing the IC allows us to find a lower
bound on the IC maximal error probability that depends on its
density. We expurgate half of the codewords with the largest error
probability and get another regular IC with density
$\frac{\gamma_{rc}}{2}$. Based on the average decoding error
probability, we upper bound the expurgated IC maximal error
probability, and based on its density we lower bound the same
maximal error probability, and get a contradiction.

Let us consider a $KT$-complex dimensional IC in the receiver,
$S^{'}_{KT}(\rho)$, with receiver density $\gamma_{rc}$ and average
decoding error probability
\begin{equation}\label{eq:LowerBoundErrorProbAppendix}
\overline{P_{e}}(H,\rho)=
(1-\epsilon^{\ast})\frac{\ol{C}(KT)}{4}e^{-\mu_{\mathrm{rc}}\cdot
\ol{A}(KT)+(KT-1)\ln (\mu_{\mathrm{rc}})}
\end{equation}
where
$\ol{A}(KT)=(\frac{1}{(1-\epsilon_{1})(1-\epsilon_{2})})^{\frac{1}{KT}}e\cdot\Gamma(KT+1)^{\frac{1}{KT}}$,
$\ol{C}(KT)=(\frac{1}{(1-\epsilon_{1})(1-\epsilon_{2})})^{\frac{KT-1}{KT}}\frac{e^{KT-\frac{3}{2}}\Gamma(KT+1)^{\frac{KT-1}{KT}}}{2\cdot\Gamma(KT)}$
and $0<\epsilon_{1},\epsilon_{2}<1$.

Next we construct a regularized IC, $S^{''}_{KT}(\rho)$, from
$S^{'}_{KT}(\rho)$, whose Voronoi regions are bounded and have
finite volumes , i.e. there exists a finite radius $r$ such that
$V(x)\subset Ball(x,r)$, $\forall x\in S^{''}_{KT}(\rho)$, where
$Ball(x,r)$ is a $KT$-complex dimensional ball centered around $x$.
We construct $S^{''}_{KT}(\rho)$ in the following manner. Let us
define $C_{0}(\rho,H)=\{S^{'}_{KT}(\rho)\bigcap (H_{ex}\cdot
cube_{KT}(b))\}$, i.e. a finite constellation derived from
$S^{'}_{KT}(\rho)$. We turn this finite constellation into an IC by
tiling $C_{0}(\rho,H)$ in the following manner
\begin{equation}\label{eq:TilingS1IntoS2}
S^{''}_{KT}(\rho)=C_{0}(\rho,H)+(b+b^{'})\tilde{H}_{ex}\mathbb{Z}^{2KT}
\end{equation}
where for simplicity we assumed that
$cube_{KT}(b)\subset\mathbb{C}^{KT}$, i.e. contained within the
first $KT$ complex dimensions. Correspondingly, under this
assumption, $\tilde{H}_{ex}$ equals the first $KT$ complex columns
of $H_{ex}$. In this case, the tiling of $C_{0}(\rho,H)$ is done
according to the complex integer combinations of $\tilde{H}_{ex}$
columns. In general, $cube_{KT}(b)$ may be a rotated cube within
$\mathbb{C}^{MT}$. In this case the tiling is done according to some
$KT$ complex linearly independent vectors, consisting of linear
combinations of $H_{ex}$ columns. An alternative way to construct
$S^{''}_{KT}(\rho)$ is by considering the transmitter IC
$S_{KT}(\rho)$. In this case we can construct another IC in the
transmitter
\begin{equation}\label{eq:TilingInTransmitter}
\ol{S}_{KT}(\rho)=\{S_{KT}(\rho)\bigcap
cube_{KT}(b)\}+(b+b^{'})\mathbb{Z}^{2KT}
\end{equation}
where without loss of generality we assumed again that
$cube_{KT}(b)\in\mathbb{C}^{KT}$. In this case
$S^{''}_{KT}(\rho)=\{H_{ex}\cdot\ol{S}_{KT}(\rho) \}$.

Next we would like to set $b$ and $b^{'}$ to be large enough such
that $S^{''}_{KT}(\rho)$ has average decoding error probability
smaller or equal to $\frac{\ol{C}(KT)}{2}e^{-\mu_{\mathrm{rc}}\cdot
\ol{A}(KT)+(KT-1)\ln (\mu_{\mathrm{rc}})}$ and density larger or
equal to $\gamma_{rc}$. Due to the symmetry that results from the
tiling \eqref{eq:TilingS1IntoS2}, it is sufficient to upper bound
the average decoding error probability of the points $x\in
C_{0}(\rho,H)\subset S^{''}_{KT}(\rho)$ denoted by
$P_{e}^{S^{''}_{KT}}(C_{0})$ in order to upper bound the average
decoding error probability of the entire IC $S^{''}_{KT}(\rho)$ .
Hence $P_{e}^{S^{''}_{KT}}(C_{0})$ is also the average decoding
error probability for the IC $S^{''}_{KT}(\rho)$. We can upper bound
the error probability in the following manner
\begin{equation}\label{eq:UpperBoundOnS2TagErrorProbability}
P_{e}^{S^{''}_{KT}}(C_{0})\le
P_{e}(C_{0})+P_{e}(S^{''}_{KT}\setminus C_{0})
\end{equation}
where $P_{e}(C_{0})$ is the average decoding error probability of
the finite constellation $C_{0}(\rho,H)$ and
$P_{e}(S^{''}_{KT}\setminus C_{0})$ is the average decoding error
probability to points in the set $\{S^{''}_{KT}\setminus
C_{0}(\rho,h)\}$, i.e. the error probability inflicted by the
replicated codewords outside the set $C_{0}(\rho,H)$.


We begin by upper bounding $P_{e}(S^{''}_{KT}\setminus C_{0})$ by
choosing $b^{'}$ to be large enough. By the tiling at the
transmitter \eqref{eq:TilingInTransmitter} and the fact that we have
finite complex dimension $KT$, for a certain channel realization
$H_{ex}$ we get that there exists $\delta(H_{ex})$ such that any
pair of points $x_{1}\in C_{0}(\rho,H)$,
$x_{2}\in\{S^{''}_{KT}\setminus C_{0}(\rho,h)\}$ fulfils $\lVert
\underline{x}_{1}-\underline{x}_{2}\rVert\ge
2b^{'}\cdot\delta(H_{ex})$. The term $\delta(H_{ex})$ is a factor
that defines the minimal distance between these 2 sets for a given
channel realization. Note that also for the case $M>N$, there must
exist such $\delta(H_{ex})$, as we assumed that $S^{''}_{KT}(\rho)$
is $KT$-complex dimensional IC, i.e. the projected IC
$S^{''}_{KT}(\rho)=H_{ex}\ol{S}_{KT}(\rho)$ is also $KT$-complex
dimensional. Hence, we get that
$$P_{e}(S^{''}_{KT}\setminus C_{0})\le Pr(\lVert \underline{\tilde{n}}_{\mathrm{ex}}\rVert\ge b^{'}\delta(H_{ex}))$$
where $\underline{\tilde{n}}_{\mathrm{ex}}$ is the effective noise
in the $KT$-complex dimensional hyperplane where $S^{''}_{KT}(\rho)$
resides. By using the upper bounds from \cite{PoltirevJournal}, we
get that for $\frac{(b^{'}\delta(H_{ex}))^{2}}{2KT}>\sigma^{2}$
$$Pr(\lVert \underline{\tilde{n}}_{\mathrm{ex}}\rVert\ge b^{'}\delta(H_{ex}))\le
e^{-\frac{(b^{'}\delta(H_{ex}))^{2}}{2\sigma^{2}}}(\frac{(b^{'}\delta(H_{ex}))^{2}e}{2KT\sigma^{2}})^{KT}.$$
Hence, for $b^{'}$ large enough we get that
$$P_{e}(S^{''}_{KT}\setminus C_{0})\le (1-\epsilon^{\ast})\frac{\ol{C}(KT)}{4}e^{-\mu_{\mathrm{rc}}\cdot
\ol{A}(KT)+(KT-1)\ln (\mu_{\mathrm{rc}})}.$$

Now we would like to upper bound the error probability,
$P_{e}(C_{0})$, of the finite constellation $C_{0}(\rho,H)$.
According to the definition of the average decoding error
probability in \eqref{eq:AverageDecodingErrorProbability}, the
definition of $C_{0}(\rho,H)$ and the assumption in
\eqref{eq:LowerBoundErrorProbAppendix}, we get that
$$P_{e}(C_{0})\le \frac{(1-\epsilon^{\ast})(1+\epsilon(b))}{4}\ol{C}(KT)e^{-\mu_{\mathrm{rc}}\cdot
\ol{A}(KT)}\cdot\mu_{\mathrm{rc}}^{(KT-1)}$$ where
$lim_{b\to\infty}\epsilon(b)=0$. It results from the fact that in
\eqref{eq:AverageDecodingErrorProbability} we take the limit
supremum, and so for $b$ large enough the average decoding error
probability of the IC must be upper bounded by the aforementioned
term. Also, for any $b$ the average decoding error probability of
the finite constellation $C_{0}(\rho,H)$ is smaller or equal to the
error probability, defined in
\eqref{eq:AverageDecodingErrorProbability}, of decoding over the
entire IC. Based on the upper bound from
\eqref{eq:UpperBoundOnS2TagErrorProbability} we get the following
upper bound on the error probability of $S^{''}_{KT}(\rho)$
\begin{equation}\label{eq:UpperBoundOnS2TagErrorProbabilityExpression}
P_{e}^{S^{''}_{KT}}(C_{0})\le
\tfrac{(1-\epsilon^{\ast})(1+\epsilon(b))}{2}\ol{C}(KT)e^{-\mu_{\mathrm{rc}}\cdot
\ol{A}(KT)}\cdot\mu_{\mathrm{rc}}^{(KT-1)}.
\end{equation}

According to the definition of $\gamma_{rc}$ and due to the fact
that we are taking limit supremum: for any $0<\epsilon_{1}<1$ there
exists $b$ large enough such that
\begin{equation}\label{eq:C0Density}
\frac{|C_{0}(\rho,H)|}{vol\big(H_{ex}\cdot cube_{KT}(b)\big)}\ge
(1-\epsilon_{1})\gamma_{rc}.
\end{equation}
where $|C_{0}(\rho,H)|$ is the number of points in $C_{0}(\rho,H)$.
In fact there exists large enough $b$ that fulfils both
\eqref{eq:UpperBoundOnS2TagErrorProbabilityExpression} and
\eqref{eq:C0Density}.

In \eqref{eq:TilingS1IntoS2} we tiled by $b+b^{'}$. If we had tiled
$C_{0}(\rho,H)$ only by $b$, then for large enough $b$ we would
have got IC with density larger or equal to
$(1-\epsilon_{1})\gamma_{rc}$. However , as we tile by $b+b^{'}$,
we get for $b$ large enough that $S^{''}_{KT}(\rho)$ has
density greater or equal to
$\frac{1-\epsilon_{1}}{1+\frac{b^{'}}{b}}\gamma_{rc}$. Hence, for any
$0<\epsilon_{2}<1$ there exists $b$ large enough such that
\begin{equation}\label{eq:LowerBoundDensityS2Tag}
\gamma^{''}_{rc}\ge (1-\epsilon_{1})(1-\epsilon_{2})\gamma_{rc}.
\end{equation}
where $\gamma^{''}_{rc}$ is the density of $S^{''}_{KT}(\rho)$.
Again, there also must exist large enough $b$ that fulfils
\eqref{eq:UpperBoundOnS2TagErrorProbabilityExpression} and
\eqref{eq:LowerBoundDensityS2Tag} simultaneously. Hence, for large
enough $b$ we can derive from $S^{'}_{KT}(\rho)$ an IC
$S^{''}_{KT}(\rho)$ with density
$\gamma^{''}_{rc}\ge(1-\epsilon_{1})(1-\epsilon_{2})\gamma_{rc}$ and
average decoding error probability smaller or equal to
$\frac{(1-\epsilon^{\ast})(1+\epsilon(b))}{2}\ol{C}(KT)e^{-\mu_{\mathrm{rc}}\cdot
\ol{A}(KT)+(KT-1)\ln (\mu_{\mathrm{rc}})}$.

By averaging arguments we know that expurgating the worst half of
the codewords in $S^{''}_{KT}(\rho)$, yields an IC
$S^{'''}_{KT}(\rho)$ with density
\begin{equation}\label{eq:LowerBoundDensityS3Tag}
\gamma^{'''}_{rc}\ge(1-\epsilon_{1})(1-\epsilon_{2})\frac{\gamma_{rc}}{2}=\ol{\gamma_{rc}}
\end{equation}
and maximal decoding error probability
\ifthenelse{\equal{\singlecolumntype}{1}}
{\begin{equation}\label{eq:UpperBoundErrorProbS3Tag} sup_{x\in
S^{'''}_{KT}} P_{e}^{S^{'''}_{KT}}(x)\le
(1-\epsilon^{\ast})(1+\epsilon(b))\ol{C}(KT)e^{-\mu_{\mathrm{rc}}\cdot
\ol{A}(KT)}\mu_{\mathrm{rc}}^{KT-1}
\end{equation}}
{\begin{equation}\label{eq:UpperBoundErrorProbS3Tag}
\small{sup_{x\in S^{'''}_{KT}} P_{e}^{S^{'''}_{KT}}(x)}\le
\footnotesize{(1-\epsilon^{\ast})(1+\epsilon(b))\ol{C}(KT)}e^{-\mu_{\mathrm{rc}}\cdot
\ol{A}(KT)}\mu_{\mathrm{rc}}^{KT-1}
\end{equation}}
where $P_{e}^{S^{'''}_{KT}}(x)$ is the error probability of $x\in
S^{'''}_{KT}(\rho)$.

From the construction method of $S^{''}_{KT}(\rho)$, defined in
\eqref{eq:TilingS1IntoS2}, it can be easily shown that tiling
$C_{0}(\rho,H)$ yields bounded and finite volume Voronoi regions,
i.e. there exists a finite radius $r$ such that $V(x)\subset
Ball(x,r)$, $\forall x\in S^{''}_{KT}(\rho)$. Due to the symmetry
that results from $S^{''}_{KT}(\rho)$ construction
\eqref{eq:TilingS1IntoS2}, it also applies for $S^{'''}_{KT}(\rho)$.
Hence, there must exist a point $x_{0}\in S^{'''}_{KT}(\rho)$ that
satisfies $|V(x_{0})|\le
\frac{1}{\gamma^{'''}_{rc}}\le\frac{1}{\ol{\gamma_{rc}}}$. According
to the definition of the effective radius in
\eqref{eq:EffectiveRadius}, we get that $r_{\eff}(x_{0})\le
r_{\eff}(\ol{\gamma_{rc}})$. Hence, we get
\ifthenelse{\equal{\singlecolumntype}{1}} {\begin{equation}
sup_{x\in S^{'''}_{KT}} P_{e}^{S^{'''}_{KT}}(x)\ge
P_{e}^{S^{'''}_{KT}}(x_{0})> Pr\big(\lVert
\underline{\tilde{n}}_{\mathrm{ex}}\rVert \ge
r_{\mathrm{eff}}(x_{0})\big)\ge Pr\big(\lVert
\underline{\tilde{n}}_{\mathrm{ex}}\rVert\ge
r_{\mathrm{eff}}(\ol{\gamma_{rc}})\big)
\end{equation}}
{\begin{align} sup_{x\in S^{'''}_{KT}} P_{e}^{S^{'''}_{KT}}(x)&\ge
P_{e}^{S^{'''}_{KT}}(x_{0})>
\nonumber\\
Pr\big(\lVert \underline{\tilde{n}}_{\mathrm{ex}}\rVert &\ge
r_{\mathrm{eff}}(x_{0})\big)\ge Pr\big(\lVert \underline{\tilde{n}}_{\mathrm{ex}}\rVert\ge
r_{\mathrm{eff}}(\ol{\gamma_{rc}})\big)
\end{align}}
where the lower bound $P_{e}^{S^{'''}_{KT}}(x_{0})>Pr(\lVert
\underline{\tilde{n}}_{\mathrm{ex}}\rVert \ge
r_{\mathrm{eff}}(x_{0}))$ was proven in \cite{PoltirevJournal}. We
calculate the following lower bound
\ifthenelse{\equal{\singlecolumntype}{1}}
{\begin{equation}\label{eq:LoweerBoundForSigmaLess1}
\Pr\big(\lVert\underline{\tilde{n}}_{\mathrm{ex}}\rVert\ge
r_{\mathrm{eff}}(\ol{\gamma_{\mathrm{rc}}})\big)>
\int_{r_{\mathrm{eff}}^{2}}^{r_{\mathrm{eff}}^{2}+\sigma^{2}}\frac{r^{KT-1}e^{-\frac{r}{2\sigma^{2}}}}{\sigma^{2KT}2^{KT}
\Gamma(KT)}dr
\ge\frac{r_{\mathrm{eff}}^{2KT-2}e^{-\frac{r_{\mathrm{eff}}^{2}}{2\sigma^{2}}}}{\sigma^{2KT-2}2^{KT}
\Gamma(KT)\sqrt{e}}
\end{equation}}
{\begin{align}\label{eq:LoweerBoundForSigmaLess1}
\Pr\big(\lVert\underline{\tilde{n}}_{\mathrm{ex}}\rVert\ge
r_{\mathrm{eff}}(\ol{\gamma_{\mathrm{rc}}})\big)&>
\nonumber\\
\int_{r_{\mathrm{eff}}^{2}}^{r_{\mathrm{eff}}^{2}+\sigma^{2}}\frac{r^{KT-1}e^{-\frac{r}{2\sigma^{2}}}}{\sigma^{2KT}2^{KT}
\Gamma(KT)}&dr
\ge\frac{r_{\mathrm{eff}}^{2KT-2}e^{-\frac{r_{\mathrm{eff}}^{2}}{2\sigma^{2}}}}{\sigma^{2KT-2}2^{KT}
\Gamma(KT)\sqrt{e}}
\end{align}}
By assigning
$r_{\eff}^{2}=(\frac{\Gamma(KT+1)}{\ol{\gamma_{rc}}\pi^{KT}})^\frac{1}{KT}$
we get \ifthenelse{\equal{\singlecolumntype}{1}} {\begin{equation}
sup_{x\in S^{'''}_{KT}} P_{e}^{S^{'''}_{KT}}(x)> \ol{C}(KT)\cdot
e^{-\frac{\gamma_{\mathrm{rc}}^{-\frac{1}{KT}}}{2\pi
e\sigma^{2}}\ol{A}(KT)+(KT-1)\ln
(\frac{\gamma_{\mathrm{rc}}^{-\frac{1}{KT}}}{2\pi e\sigma^{2}})}.
\end{equation}}
{\begin{align} sup_{x\in S^{'''}_{KT}} P_{e}^{S^{'''}_{KT}}(x)&>
\nonumber\\
\ol{C}(KT)&\cdot
e^{-\frac{\gamma_{\mathrm{rc}}^{-\frac{1}{KT}}}{2\pi
e\sigma^{2}}\ol{A}(KT)+(KT-1)\ln
(\frac{\gamma_{\mathrm{rc}}^{-\frac{1}{KT}}}{2\pi e\sigma^{2}})}.
\end{align}}
Hence, for certain $\epsilon_{1}$ and $\epsilon_{2}$ we get
\ifthenelse{\equal{\singlecolumntype}{1}}
{\begin{equation}\label{eq:LoweBoundWithEpsEst} sup_{x\in
S^{'''}_{KT}} P_{e}^{S^{'''}_{KT}}(x)> \ol{C}(KT)\cdot
e^{-\mu_{rc}\ol{A}(KT)+(KT-1)\ln (\mu_{rc})}
\end{equation}}
{\begin{align}\label{eq:LoweBoundWithEpsEst} sup_{x\in S^{'''}_{KT}}
P_{e}^{S^{'''}_{KT}}(x)&>
\nonumber\\
\ol{C}(KT)&\cdot e^{-\mu_{rc}\ol{A}(KT)+(KT-1)\ln (\mu_{rc})}
\end{align}}
where $\mu_{rc}=\frac{\gamma_{\mathrm{rc}}^{-\frac{1}{KT}}}{2\pi
e\sigma^{2}}$. For $b$ large enough we get
$(1-\epsilon^{\ast})(1+\epsilon(b))<1$, and so
\eqref{eq:LoweBoundWithEpsEst} contradicts
\eqref{eq:UpperBoundErrorProbS3Tag}. As a result we get
contradiction of the initial assumption in
\eqref{eq:LowerBoundErrorProbAppendix}. This contradiction also
holds for any $\overline{P_{e}}(H,\rho)<
\frac{(1-\epsilon^{\ast})\ol{C}(KT)}{4}e^{-\mu_{\mathrm{rc}}\cdot
\ol{A}(KT)+(KT-1)\ln (\mu_{\mathrm{rc}})}$. Hence, we get that
\begin{equation}\label{eq:FiniteLowerBoundErrorProb}
\overline{P_{e}}(H,\rho)>
\frac{\ol{C}(KT)}{4}e^{-\mu_{\mathrm{rc}}\cdot \ol{A}(KT)+(KT-1)\ln
(\mu_{\mathrm{rc}})}.
\end{equation}
Note that the lower bound holds for any
$0<\epsilon_{1},\epsilon_{2},\epsilon^{\ast}<1$ and also that the
expressions in \eqref{eq:LowerBoundErrorProbAppendix},
\eqref{eq:FiniteLowerBoundErrorProb} are continuous. As a result we
can also set $\epsilon_{1}=\epsilon_{2}=\epsilon^{\ast}=0$ and get
the desired lower bound. Finally, note that we are interested in a
lower bound on the error probability of any IC for a given channel
realization. Hence, we are free to choose different values for $b$
and $b^{'}$ for each channel realization.
and $b^{'}$.

\section{Proof of the optimization problem in Theorem \ref{Th:UpperBoundDiversityOrder}}\label{Append:OptimizationProblemSolutionOfUpperBoundDivOrder}
We would like to solve the optimization problem in
\eqref{eq:DiversityOrderUpperBoundOptimizationProblem} for any value
of $K=B+\beta\le L$, where $B\in\mathbb{N}$ and $0< \beta \le 1$.
First we consider the case of $0<K\le 1$, i.e. the case where $B=0$.
In this case the constraint boils down to
$\alpha_{L}=1-\frac{r}{K}$. By assigning
$\alpha_{1}=\dots=\alpha_{L}=1-\frac{r}{K}$ we get that
$d_{KT}(r)\le MN(1-\frac{r}{K})$. Next we analyze the case where
$K> 1$. Due to the constraint, the minimal value must satisfy
$\alpha_{1}=\dots=\alpha_{L-B}$. From the constraint we also know
that
$\alpha_{L}=K-r-\sum_{i=1}^{B-1}\alpha_{L-i}-\beta\alpha_{L-B}$. By
assigning in \eqref{eq:DiversityOrderUpperBoundOptimizationProblem}
we get
\ifthenelse{\equal{\singlecolumntype}{1}}
{\begin{equation}\label{Eq:EquivalenOptimizationProblem}
\min_{\udl{\alpha}> 0}(K-r)(N+M-1)+\big((M-B)(N-B)
-\beta(N+M-1)\big)\alpha_{L-B}-\sum_{i=1}^{B-1}2i\cdot\alpha_{L-i}
\end{equation}}
{\begin{align}\label{Eq:EquivalenOptimizationProblem}
\min_{\udl{\alpha}> 0}(K-r)&(N+M-1)+\big((M-B)(N-B)
\nonumber\\
&-\beta(N+M-1)\big)\alpha_{L-B}-\sum_{i=1}^{B-1}2i\cdot\alpha_{L-i}
\end{align}}
where $\alpha>0$ signifies $\alpha_{1}\ge\dots\ge\alpha_{L}\ge 0$.
We would like to consider two cases. The case where
$\big((M-B)(N-B)-\beta(N+M-1)\big)>\sum_{i=1}^{B-1}2i$ and the case
where $\big((M-B)(N-B)-\beta(N+M-1)\big)\le\sum_{i=1}^{B-1}2i$. The
first case, where $\big((M-B)(N-B)-\beta(N+M-1)\big)>B(B-1)$, is
achieved for $K<\frac{MN}{N+M-1}$. In this case we use the following
Lemma in order to find the optimal solution
\begin{lem}\label{Lem:LemmaOfTheOptimizationProblem}
Consider the optimization problem
$$\min_{\udl{c}}B_{1}c_{1}-\sum_{i=2}^{D}B_{i}c_{i}$$
where: $(1).$ $c_{1}\ge\dots\ge c_{D}\ge 0$; $(2)$.
$B_{1}>\sum_{i=2}^{D}B_{i}$ and $B_{2}>\dots >B_{D}>0$; $(3)$.
$\beta c_{1}+\sum_{i=2}^{D}c_{i}=\delta>0$, where $0<\beta\le 1$. The
minimal value is achieved for $c_{1}=\dots
=c_{D}=\frac{\delta}{D-1+\beta}$.
\end{lem}
\begin{proof}
We prove by induction. First let us consider the case where $D=2$.
In this case we would like to find
\begin{equation}\label{Lem:InductionFirstSum}
\min_{\udl{c}}B_{1}c_{1}-B_{2}c_{2}.
\end{equation}
where $c_{1}\ge c_{2}\ge 0$, $\beta c_{1}+c_{2}=\delta>0$,
$B_{1}>B_{2}>0$ and $0<\beta\le 1$. It is easy to see that for this
case the minimum is achieved for $c_{1}=c_{2}$, as increasing
$c_{1}$ while decreasing $c_{2}$ to satisfy $\beta
c_{1}+c_{2}=\delta$ will only increase
\eqref{Lem:InductionFirstSum}.

Now let assume that for $D$ elements, the minimum is achieved for
$c_{1}=\dots=c_{D}=\frac{\delta}{D-1+\beta}$. Let us consider $D+1$
elements with constraint $\beta c_{1}+\sum_{i=2}^{D+1}c_{i}=\delta$.
If we take $c_{1}=\dots=c_{D+1}=\frac{\delta}{D+\beta}$ we get
\begin{equation}\label{eq:LemmaInductionAssumption}
(B_{1}-\sum_{i=2}^{D+1}B_{i})\frac{\delta}{D+\beta}.
\end{equation}
We would like to show that this is the minimal possible value for
this problem. Take $c^{'}_{D+1}=\frac{\delta}{D+\beta}-\epsilon\ge
0$. In this case $\beta
c^{'}_{1}+\sum_{i=2}^{D}c_{i}^{'}=\frac{(D-1+\beta)\delta+(D+\beta)\epsilon}{D+\beta}$
in order to satisfy $\beta
c_{1}^{'}+\sum_{i=2}^{D+1}c_{i}^{'}=\delta$. According to our
assumption $B_{1}c_{1}^{'}-\sum_{i=2}^{D}B_{i}c_{i}^{'}$ is minimal
for
$c^{'}_{1}=\dots=c_{D}^{'}=\frac{\delta}{D+\beta}+\frac{\epsilon}{D-1+\beta}$.
By assigning these values we get
$$(B_{1}-\sum_{i=2}^{D+1}B_{i})\frac{\delta}{D+\beta}+(B_{1}-\sum_{i=2}^{D}B_{i})\frac{\epsilon}{D-1+\beta}+B_{D+1}\epsilon$$
which is greater than \eqref{eq:LemmaInductionAssumption}. This
concludes the proof.
\end{proof}
For the case $\big((M-B)(N-B)-\beta(N+M-1)\big)>B(B-1)$, the
optimization problem coincides with Lemma
\ref{Lem:LemmaOfTheOptimizationProblem} as it fulfils the condition $B_{1}>\sum_{i=2}^{D}B_{i}$ in the lemma. Hence, the optimization
problem solution for $K<\frac{MN}{N+M-1}$ is $\alpha_{1}=\dots
=\alpha_{L-1}=\frac{K-r-\alpha_{L}}{K-1}=\alpha$. The minimum is
achieved when $\alpha_{L}=\alpha$, i.e. the maximal value
$\alpha_{L}$ can receive under the constraint $\alpha_{1}\ge\dots\ge\alpha_{L}\ge 0$. We get that $\alpha=1-\frac{r}{K}$, and the
optimization problem solution of
\eqref{eq:DiversityOrderUpperBoundOptimizationProblem} for the case $K<\frac{MN}{M+N-1}$ is
$d_{KT}(r)\le MN(1-\frac{r}{K})$, .

For the case $\big((M-B)(N-B)-\beta(N+M-1)\big)\le B(B-1)$, or
equivalently $K\ge\frac{MN}{N+M-1}$, we would like to show that the
optimal solution must fulfil $\alpha_{L}=0$. It results from the
fact that for the optimal solution, the term
$\big((M-B)(N-B)-\beta(N+M-1)\big)\alpha_{L-B}-\sum_{i=1}^{B-1}2i\cdot\alpha_{L-i}$
in \eqref{Eq:EquivalenOptimizationProblem} must be negative. This is
due to the fact that taking $\alpha_{1}=\dots=\alpha_{L-1}$ gives
negative value. Hence, for the optimal solution we would like to
maximize
$\sum_{i=1}^{B-1}\alpha_{L-i}-\beta\alpha_{L-B}=K-r-\alpha_{L}$. By
taking $\alpha_{L}=0$ the sum is maximized. Hence, the optimal
solution for $K\ge\frac{MN}{M+N-1}$ must have $\alpha_{L}=0$.

Now consider the general case. Assume that for
$K\ge\frac{(M-l+1)(N-l+1)}{N+M-1-2(l-1)}+l-1$ the optimal solution
must have $\alpha_{L}=\dots=\alpha_{L-l+1}=0$. First consider the
case where $1\le l\le B-1$. For this case the constraint is
$\sum_{i=l}^{B-1}\alpha_{L-i}+\beta\alpha_{L-B}=K-r$, i.e. the
constraint contains at least two singular values. We can rewrite
\eqref{eq:DiversityOrderUpperBoundOptimizationProblem} as follows
\ifthenelse{\equal{\singlecolumntype}{1}}
{\begin{equation}\label{Eq:EquivalenOptimizationProblemGeneralCase}
\min_{\udl{\alpha}>0}(K-r)(N+M-1-2\cdot l)+\big((M-B)(N-B)
-\beta(N+M-1-2\cdot
l)\big)\alpha_{L-B}-\sum_{i=l+1}^{B-1}2(i-l)\cdot\alpha_{L-i}.
\end{equation}}
{\begin{align}\label{Eq:EquivalenOptimizationProblemGeneralCase}
\min_{\udl{\alpha}>0}&(K-r)(N+M-1-2\cdot l)+\big((M-B)(N-B)
\nonumber\\
&-\beta(N+M-1-2\cdot
l)\big)\alpha_{L-B}-\sum_{i=l+1}^{B-1}2(i-l)\cdot\alpha_{L-i}.
\end{align}}
For the case $\big((M-B)(N-B)-\beta(N+M-1-2\cdot
l)\big)>(B-1-l)(B-l)$ we get that $K<\frac{(M-l)(N-l)}{N+M-1-2\cdot
l}+l$ and we also assumed that
$K\ge\frac{(M-l+1)(N-l+1)}{N+M-1-2(l-1)}+l-1$. For this case we can
use Lemma \ref{Lem:LemmaOfTheOptimizationProblem} and get that the
optimization problem solution is
$\alpha_{L-l-1}=\dots=\alpha_{L-B}=\frac{K-r-\alpha_{L-l}}{K-l-1}=\alpha$.
The minimum is achieved for $\alpha_{L-l}=\alpha$. We get that
$\alpha_{L}=\dots=\alpha_{L-l+1}=0$ and
$\alpha_{1}=\dots=\alpha_{L-l}=\frac{K-r}{K-l}$. Hence, for the case
$\frac{(M-l+1)(N-l+1)}{N+M-1-2(l-1)}+l-1\le
K<\frac{(M-l)(N-l)}{N+M-1-2\cdot l}+l$ the solution is $d_{KT}(r)\le
(N-l)(M-l)\frac{K-r}{K-l}$.

For the case $\big((M-B)(N-B)-\beta(N+M-1-2\cdot l)\big)\le
(B-1-l)(B-l)$, or equivalently $K\ge\frac{(M-l)(N-l)}{N+M-1-2\cdot
l}+l$, the term $\big((M-B)(N-B)-\beta(N+M-1-2\cdot
l)\big)\alpha_{L-B}-\sum_{i=l+1}^{B-1}2(i-l)\cdot\alpha_{L-i}$ in
\eqref{Eq:EquivalenOptimizationProblemGeneralCase} must be negative
for the optimal solution. This is due to the fact that by taking
$\alpha_{1}=\dots=\alpha_{L-l-1}$ we get a negative value. Hence we
would like to maximize the sum
$\sum_{i=l+1}^{B-1}\alpha_{L-i}+\beta\alpha_{L-B}=K-r-\alpha_{L-l}$.
The sum is maximized by taking $\alpha_{L-l}=0$. Hence the optimal
solution for the case $K\ge\frac{(M-l)(N-l)}{N+M-1-2\cdot l}+l$ must
have $\alpha_{L-l}=\dots=\alpha_{L}=0$. Note that for the case
$l=B-1$ we have only two terms in the constraint
$\alpha_{L-B+1}+\beta\alpha_{L-B}=K-r$. However, the solution
remains the same.

For the case $K\ge\frac{(M-l+1)(N-l+1)}{N+M-1-2(l-1)}+l-1$ and $l=B$
the constraint is of the form $\alpha_{L-B}=\frac{K-r}{K-l}$. Again
we assume that $\alpha_{L-B+1}=\dots=\alpha_{L}=0$. In this case the
solution is $\alpha_{1}=\dots=\alpha_{L-l}=\frac{K-r}{K-l}$ and so
$d_{KT}(r)\le (M-l)(N-l)\frac{K-r}{K-l}$. This concludes the proof.

\section{Proof of Lemma \ref{Lem:LemmaOfTheOccurrencesOfEachColumn}}\label{Append:NumberEachColumnOccurrences}
We begin by proving the case $N\ge M$. From the construction of
$G_{l}$ it can be seen that a set of columns
$\{\udl{h}_{j},\dots,\udl{h}_{i}\}$ may occur in $N-i+j$ blocks at
most. It results from the fact that we can only subtract $M-i$
columns to the right of $\udl{h}_{i}$
\eqref{eq:HeffBlocksForTheCaseNgeM2}, and $j-1$ columns to the left
of $\udl{h}_{j}$ \eqref{eq:HeffBlocksForTheCaseNgeM3}, and still get
a block that contains $\{\udl{h}_{j},\dots,\udl{h}_{i}\}$ (or even
more specifically a block that contains
$\{\udl{h}_{j},\udl{h}_{i}\}$). In addition, columns
$\{\udl{h}_{j},\dots,\udl{h}_{i}\}$  must occur in the first $N-M+1$
blocks, as these blocks equal to $H$
\eqref{eq:HeffBlocksForTheCaseNgeM1}. Hence, we can upper bound the
number of occurrences by $N-M+1+j-1+M-i=N-i+j$.

Next we prove the case  $M>N$. When $0\le i-j<N$, the set of columns
$\{\udl{h}_{j},\dots,\udl{h}_{i}\}$ may occur in $N-i+j$ blocks at
most. We divide the proof into four cases.
\begin{enumerate}
\item $i\le N$ and $j\ge M-N+1$. In this case the set of columns $\{\udl{h}_{j},\dots,\udl{h}_{i}\}$ occurs in the first $M-N+1$ blocks \eqref{eq:HeffBlocksForTheCaseMgN1}. As for the additional $N-1-l$ pairs of columns, the set of columns belongs both to the set $\{\udl{h}_{1},\dots,\udl{h}_{N}\}$ and $\{\udl{h}_{M-N+1},\dots,\udl{h}_{M}\}$. Hence, in the additional column pairs we can subtract $N-i$ columns to the right of $\udl{h}_{i}$ \eqref{eq:HeffBlocksForTheCaseMgN2} and $j-M+N-1$ columns to the left of $\udl{h}_{j}$ \eqref{eq:HeffBlocksForTheCaseMgN3}. Added together we observe that the number of occurrences can not exceed $N-i+j$.
\item $i\le N$ and $j<M-N+1$. In this case the set of columns can have only $j$ occurrences in the first $M-N+1$ blocks. In this case the set $\{\udl{h}_{j},\dots,\udl{h}_{i}\}$ occurs within $\{\udl{h}_{1},\dots,\udl{h}_{N}\}$ but does not occur within $\{\udl{h}_{M-N+1},\dots,\udl{h}_{M}\}$. Hence, the transmission scheme only subtracts columns to the right of $\udl{h}_{i}$ \eqref{eq:HeffBlocksForTheCaseMgN2}. In this case we can have $N-i$ subtractions and together we get $N-i+j$ occurrences at most.
\item $i>N$ and $j\ge M-N+1$. We have here $M-i+1$ occurrences in the first $M-N+1$ blocks. In this case the set $\{\udl{h}_{j},\dots,\udl{h}_{i}\}$ occurs within $\{\udl{h}_{M-N+1},\dots,\udl{h}_{M}\}$ but does not occur within $\{\udl{h}_{1},\dots,\udl{h}_{N}\}$. Hence we can subtract up to $j-M+N-1$ columns to the left of $\udl{h}_{j}$ \eqref{eq:HeffBlocksForTheCaseMgN3}. Together there are $N-i+j$ occurrences at most.
\item Last case, $i>N$ and $j<M-N+1$. Here the set of columns can only occur in the first $M-N+1$ blocks. In this case there are exactly $N-i+j$ occurrences in the first $M-N+1$ blocks.
\end{enumerate}
In case $i-j\ge N$, the set of columns does not occur in any block
as each column of $G_{l}$ does not have more than $N$ non-zero
entries.

\section{Proof of Theorem \ref{Th:UpperBoundErrorProb}}\label{Append:UpperBoundErrorProb}
Based on \cite{PoltirevJournal} we have the following upper bound on
the maximum-likelihood (ML) decoding error probability of each
$K_{l}T_{l}$-complex dimensional IC point $\udl{x}^{'}\in
S_{K_{l}T_{l}}$ \ifthenelse{\equal{\singlecolumntype}{1}}
{\begin{equation}\label{eq:MLPointErrorProbUpperBound}
P_{e}(\underline{x}^{'})\le
Pr(\lv\underline{\tilde{n}}_{\mathrm{ex}}\rv\ge R)+
\sum_{\underline{l}\in Ball(\underline{x}^{'},2R)\bigcap
S_{K_{l}T_{l}}, \udl{l}\ne \udl{x}^{'}} Pr(\lv
\udl{l}-\udl{x}^{'}-\udl{\tilde{n}}_{ex}\rv <
\lv\udl{\tilde{n}}_{ex}\rv)
\end{equation}}
{\begin{align}\label{eq:MLPointErrorProbUpperBound}
&P_{e}(\underline{x}^{'})\le
Pr(\lv\underline{\tilde{n}}_{\mathrm{ex}}\rv\ge R)+
\nonumber\\
&\sum_{\underline{l}\in Ball(\underline{x}^{'},2R)\bigcap
S_{K_{l}T_{l}}, \udl{l}\ne \udl{x}^{'}} Pr(\lv
\udl{l}-\udl{x}^{'}-\udl{\tilde{n}}_{ex}\rv <
\lv\udl{\tilde{n}}_{ex}\rv)
\end{align}}
where $Ball(\udl{x}^{'},2R)$ is a $K_{l}T_{l}$-complex dimensional
ball of radius $2R$ centered around $\udl{x}^{'}$, and
$\udl{\tilde{n}}_{\ex}$ is the effective noise in the
$K_{l}T_{l}$-complex dimensional hyperplane where the IC's resides.
Note that the second term in \eqref{eq:MLPointErrorProbUpperBound}
represents the pairwise error probability to points within
$Ball(\udl{x}^{'},2R)$, i.e. the decision region is at distance $R$
at most.

Next we upper bound the average decoding error probability of an
ensemble of constellations drawn uniformly within
$cube_{K_{l}T_{l}}(b)$. Each code-book contains $\lfloor
\gamma_{\mathrm{tr}}b^{2K_{l}T_{l}}\rfloor$ points, where each point
is drawn uniformly within $cube_{K_{l}T_{l}}(b)$. In the receiver,
the random ensemble is uniformly distributed within
$\{H_{\eff}^{(l)}\cdot\cube_{K_{l}T_{l}}(b)\}$. Let us consider a
certain point,
$\udl{x}^{'}\in\{H_{\eff}^{(l)}\cdot\cube_{K_{l}T_{l}}(b)\}$, from
the random ensemble in the receiver. We denote the ring around
$\udl{x}^{'}$ by
$Ring(\underline{x}^{'},i\Delta)=Ball(\underline{x}^{'},i\Delta)\setminus
Ball(\underline{x}^{'},(i-1)\Delta)$. The average number of points
within $Ring(\underline{x}^{'},i\Delta)$ of the random ensemble is
\ifthenelse{\equal{\singlecolumntype}{1}}
{\begin{equation}\label{eq:AverageNumberPoints}
Av(\underline{x}^{'},i\Delta) =
\gamma_{\mathrm{rc}}|H_{\eff}^{(l)}\cdot\cube_{K_{l}T_{l}}(b)\bigcap
Ring(\underline{x}^{'},i\Delta)| \le
\gamma_{\mathrm{rc}}|Ring(\underline{x}^{'},i\Delta)|\le
\frac{\gamma_{\mathrm{rc}}\pi^{K_{l}T_{l}}2K_{l}T_{l}}{\Gamma(K_{l}T_{l}+1)}(i\Delta)^{2K_{l}T_{l}-1}\Delta
\end{equation}}
{\begin{align}\label{eq:AverageNumberPoints}
Av(\underline{x}^{'},i\Delta) =
\gamma_{\mathrm{rc}}|H_{\eff}^{(l)}\cdot\cube_{K_{l}T_{l}}(b)\bigcap
Ring(\underline{x}^{'},i\Delta)|\nonumber \\
\le \gamma_{\mathrm{rc}}|Ring(\underline{x}^{'},i\Delta)|\le
\frac{\gamma_{\mathrm{rc}}\pi^{K_{l}T_{l}}2K_{l}T_{l}}{\Gamma(K_{l}T_{l}+1)}(i\Delta)^{2K_{l}T_{l}-1}\Delta
\end{align}}
where
$\gamma_{\mathrm{rc}}=\rho^{rT_{l}+\sum_{i=1}^{K_{l}T_{l}}\eta_{i}}$.
By using the upper bounds on the error probability
\eqref{eq:MLPointErrorProbUpperBound}, and the average number of
points within the rings \eqref{eq:AverageNumberPoints}, we get for a
certain channel realization the following upper bound on the average
decoding error probability of the finite constellations ensemble, at
point $\udl{x}^{'}$ \ifthenelse{\equal{\singlecolumntype}{1}}
{\begin{equation} \ol{P_{e}^{FC}}(\udl{x}^{'},\rho,\udl{\eta})\le
Pr(\lv\underline{\tilde{n}}_{\mathrm{ex}}\rv\ge R)+
\gamma_{\mathrm{rc}}Q(K_{l}T_{l})\sum_{i=1}^{\lceil
\frac{2R}{\Delta}\rceil}Pr(\tilde{n}_{ex,1}>\frac{(i-1)\Delta}{2})
\cdot(i\Delta)^{2K_{l}T_{l}-1}\Delta
\end{equation}}
{\begin{align} \ol{P_{e}^{FC}}(\udl{x}^{'},\rho,\udl{\eta})&\le
Pr(\lv\underline{\tilde{n}}_{\mathrm{ex}}\rv\ge R)+
\nonumber\\
\gamma_{\mathrm{rc}}Q(K_{l}T_{l})&\sum_{i=1}^{\lceil
\frac{2R}{\Delta}\rceil}Pr(\tilde{n}_{ex,1}>\frac{(i-1)\Delta}{2})
\cdot(i\Delta)^{2K_{l}T_{l}-1}\Delta
\end{align}}
where
$Q(K_{l}T_{l})=\frac{\pi^{K_{l}T_{l}}2K_{l}T_{l}}{\Gamma(K_{l}T_{l}+1)}$,
and $\tilde{n}_{ex,1}$ is the first component of
$\udl{\tilde{n}}_{ex}$ (the pairwise error probability has scalar
decision region). By taking $\Delta\to 0$ we get
\ifthenelse{\equal{\singlecolumntype}{1}}
{\begin{equation}\label{eq:IntErrorUpperBound}
\ol{P_{e}^{FC}}(\udl{x}^{'},\rho,\udl{\eta})\le
Pr(\lv\underline{\tilde{n}}_{\mathrm{ex}}\rv\ge R)+
\gamma_{\mathrm{rc}}Q(K_{l}T_{l})\int_{0}^{2R}
Pr(\tilde{n}_{ex,1}>\frac{x}{2})x^{2K_{l}T_{l}-1}dx.
\end{equation}}
{\begin{align}\label{eq:IntErrorUpperBound}
\ol{P_{e}^{FC}}(\udl{x}^{'},\rho,\udl{\eta})&\le
Pr(\lv\underline{\tilde{n}}_{\mathrm{ex}}\rv\ge R)+
\nonumber\\
\gamma_{\mathrm{rc}}Q(K_{l}T_{l})&\int_{0}^{2R}
Pr(\tilde{n}_{ex,1}>\frac{x}{2})x^{2K_{l}T_{l}-1}dx.
\end{align}}
Note that this upper bound applies for any value of $R\ge 0$ and
$b$, and does not depend on  $\udl{x}^{'}$, i.e.
$\ol{P_{e}^{FC}}(\udl{x}^{'},\rho,\udl{\eta})=\ol{P_{e}^{FC}}(\rho,\udl{\eta})$.

Now we divide the channel realization into two subsets:
$\mathcal{A}=\{\udl{\eta}\mid\sum_{i=1}^{K_{l}T_{l}}\eta_{i}\le
T_{l}(K_{l}-r), \eta_{i}\ge 0\}$, where
$\udl{\eta}=(\eta_{1},\dots,\eta_{K_{l}T_{l}})$ and
$\ol{\mathcal{A}}=\{\udl{\eta}\mid\sum_{i=1}^{K_{l}T_{l}}\eta_{i}>
T_{l}(K_{l}-r), \eta_{i}\ge 0\}$. For each set we upper bound the
error probability. We begin with the case $\udl{\eta}\in
\mathcal{A}$. For this case we upper bound the terms in
\eqref{eq:IntErrorUpperBound} and find an upper bound on the error
probability as a function of the receiver VNR,
$\mu_{rc}=\rho^{1-\frac{r}{K_{l}}-\frac{\sum_{i=1}^{K_{l}T_{l}}{\eta_{i}}}{K_{l}T_{l}}}$.
We begin by upper bounding the integral of the second term in
\eqref{eq:IntErrorUpperBound}. Note that
\begin{equation}
Pr(\tilde{n}_{\mathrm{ex},1}\ge \frac{x}{2})\le
e^{-\frac{x^{2}}{8\sigma^{2}}}. \nonumber
\end{equation}
Hence, the integral in the second term in
\eqref{eq:IntErrorUpperBound} can be upper bounded by
$$\sigma^{2K_{l}T_{l}}\Gamma(K_{l}T_{l})2^{3K_{l}T_{l}-2}\int_{0}^{2R}\frac{e^{-\frac{x^{2}}{8\sigma^{2}}}x^{2K_{l}T_{l}-1}}
{\sigma^{2K_{l}T_{l}}\Gamma(K_{l}T_{l})2^{3K_{l}T_{l}-2}}dx$$ where
$\int_{0}^{2R}\frac{e^{-\frac{x^{2}}{8\sigma^{2}}}x^{2K_{l}T_{l}-1}}
{\sigma^{2K_{l}T_{l}}\Gamma(K_{l}T_{l})2^{3K_{l}T_{l}-2}}dx=Pr(\lv
\tilde{n}_{\mathrm{ex}}\rv\le 2R)\le 1$. As a result we get the
following upper bound
\begin{equation}
\int_{0}^{2R}
Pr(\tilde{n}_{ex,1}>\frac{x}{2})x^{2K_{l}T_{l}-1}dx\le\sigma^{2K_{l}T_{l}}\Gamma(K_{l}T_{l})2^{3K_{l}T_{l}-2}.
\end{equation}
By assigning this upper bound in the second term of
\eqref{eq:IntErrorUpperBound} we get
\ifthenelse{\equal{\singlecolumntype}{1}}
{\begin{equation}\label{eq:UpperBoundSecondTerm}
\gamma_{\mathrm{rc}}Q(K_{l}T_{l})\int_{0}^{2R}
Pr(\tilde{n}_{ex,1}>\frac{x}{2})x^{2K_{l}T_{l}-1}dx
\le\frac{\gamma_{\mathrm{rc}}\sqrt{\pi}^{2K_{l}T_{l}}2K_{l}T_{l}\sigma^{2K_{l}T_{l}}\Gamma(K_{l}T_{l})2^{3K_{l}T_{l}-2}}{\Gamma(K_{l}T_{l}+1)}
=\rho^{-T_{l}(K_{l}-r)+\sum_{i=1}^{K_{l}T_{l}}\eta_{i}}\cdot\frac{4^{K_{l}T_{l}}}{2e^{K_{l}T_{l}}}.
\end{equation}}
{\begin{align}\label{eq:UpperBoundSecondTerm}
&\gamma_{\mathrm{rc}}Q(K_{l}T_{l})\int_{0}^{2R}
Pr(\tilde{n}_{ex,1}>\frac{x}{2})x^{2K_{l}T_{l}-1}dx\nonumber\\
&\le\frac{\gamma_{\mathrm{rc}}\sqrt{\pi}^{2K_{l}T_{l}}2K_{l}T_{l}\sigma^{2K_{l}T_{l}}\Gamma(K_{l}T_{l})2^{3K_{l}T_{l}-2}}{\Gamma(K_{l}T_{l}+1)}\nonumber\\
&=\rho^{-T_{l}(K_{l}-r)+\sum_{i=1}^{K_{l}T_{l}}\eta_{i}}\cdot\frac{4^{K_{l}T_{l}}}{2e^{K_{l}T_{l}}}.
\end{align}}
Next we upper bound $Pr(\lv \tilde{n}_{\mathrm{ex}}\rv\ge R)$, the
first term in \eqref{eq:IntErrorUpperBound}. We choose
$$R^{2}=R_{\eff}^{2}=\frac{2K_{l}T_{l}}{{2\pi
e}}\gamma_{rc}^{-\frac{1}{K_{l}T_{l}}}=\frac{2K_{l}T_{l}}{2\pi
e}\rho^{-\frac{r}{K_{l}}-\sum_{i=1}^{K_{l}T_{l}}\frac{\eta_{i}}{K_{l}T_{l}}}.$$
For $\udl{\eta}\in \mathcal{A}$ we get that
$$\frac{R_{\eff}^{2}}{2K_{l}T_{l}\cdot\sigma^{2}}=\rho^{1-\frac{r}{K_{l}}-\sum_{i=1}^{K_{l}T_{l}}\frac{\eta_{i}}{K_{l}T_{l}}}\ge
1.$$ By using the upper bounds from \cite{PoltirevJournal}, we know
that for the case
$\frac{R_{\eff}^{2}}{2K_{l}T_{l}\cdot\sigma^{2}}\ge 1$, $Pr(\lv
\tilde{n}_{\mathrm{ex}}\rv\ge R_{\eff})\le
e^{-\frac{R_{\eff}^{2}}{2\sigma^{2}}}(\frac{R_{\eff}^{2}e}{2K_{l}T_{l}\sigma^{2}})^{K_{l}T_{l}}$.
Hence we get \ifthenelse{\equal{\singlecolumntype}{1}}
{\begin{equation} \label{eq:UpperBoundFirstTerm} Pr(\lv
\tilde{n}_{\mathrm{ex}}\rv\ge R_{\eff})\le
e^{-K_{l}T_{l}\rho^{1-\frac{r}{K_{l}}-\sum_{i=1}^{K_{l}T_{l}}\frac{\eta_{i}}{K_{l}T_{l}}}}\cdot
\rho^{T_{l}(K_{l}-r)-\sum_{i=1}^{K_{l}T_{l}}\eta_{i}}\cdot
e^{K_{l}T_{l}}.
\end{equation}}
{\begin{align}\label{eq:UpperBoundFirstTerm}
&Pr(\lv \tilde{n}_{\mathrm{ex}}\rv\ge R_{\eff})\le\nonumber\\
&e^{-K_{l}T_{l}\rho^{1-\frac{r}{K_{l}}-\sum_{i=1}^{K_{l}T_{l}}\frac{\eta_{i}}{K_{l}T_{l}}}}\cdot
\rho^{T_{l}(K_{l}-r)-\sum_{i=1}^{K_{l}T_{l}}\eta_{i}}\cdot
e^{K_{l}T_{l}}.
\end{align}}
The fact that $\udl{\eta}\in \mathcal{A}$ has two significant
consequences: the VNR is greater or equal to 1, and as $\rho$
increases the maximal VNR in the set also increases. For very large
VNR in the receiver, the upper bound of the first term,
\eqref{eq:UpperBoundFirstTerm}, is negligible compared to the upper
bound on the second term, \eqref{eq:UpperBoundSecondTerm}. On the
other hand, the set of rather small VNR values is fixed for
increasing $\rho$ (the VNR is grater or equal to 1). Hence there
must exist a coefficient $D^{'}(K_{l}T_{l})$ that gives us
\begin{equation}\label{eq:FinalUpperBoundForFiniteConstellation}
\ol{P_{e}^{FC}}(\rho,\udl{\eta})\le
D^{'}(K_{l}T_{l})\rho^{-T_{l}(K_{l}-r)+\sum_{i=1}^{K_{l}T_{l}}\eta_{i}}
\end{equation}
for any $\rho$ and $\udl{\eta}\in \mathcal{A}$, where
$\ol{P_{e}^{FC}}(\rho,\udl{\eta})$ is the average decoding error
probability of the ensemble of constellations, for a certain channel
realizations.

Note that we could also take $R\ge R_{\eff}$, as the upper bound in
\eqref{eq:UpperBoundSecondTerm} does not depend on $R$ and the upper
bound in \eqref{eq:UpperBoundFirstTerm} would only decrease in this
case. It results from the fact that we are interested in the
exponential behavior of the error probability, and we consider a
fixed VNR (as a function of $\rho$) as an outage event. This allows
us to take cruder bounds than \cite{PoltirevJournal} in
\eqref{eq:UpperBoundSecondTerm}, that do not depend on $R$.

For the case $\udl{\eta}\in\ol{\mathcal{A}}$, we get
$$\rho^{-T_{l}(K_{l}-r)+\sum_{i=1}^{K_{l}T_{l}}\eta_{i}}\ge 1.$$
Hence, we can upper bound the error probability for
$\udl{\eta}\in\ol{\mathcal{A}}$ by 1. We can also upper bound the
error probability for this case by the upper bound from equation
\eqref{eq:FinalUpperBoundForFiniteConstellation}, as long as we
state that $D^{'}(K_{l}T)\ge 1$. Hence, the upper bound from
\eqref{eq:FinalUpperBoundForFiniteConstellation} applies for
$\eta_{i}\ge 0$, $1\le i\le K_{l}T_{l}$.

So far we upper bounded the average decoding error probability of
the ensemble of finite constellations. We extend now these finite
constellations into an ensemble of IC's with density $\gamma_{tr}$,
and show that the upper bound on the average decoding error
probability does not change. Let us consider a certain finite
constellation, $C_{0}(\rho,b)\subset cube_{K_{l}T_{l}}(b)$, from the
random ensemble. We extend it into IC

\begin{equation}\label{eq:ExtendingFinteConstoICInTheTransmitter}
IC(\rho,K_{l}T_{l})=C_{0}(\rho,b)+(b+b^{'})\cdot\mathbb{Z}^{2K_{l}T_{l}}
\end{equation}
where without loss of generality we assumed that
$cube_{K_{l}T_{l}}(b)\in\mathbb{C}^{K_{l}T_{l}}$. In the receiver we
have
\begin{equation}\label{eq:ExtendingICInTheReceiver}
IC(\rho,K_{l}T_{l},H_{\eff}^{(l)})=H_{\eff}^{(l)}\cdot
C_{0}(\rho,b)+(b+b^{'})H_{\eff}^{(l)}\cdot\mathbb{Z}^{2K_{l}T_{l}}.
\end{equation}
By extending each finite constellation in the ensemble into an IC
according to the method presented in
\eqref{eq:ExtendingFinteConstoICInTheTransmitter}, we get a new
ensemble of IC's. We would like to set $b$ and $b^{'}$ to be large
enough such that the IC's ensemble average decoding error
probability has the same upper bound as in
\eqref{eq:FinalUpperBoundForFiniteConstellation}, and a density that
equals $\gamma_{rc}$ up to a coefficient. First we would like to set
a value for $b^{'}$. Increasing $b^{'}$ decreases the error
probability inflicted by the codewords outside the set
$\{H_{\eff}^{(l)}\cdot C_{0}(\rho,b)\}$. Without loss of generality,
we upper bound the error probability of the points $x\in
\{H_{\eff}^{(l)}\cdot C_{0}(\rho,b)\}\subset
IC(\rho,K_{l}T_{l},H_{\eff}^{(l)})$, denoted by
$P_{e}^{IC}(H_{\eff}^{(l)}\cdot C_{0})$. Due to the tiling symmetry,
$P_{e}^{IC}(H_{\eff}^{(l)}\cdot C_{0})$ is also the average decoding
error probability of the entire IC. We begin with $\udl{\eta}\in
\mathcal{A}$. For this case, we upper bound the IC error probability
in the following manner
$$P_{e}^{IC}(H_{\eff}^{(l)}\cdot C_{0})\le P_{e}^{FC}(H_{\eff}^{(l)}\cdot C_{0})+P_{e}\big(H_{\eff}^{(l)}\cdot(IC\setminus C_{0})\big)$$
where $P_{e}^{FC}(H_{\eff}^{(l)}\cdot C_{0})$ is the error
probability of the finite constellation $\{H_{\eff}^{(l)}\cdot
C_{0}\}$, and $P_{e}\big(H_{\eff}^{(l)}\cdot(IC\setminus
C_{0})\big)$ is the average decoding error probability to points in
the set $\{H_{\eff}^{(l)}\cdot(IC\setminus C_{0})\}$. For the case
$\udl{\eta}\in \mathcal{A}$, we know that $0\le \eta_{i}\le
T_{l}(K_{l}-r)$. Hence, the constriction caused by the channel in
each dimension can not be smaller than
$\rho^{-\frac{T_{l}}{2}(K_{l}-r)}$. As a result, for any $x_{1}\in
\{H_{\eff}^{(l)}\cdot C_{0}\}$ and $x_{2}\in
\{H_{\eff}^{(l)}\cdot(IC\setminus C_{0})\}$ we get $\lVert
\underline{x}_{1}-\underline{x}_{2}\rVert\ge
2b^{'}\cdot\rho^{-\frac{T_{l}}{2}(K_{l}-r)}$. By choosing
$b^{'}=\sqrt{\frac{K_{l}T_{l}}{\pi
e}}\rho^{\frac{T_{l}}{2}(K_{l}-r)+\epsilon}$, we get for
$\udl{\eta}\in \mathcal{A}$ that $\lVert
\underline{x}_{1}-\underline{x}_{2}\rVert\ge
2\sqrt{\frac{K_{l}T_{l}}{\pi e}}\rho^{\epsilon}$. Hence we get
$$P_{e}\big(H_{\eff}^{(l)}\cdot(IC\setminus C_{0})\big)\le Pr(\lVert \underline{\tilde{n}}_{\mathrm{ex}}\rVert\ge \sqrt{\frac{K_{l}T_{l}}{\pi e}}\rho^{\epsilon}).$$ For $\rho\ge 1$ we get according to the bounds in \cite{PoltirevJournal} that
$$Pr(\lVert \underline{\tilde{n}}_{\mathrm{ex}}\rVert\ge \sqrt{\frac{K_{l}T_{l}}{\pi e}}\rho^{\epsilon}))\le
e^{-K_{l}T_{l}\rho^{1+\epsilon}}\rho^{K_{l}T_{l}(1+\epsilon)}e^{K_{l}T_{l}}.$$
As a result, there exists a coefficient $D^{''}(K_{l}T_{l})$ such
that
$$P_{e}\big(H_{\eff}^{(l)}\cdot(IC\setminus C_{0})\big)\le
D^{''}(K_{l}T_{l})\rho^{-T_{l}(K_{l}-r)+\sum_{i=1}^{K_{l}T_{l}}\eta_{i}}$$
for $\udl{\eta}\in \mathcal{A}$ and $\rho\ge 1$. This bound applies
for any IC in the ensemble. From
\eqref{eq:FinalUpperBoundForFiniteConstellation} we can state that
$\ol{P_{e}^{FC}}(\rho,\udl{\eta})=E_{C_{0}}\big(P_{e}^{FC}(H_{\eff}^{(l)}\cdot
C_{0})\big)\le
D^{'}(K_{l}T_{l})\rho^{-T_{l}(K_{l}-r)+\sum_{i=1}^{K_{l}T_{l}}\eta_{i}}$.
Hence
\begin{equation}\label{eq:ICUpperBoundErrorProbForAllCases}
\ol{P_{e}}(\rho,\udl{\eta})\le
D(K_{l}T_{l})\rho^{-T_{l}(K_{l}-r)+\sum_{i=1}^{K_{l}T_{l}}\eta_{i}}
\end{equation}
where
$\ol{P_{e}}(\rho,\udl{\eta})=E_{C_{0}}\big(P_{e}^{IC}(H_{\eff}^{(l)}\cdot
C_{0})\big)$ is the average decoding error probability of the
ensemble of IC's defined in \eqref{eq:ExtendingICInTheReceiver}, and
$D=2\max(D^{'},D^{''})>1$.

Next, we set the value of $b$ to be large enough such that each IC
density from the ensemble in \eqref{eq:ExtendingICInTheReceiver},
$\gamma_{rc}^{'}$, equals $\gamma_{rc}$ up to a factor of 2. By
choosing $b=b^{'}\cdot\rho^{\epsilon}$ we get
$$\gamma_{rc}^{'}=\gamma_{rc}(\frac{b}{b+b^{'}})^{2K_{l}T}=\gamma_{rc}\frac{1}{1+\rho^{-\epsilon}}.$$
For each value $\rho\ge 1$, we get
$\frac{1}{2}\gamma_{rc}\le\gamma_{rc}^{'}\le\gamma_{rc}$. As a
result we have
$$\mu_{rc}\le\mu_{rc}^{'}=\frac{(\gamma_{rc}^{'})^{-\frac{1}{K_{l}T}}}{2\pi e\sigma^{2}}\le
2\mu_{rc}.$$ Note that in our proof we referred to a matrix of
dimension $NT_{l}\times K_{l}T_{l}$. However these results apply for
any full rank matrix with number of rows which is greater or equal
to the number of columns.

\section{Proof of theorem \ref{Th:LowerBoundDiversityOrder}}\label{Append:LowerBoundDiversityOrder}
Specifically, we first lower bound the contribution of $\udl{h}_{j}$
to the determinant \eqref{eq:TheContributionOfhjIntermsofbetta}, by
upper bounding $\sum_{k=0}^{\min(j,L)-1}b_{j}(k)a(k,\udl{\xi}_{j})$.
Based on Lemma \ref{Lem:LemmaOfTheOccurrencesOfEachColumn}, and the
fact that when two columns of $H$ occur together in a block of
$H_{\eff}^{(l)}$, all the columns of $H$ between them must also
occur in the same block, we get
\begin{equation}\label{eq:InequalityAsResuktFromLemmaOfNumberColumns}
\sum_{s=k}^{\min(j,L)-1}b_{j}(s)\le N-k\qquad 0\le k\le \min(j,L)-1.
\end{equation}
where $\sum_{s=k}^{\min(j,L)-1}b_{j}(s)$ is the number of
occurrences of $\{\udl{h}_{j},\dots,\udl{h}_{j-k}\}$ in the blocks
of $H_{\eff}^{(l)}$. Hence, we can state that
\ifthenelse{\equal{\singlecolumntype}{1}}{$$\sum_{s=0}^{\min(j,L)-1}b_{j}(s)\le
N$$} {$\sum_{s=0}^{\min(j,L)-1}b_{j}(s)\le N,$} by assigning $k=0$
in \eqref{eq:InequalityAsResuktFromLemmaOfNumberColumns}. Also note
that for $l=0$, the sum
$\sum_{s=0}^{\min(j,L)-1}b_{j}(s)a(s,\udl{\xi}_{j})$ is larger than
for any other $1\le l\le L-1$. From the inequalities in
\eqref{eq:InequalityOfTheNumberOfContributions}, and the fact that
for $l=0$ we get $b_{j}(k)>0$ for any $1\le k\le\min(j,L)-1$, we can
state that \ifthenelse{\equal{\singlecolumntype}{1}}
{\begin{equation}\label{UpperBoundOnhjContExponent}
\sum_{s=0}^{\min(j,L)-1}b_{j}(s)a(s,\udl{\xi}_{j})\le\sum_{s=0}^{\min(j,L)-2}a(s,\udl{\xi}_{j})
+(N-\min(j,L)+1)a(\min(j,L)-1,\udl{\xi}_{j})=c(j).
\end{equation}}
{\begin{align}\label{UpperBoundOnhjContExponent}
\sum_{s=0}^{\min(j,L)-1}b_{j}(s)a(s,\udl{\xi}_{j})\le\sum_{s=0}^{\min(j,L)-2}a(s,\udl{\xi}_{j})\nonumber\\
+(N-\min(j,L)+1)a(\min(j,L)-1,\udl{\xi}_{j})=c(j).
\end{align}}
Using \eqref{eq:TheContributionOfhjIntermsofbetta} and
\eqref{UpperBoundOnhjContExponent} we can state that for a vector
$\udl{\xi}_{j}$, whose PDF is proportional to
$\rho^{-\sum_{i=1}^{N}\xi_{i,j}}$, we can lower bound the
contribution of $\udl{h}_{j}$ to
$|H_{\eff}^{(l)\dagger}H_{\eff}^{(l)}|$ by
\begin{equation}\label{eq:LowerNoundOnTheContributionOfjj}
\lv\udl{h}_{j}\rv^{2b_{j}(0)}\prod_{k=1}^{\min(j,L)-1}\lv\udl{h}_{j\perp
j-1,\dots,j-k}\rv^{2b_{j}(k)}\ge \rho^{-c(j)}.
\end{equation}
By taking into account the contribution of each column $\udl{h}_{j}$
to the determinant we get that
\ifthenelse{\equal{\singlecolumntype}{1}} {\begin{equation}
|H_{\eff}^{(l)\dagger}H_{\eff}^{(l)}|=
\prod_{j=1}^{M}\lv\udl{h}_{j}\rv^{2b_{j}(0)}\prod_{k=1}^{\min(j,L)-1}\lv\udl{h}_{j\perp
j-1,\dots,j-k}\rv^{2b_{j}(k)}.
\end{equation}}
{\begin{align}
&|H_{\eff}^{(l)\dagger}H_{\eff}^{(l)}|=\nonumber\\
&\prod_{j=1}^{M}\lv\udl{h}_{j}\rv^{2b_{j}(0)}\prod_{k=1}^{\min(j,L)-1}\lv\udl{h}_{j\perp
j-1,\dots,j-k}\rv^{2b_{j}(k)}.
\end{align}}
By considering the set of vectors
$\udl{\xi}_{1},\dots,\udl{\xi}_{M}$, whose PDF is proportional to
$\rho^{-\sum_{j=1}^{M}\sum_{i=1}^{N}\xi_{i,j}}$, and by using the
lower bound from \eqref{eq:LowerNoundOnTheContributionOfjj} we get
\begin{equation}\label{eq:LowerNoundOnThedeterminantHeff}
|H_{\eff}^{(l)\dagger}H_{\eff}^{(l)}|\ge\rho^{-\sum_{j=1}^{M}c(j)}
\end{equation}

The upper bound on the error probability presented in Theorem
\ref{Th:UpperBoundErrorProb} is proportional to
\begin{equation}\label{eq:TheUpperBoundErrorProbWithHEntries}
\rho^{-T_{l}(K_{l}-r)}\cdot|H_{\eff}^{(l)\dagger}H_{\eff}^{(l)}|^{-1}=\rho^{-T_{l}(K_{l}-r)+\sum_{i=1}^{K_{l}T}\eta_{i}}
\end{equation}
for $\eta_{i}\ge 0$ and $1\le i\le K_{l}T_{l}$, where
$\rho^{-\frac{\eta_{i}}{2}}$ are the singular values of
$H_{\eff}^{(l)}$. Hence, in order to use the upper bound from
Theorem \ref{Th:UpperBoundErrorProb} in our analysis, we need to
show that by taking $\xi_{i,j}\ge 0$, $1\le i\le N$, $1\le j\le M$
we also get that $\eta_{i}\ge 0$, $1\le i\le K_{l}T_{l}$. Note that
the entries of $H_{\eff}^{(l)}$ are elements of the channel matrix
$H$. Also, all the columns of $H$ must appear in $H_{\eff}^{(l)}$.
Hence, from trace considerations we get
$$\frac{\rho^{-\min_{i,j}(\xi_{i,j})}}{K_{l}T_{l}}\le\rho^{-\min_{s}(\eta_{s})}\le N\cdot K_{l}T_{l}^{2}\rho^{-\min_{i,j}(\xi_{i,j})}.$$
As a result $\min_{i,j}(\xi_{i,j})\ge 0$ if and only if
$\min_{s}(\eta_{s})\ge 0$, and so $\eta_{s}\ge 0$ for every $1\le
s\le K_{l}T_{l}$. As the upper bound on the error probability in
\eqref{eq:TheUpperBoundErrorProbWithHEntries} applies for
$\eta_{i}\ge 0$, $1\le i\le K_{l}T_{l}$, this upper bound also
applies whenever $\xi_{i,j}\ge 0$, $1\le i\le N$ and $1\le j\le M$.
In equation \eqref{eq:LowerNoundOnThedeterminantHeff} we found a
lower bound on the determinant. We use this lower bound to upper
bound the determinant of the matrix inverse
$|H_{\eff}^{(l)\dagger}H_{\eff}^{(l)}|^{-1}$
\begin{equation}\label{eq:UpperBoundInverseDetByhjTerms}
|H_{\eff}^{(l)\dagger}H_{\eff}^{(l)}|^{-1}\le\rho^{\sum_{j=1}^{M}c(j)}.
\end{equation}
and as a consequence we can upper bound the error probability.

We can express the average decoding error probability over the
ensemble of IC's for large $\rho$ as follows
\ifthenelse{\equal{\singlecolumntype}{1}}
{\begin{equation}\label{eq:UpperBoundErrorProbDueExact}
\ol{P_{e}}(\rho)=\int_{H}P_{e}(\rho,H)f(H)dH\dot{=}
\int_{\xi_{\udl{i},\udl{j}}\ge
0}P_{e}(\rho,\xi_{\udl{i},\udl{j}})f(\xi_{\udl{i},\udl{j}})d\xi_{\udl{i},\udl{j}}
\end{equation}}
{\begin{align}\label{eq:UpperBoundErrorProbDueExact}
\ol{P_{e}}(\rho)=\int_{H}P_{e}(\rho,H)f(H)dH\dot{=}\nonumber\\
\int_{\xi_{\udl{i},\udl{j}}\ge
0}P_{e}(\rho,\xi_{\udl{i},\udl{j}})f(\xi_{\udl{i},\udl{j}})d\xi_{\udl{i},\udl{j}}
\end{align}}
where $P_{e}(\rho,H)=P_{e}(\rho,\xi_{\udl{i},\udl{j}})$ is the
ensemble average decoding error probability per channel realization,
and $\xi_{\udl{i},\udl{j}}\ge 0$ means $\xi_{i,j}\ge 0$ for $1\le
i\le N$ and $1\le j\le M$. We divide the integration range into two
sets:
$\mathcal{A}=\{\xi_{\udl{i},\udl{j}}\mid\sum_{i=1}^{N}\sum_{j=1}^{M}\xi_{i,j}\le
T_{l}(K_{l}-r);\xi_{\udl{i},\udl{j}}\ge 0\}$ and
$\ol{\mathcal{A}}=\{\xi_{\udl{i},\udl{j}}\mid\sum_{i=1}^{N}\sum_{j=1}^{M}\xi_{i,j}>
T_{l}(K_{l}-r);\xi_{\udl{i},\udl{j}}\ge 0\}$. Hence, we can write
the average decoding error probability as follows
\ifthenelse{\equal{\singlecolumntype}{1}}
{\begin{equation}\label{eq:UpperBoundErrorProbDueExact2Ranges}
\ol{P_{e}}(\rho)\dot{=}\int_{\xi_{\udl{i},\udl{j}}\in \mathcal{A}
}P_{e}(\rho,\xi_{\udl{i},\udl{j}})f(\xi_{\udl{i},\udl{j}})d\xi_{\udl{i},\udl{j}}+
\int_{\xi_{\udl{i},\udl{j}}\in
\ol{\mathcal{A}}}P_{e}(\rho,\xi_{\udl{i},\udl{j}})f(\xi_{\udl{i},\udl{j}})d\xi_{\udl{i},\udl{j}}.
\end{equation}}
{\begin{align}\label{eq:UpperBoundErrorProbDueExact2Ranges}
\ol{P_{e}}(\rho)\dot{=}\int_{\xi_{\udl{i},\udl{j}}\in \mathcal{A} }P_{e}(\rho,\xi_{\udl{i},\udl{j}})f(\xi_{\udl{i},\udl{j}})d\xi_{\udl{i},\udl{j}}+\nonumber\\
\int_{\xi_{\udl{i},\udl{j}}\in
\ol{\mathcal{A}}}P_{e}(\rho,\xi_{\udl{i},\udl{j}})f(\xi_{\udl{i},\udl{j}})d\xi_{\udl{i},\udl{j}}.
\end{align}}
We begin by upper bounding the first term of the error probability
in \eqref{eq:UpperBoundErrorProbDueExact2Ranges}. Based on Theorem
\ref{Th:UpperBoundErrorProb}, the average decoding error probability
per channel realization is upper bounded by
$P_{e}(\rho,H)\le\rho^{-T_{l}(K_{l}-r)+\sum_{i=1}^{K_{l}T_{l}}\eta_{i}}$.
Using the upper bound on the determinant
\eqref{eq:UpperBoundInverseDetByhjTerms} and the fact that
$|H_{\eff}^{(l)\dagger}H_{\eff}^{(l)}|^{-1}=\rho^{\sum_{i=1}^{K_{l}T_{l}}\eta_{i}}$,
we get that the first term of the error probability
\eqref{eq:UpperBoundErrorProbDueExact2Ranges} is upper bounded by
\begin{equation}\label{eq:UpperBoundAverageErrorProbabilityFirstTerm}
\int_{\xi_{\udl{i},\udl{j}}\in
\mathcal{A}}\rho^{-T_{l}(K_{l}-r)+\sum_{j=1}^{M}(c(j)-\sum_{i=1}^{N}\xi_{i,j})}d\xi_{\udl{i},\udl{j}}.
\end{equation}
Now we prove a Lemma that shows that the exponent of the integrand
in the upper bound from
\eqref{eq:UpperBoundAverageErrorProbabilityFirstTerm} is negative
for $\xi_{\udl{i},\udl{j}}\ge 0$.
\begin{lem}\label{lem:ErrorIntegrandAlwaysNegative}
consider $\xi_{i,j}\ge 0$ for $1\le i\le N$ and $1\le j\le M$. The
sum
$$c(j)-\sum_{i=1}^{N}\xi_{i,j}\le 0$$
for every $1\le j\le M$.
\end{lem}
\begin{proof}
See appendix \ref{Append:ErrorProbUpperBoundPositive}.
\end{proof}
In a similar manner to \cite{TseDivMult2003},
\cite{ElGamalLAST2004}, for a very large $\rho$ and a finite
integration range, we can approximate the integral by finding the
most dominant exponential term in
\eqref{eq:UpperBoundAverageErrorProbabilityFirstTerm}. Based on
Lemma \ref{lem:ErrorIntegrandAlwaysNegative} we know that the
exponent of the integrand is always negative. Hence, we can
approximate the upper bound by finding
$$\min_{\xi_{\udl{i},\udl{j}}\in
\mathcal{A}}T_{l}(K_{l}-r)+\sum_{j=1}^{M}(\sum_{i=1}^{N}\xi_{i,j}-c(j)).$$
As $\sum_{i=1}^{N}\xi_{i,j}-c(j)\ge 0$ the minimum is achieved when
$\sum_{i=1}^{N}\xi_{i,j}-c(j)=0$ for $1\le j\le M$. This can be
achieved for instance by taking $\xi_{i,j}=0$ for $1\le i\le N$,
$1\le j\le M$. In this case we get that the diversity order equals
$T_{l}(K_{l}-r)$ which is the best diversity order possible for IC's
of complex dimension $K_{l}T_{l}$.

Next we upper bound the second term of the error probability from
\eqref{eq:UpperBoundErrorProbDueExact2Ranges}. For
$\xi_{\udl{i},\udl{j}}\in\ol{\mathcal{A}}$ we upper bound the
average decoding error probability per channel realization by 1. In
this case we get
$$\int_{\xi_{\udl{i},\udl{j}}\in \ol{\mathcal{A}}}\rho^{-\sum_{j=1}^{M}\sum_{i=1}^{N}\xi_{i,j}}d\xi_{\udl{i},\udl{j}}.$$
Again we approximate this integral by calculating the most dominant
exponential term, i.e.
$\min_{\xi_{\udl{i},\udl{j}}\in\ol{\mathcal{A}}}\sum_{i=1}^{N}\sum_{j=1}^{M}\xi_{i,j}$.
The minimal value for this case is also $T_{l}(K_{l}-r)$. Hence, we
get a diversity order $T_{l}(K_{l}-r)$ for the second term. As a
result we can state that for both terms in
\eqref{eq:UpperBoundErrorProbDueExact2Ranges} we get the same
diversity order, and the transmission scheme diversity order is
upper bounded by $T_{l}(K_{l}-r)$. The proof is concluded.

\section{Proof of Lemma \ref{lem:ErrorIntegrandAlwaysNegative}}\label{Append:ErrorProbUpperBoundPositive}
We know that
\ifthenelse{\equal{\singlecolumntype}{1}}
{\begin{equation*}
c(j)=\sum_{s=0}^{\min(j,L)-2}a(s,\udl{\xi}_{j})+
(N-\min(j,L)+1)a(\min(j,L)-1,\udl{\xi}_{j})
\end{equation*}}
{\begin{align*}
&c(j)=\sum_{s=0}^{\min(j,L)-2}a(s,\udl{\xi}_{j})\\+
&(N-\min(j,L)+1)a(\min(j,L)-1,\udl{\xi}_{j})
\end{align*}}
where
$$a(k,\udl{\xi}_{j})=\min_{s\in\{k+1,\dots,N\}}\xi_{s,j}\qquad 0\le k\le \min(j,L)-1$$
and by definition
$$a(\min(j,L)-1,\udl{\xi}_{j})\ge\dots\ge a(0,\udl{\xi}_{j})\ge 0.$$
In order to prove the Lemma we begin with
$a(\min(j,L)-1,\udl{\xi}_{j})$. We know that
\begin{equation}
\sum_{s=\min(j,L)}^{N}\xi_{s,j}\ge
(N-\min(j,L)+1)\cdot\min_{s}\xi_{s,j}
\end{equation}
where $s\in\{\min(j,L),\dots,N\}$. We can also see that
\begin{equation}
\xi_{k+1,j}\ge \min_{s\in\{k+1,\dots,N\}}\xi_{s,j}
\end{equation}
for $0\le k\le \min(j,L)-2$. Hence we get
$$c(j)-\sum_{i=1}^{N}\xi_{i,j}\le 0.$$
This concludes the proof.

\section{Proof of Theorem \ref{Th:LowerBoundDiversityOrderLattices}}\label{Append:LatticesDiversityOrder}
We prove that there exists a sequence of
$2K_{l}T_{l}$-real dimensional lattices (as a function of $\rho$)
that attains the same diversity order as in Theorem
\ref{Th:LowerBoundDiversityOrder}. By using the
\emph{Minkowski-Hlawaka-Siegel} Theorem
\cite{PoltirevJournal},\cite{LekkerkerkerGeomety}, we upper bound
the error probability of the ensemble of lattices, for each channel
realization. This upper bound equals to the upper bound derived in
Theorem \ref{Th:UpperBoundErrorProb}. Then we average the upper
bound over all channel realizations, and receive the desired
diversity order.

We consider a $2K_{l}T_{l}$-real dimensional ensemble of lattices,
transmitted using the transmission scheme defined in subsection
\ref{subsec:TheTransmissionScheme}. We spread the first $K_{l}T_{l}$
dimensions of the lattice on the real part of the non-zero entries
of $G_{l}$, and the other $K_{l}T_{l}$ dimensions of the lattice on
the imaginary part of the non-zero entries of $G_{l}$. Each lattice
in the ensemble has transmitter density $\gamma_{tr}=\rho^{rT_{l}}$,
i.e. multiplexing gain $r$. We begin by analyzing the performance of
the ensemble of lattices in the receiver, for each channel
realization. We assume a certain channel realization that induces a
receiver VNR
$\mu_{rc}=\rho^{1-\frac{r}{K_{l}}-\sum_{i=1}^{K_{l}T_{l}}\frac{\eta_{i}}{K_{l}T_{l}}}$,
where $\udl{\eta}\ge 0$. For each lattice in the ensemble we get
that the channel realization induces a new lattice in the receiver,
$H_{eff}^{(l)}\cdot\udl{x}$, with density $\gamma_{rc}$ in
accordance with \eqref{eq:ExtendedChannelModel} and subsection
\ref{subsec:TheEffectiveChannel}. For lattices with regular lattice
decoding, the error probability is equal among all codewords. Hence,
it is sufficient to analyze the lattice's zero codeword error
probability. We define the indication function
$$I_{Ball(0,2R)}(\udl{x})=\left\{ \begin{array}{ll}
1, & \lv\udl{x}\rv\le 2R\\
0, & else
\end{array}\right..$$
In a similar manner to \eqref{eq:MLPointErrorProbUpperBound} we can
state that for each lattice induced in the receiver,
$\Lambda_{\rc}$, the lattice zero codeword error probability is
upper bounded by
\ifthenelse{\equal{\singlecolumntype}{1}}
{\begin{equation}\label{eq:LatticesUpperBoundErrorProbCrudeForm}
\sum_{\udl{x}\in\Lambda_{\rc},\udl{x}\neq
0}I_{{Ball(0,2R_{\eff})}}(\udl{x})\cdot Pr(\rv
\udl{\tilde{n}}_{\ex}\lv>\rv \udl{x}-\udl{\tilde{n}}_{\ex}\lv)
+Pr(\lVert\underline{\tilde{n}}_{\mathrm{ex}}\rVert\ge R_{\eff})
\end{equation}}
{\begin{align}\label{eq:LatticesUpperBoundErrorProbCrudeForm}
\sum_{\udl{x}\in\Lambda_{\rc},\udl{x}\neq
0}I_{{Ball(0,2R_{\eff})}}(\udl{x})\cdot Pr(\rv
\udl{\tilde{n}}_{\ex}\lv>\rv \udl{x}-\udl{\tilde{n}}_{\ex}\lv)\nonumber\\
+Pr(\lVert\underline{\tilde{n}}_{\mathrm{ex}}\rVert\ge R_{\eff})
\end{align}}
where $\frac{R_{\eff}^{2}}{2K_{l}T_{l}\sigma^{2}}=\mu_{rc}$, and
$\udl{\tilde{n}}_{\ex}$ is the effective noise in the
$K_{l}T_{l}$-complex hyperplane where $\Lambda_{rc}$ resides in. By
defining $f_{rc}(\udl{x})=I_{{Ball(0,2R_{\eff})}}(\udl{x})\cdot
Pr(\rv \udl{\tilde{n}}_{\ex}\lv>\rv
\udl{x}-\udl{\tilde{n}}_{\ex}\lv)$, we can rewrite the upper bound
on the error probability from
\eqref{eq:LatticesUpperBoundErrorProbCrudeForm}
\begin{equation}\label{eq:LatticeUpperBoundErrorProb}
\sum_{\udl{x}\in\Lambda_{\rc},\udl{x}\neq
0}f_{\mathrm{rc}}(\udl{x})+Pr(\lVert\underline{\tilde{n}}_{\mathrm{ex}}\rVert\ge
R_{\eff}).
\end{equation}
Note that
\begin{equation}\label{eq:LatticesEnsembleAverageDecodingErrorProb}
\gamma_{rc}\int_{\mathbb{R}^{2K_{l}T_{l}}}f_{rc}(\udl{x})d\udl{x}+Pr(\lVert\underline{\tilde{n}}_{\mathrm{ex}}\rVert\ge
R_{\eff})
\end{equation}
is equal to the expression in \eqref{eq:IntErrorUpperBound}, where
$\gamma_{rc}$ is the density of the lattice induced in the receiver
$\Lambda_{\rc}$, as defined above.

We need to show that there exists a single probability measure for
all channel realizations, that gives an average decoding error
probability over the ensemble, which is upper bounded by
\eqref{eq:LatticesEnsembleAverageDecodingErrorProb}. Hence, we
consider the ensemble of lattices in the transmitter which is fixed
for each channel realization. For this reason we define
\begin{equation}\label{eq:TheEffectiveChannelModelWithTheTransmitter}
\udl{y}_{\ex}^{'}=\big(H_{\eff}^{(l)\dagger}\cdot
H_{\eff}^{(l)}\big)^{-1}H_{\eff}^{(l)\dagger}\cdot\udl{y}_{\ex}.
\end{equation}
Note that the operation in
\eqref{eq:TheEffectiveChannelModelWithTheTransmitter} does not
change the error probability of the lattice when we use regular
lattice decoding. Each lattice in the ensemble has density
$\gamma_{tr}=\rho^{rT_{l}}$. Now we define the following indication
function
$$I_{ellipse(H,2R)}(\udl{x})=\left\{ \begin{array}{ll}
1, & \rv H\cdot\udl{x}\lv\le 2R\\
0, & else
\end{array}\right.,$$
that is the function is one if $\udl{x}$ is within the ellipse and zero otherwise. Let us denote the error probability of a lattice in the ensemble for certain channel realization $\udl{\eta}$ by
$P_{e}^{(\nu)}(\udl{\eta},\rho)$, where $\nu$ is a random variable
that represents a certain lattice in the ensemble. Using regular
lattice decoding, we get the following upper bound on the error
probability for each lattice codeword
\ifthenelse{\equal{\singlecolumntype}{1}}
{\begin{equation}\label{eq:LatticeTransmiterEnsembleUpperBoundErrorProbCrudeForm}
P_{e}^{(\nu)}(\udl{\eta},\rho)\le\sum_{\udl{x}\in\Lambda_{\tr},\udl{x}\neq
0}I_{{ellipse(H_{\eff}^{(l)},2R_{\eff})}}(\udl{x})\cdot Pr\big(\rv
A\cdot\udl{\hat{n}}_{\ex}\lv> \rv A\cdot
(\udl{x}-\udl{\hat{n}}_{\ex})\lv\big) +Pr(\lVert
A\cdot\underline{\hat{n}}_{\mathrm{ex}}\rVert\ge R_{\eff})
\end{equation}}
{\begin{align}\label{eq:LatticeTransmiterEnsembleUpperBoundErrorProbCrudeForm}
P_{e}^{(\nu)}(\udl{\eta},\rho)\le Pr(\lVert
A\cdot\underline{\hat{n}}_{\mathrm{ex}}\rVert\ge R_{\eff})&+\nonumber\\
\sum_{\udl{x}\in\Lambda_{\tr},\udl{x}\neq
0}I_{{ellipse(H_{\eff}^{(l)},2R_{\eff})}}(\udl{x})\cdot Pr\big(\rv
A &\cdot\udl{\hat{n}}_{\ex}\lv \nonumber\\
>\rv A \cdot(\udl{x}-\udl{\hat{n}}_{\ex})\lv\big)
\end{align}}
where $A$ is a $K_{l}T_{l}$x$K_{l}T_{l}$ matrix that satisfies
$A^{\dagger}A=H_{\eff}^{(l)\dagger}H_{\eff}^{(l)}$, $\Lambda_{tr}$
is the lattice from the ensemble that corresponds to $\nu$ and
$\udl{\hat{n}}_{\ex}\sim
CN\big(0,(H_{\eff}^{(l)\dagger}H_{\eff}^{(l)})^{-1}\big)$. Note that
\eqref{eq:LatticeTransmiterEnsembleUpperBoundErrorProbCrudeForm} is
equal to \eqref{eq:LatticeUpperBoundErrorProb}, and the
corresponding terms in the expressions are also equal.

Let us define
$g_{rc}(\udl{x})=I_{{ellipse(H_{\eff}^{(l)},2R_{\eff})}}(\udl{x})\cdot
Pr\big(\rv A\udl{\hat{n}}_{\ex}\lv> \rv A
(\udl{x}-\udl{\hat{n}}_{\ex})\lv\big)$. We get that
\begin{equation}\label{eq:EqualityBetweenElipseAndBallIntegral}
\gamma_{\tr}\int_{\mathbb{R}^{2K_{l}T_{l}}}g_{\rc}(\udl{x})d\udl{x}=\gamma_{\rc}\int_{\mathbb{R}^{2K_{l}T_{l}}}f_{\rc}(\udl{x})d\udl{x}.
\end{equation}

Next we show that by averaging the upper bound in
\eqref{eq:LatticeTransmiterEnsembleUpperBoundErrorProbCrudeForm}
over the ensemble of lattices in the transmitter, with the correct
probability measure, we get
\begin{equation}\label{eq:TheAverageLatticesUpperBound}
E_{\nu}\{P_{e}^{(\nu)}(\udl{\eta},\rho)\}\le
\gamma_{rc}\int_{\mathbb{R}^{2K_{l}T_{l}}}f_{rc}(\udl{x})d\udl{x}+Pr(\lVert\underline{\tilde{n}}_{\mathrm{ex}}\rVert\ge
R_{\eff}).
\end{equation}
We prove \eqref{eq:TheAverageLatticesUpperBound} by using the
\emph{Minkowski-Hlawaka-Siegel} theorem \cite{PoltirevJournal}:
\begin{theorem}
(\label{Th:MinHal}Minkowski-Hlawaka-Siegel Theorem) In the set of
all the lattices of density $\gamma$ in $\mathbb{R}^{2K_{l}T_{l}}$,
there exists a probability measure $\nu$ such that for any Riemann
integrable function $f(\udl{x})$ which vanishes outside some bounded
region we have
\begin{equation}\label{eq:MinHalTh}
E_{\nu}\{\sum_{\udl{x}\in\Lambda}g(\udl{x})\}=\gamma\int_{\mathbb{R}^{2K_{l}T_{l}}}g(\udl{x})d\udl{x}
\end{equation}
where $E_{\nu}\{\cdot\}$ represents the expectation with respect to
the measure $\nu$.
\end{theorem}
Note that considering a $2K_{l}T_{l}$-real dimensional lattices enables us to use this theorem. Hence, by
choosing $\gamma=\gamma_{\tr}$, $g(\udl{x})=g_{\rc}(\udl{x})$, and considering \eqref{eq:LatticeTransmiterEnsembleUpperBoundErrorProbCrudeForm}, \eqref{eq:EqualityBetweenElipseAndBallIntegral} we
get the desired upper bound \eqref{eq:TheAverageLatticesUpperBound}. As a
result, we can upper bound the ensemble average decoding error
probability for each channel realization by the upper bound from
Theorem \ref{Th:UpperBoundErrorProb}
\eqref{eq:ICUpperBoundErrorProbForAllCases}.

Now we are ready to lower bound the diversity order. According to
Theorem \ref{Th:MinHal} there exists a single probability measure
that satisfies \eqref{eq:MinHalTh}, for any Riemann integrable
function that vanishes outside some bounded region. Based on
\eqref{eq:LowerNoundOnThedeterminantHeff} and Lemma
\ref{lem:ErrorIntegrandAlwaysNegative}, we get for the set
$\{\xi_{\udl{i},\udl{j}}|\sum_{i=1}^{N}\sum_{j=1}^{M}\xi_{i,j}\le
T_{l}(K_{l}-r);\xi_{\udl{i},\udl{j}}\ge 0\}$ a set of functions,
$g_{\rc}(\udl{x})$, which are bounded. As a result we can upper
bound the ensemble average decoding error probability for this set
by the expression from \eqref{eq:ICUpperBoundErrorProbForAllCases}.
For the set of events
$\{\xi_{\udl{i},\udl{j}}|\sum_{i=1}^{N}\sum_{j=1}^{M}\xi_{i,j}>
T_{l}(K_{l}-r);\xi_{\udl{i},\udl{j}}\ge 0\}$ we upper bound the
ensemble average decoding error probability by 1. This bounds are
the exact same bounds we used in order to average over the channel
realizations in Theorem \ref{Th:LowerBoundDiversityOrder}. Hence, by
averaging over the channel realizations we get for the ensemble the
same lower bound on the diversity order as in Theorem
\ref{Th:LowerBoundDiversityOrder}. This concludes the proof.

\section{Proof of Corollary
\ref{Cor:SequenceICAttainsTheEntireLine}}\label{append:SequenceICAttainsTheEntireLine}
Let $P_{e}(S(\rho),r)$ denote the average decoding error probability
of the IC $S(\rho)$ with density $\gamma_{tr}=\rho^{rT}$. Since
$S_{KT}(\rho)$ has density $\gamma_{tr}=1$ for every $\rho$, this
IC's sequence has multiplexing gain $r=0$. Hence, in accordance with
our definitions, we denote $S_{KT}(\rho)$ average decoding error
probability by $P_{e}(S_{KT}(\rho),0)$. Assume
$$P_{e}(S_{KT}(\rho),0)=A^{'}(\rho)\rho^{-d}$$
where $-\lim_{\rho\to\infty}\log_{\rho}P_{e}(S_{KT}(\rho),0)=d$,
i.e. $S_{KT}(\rho)$ has diversity order $d$. By scaling the sequence
of IC's such that
$$\ol{S}_{KT}(\rho)=S_{KT}(\rho)\cdot\rho^{-\frac{r}{2K}}\qquad 0\le r\le K,$$
i.e., scaling $S_{KT}(\rho)$ by a factor of $\rho^{-\frac{r}{2K}}$,
we get that $\ol{S}_{KT}(\rho)$ has density $\gamma_{tr}=\rho^{rT}$,
multiplexing gain $r$ and so its error probability
$$P_{e}(\ol{S}_{KT}(\rho),r)=P_{e}(S_{KT}(\rho^{1-\frac{r}{K}}),0)=A^{'}(\rho^{1-\frac{r}{K}})\rho^{-d(1-\frac{r}{K})}.$$
As a result we get
$-\lim_{\rho\to\infty}\log_{\rho}P_{e}(\ol{S}_{KT}(\rho),r)=d(1-\frac{r}{K})$,
i.e. $\ol{S}_{KT}(\rho)$ has diversity order $d(1-\frac{r}{K})$.

\section{Proof of Corollary
\ref{cor:SingleICAttainsTheOptimalDMT}}\label{append:SingleICAttainsTheOptimalDMT}
The proof of this corollary relies heavily on Theorem
\ref{Th:UpperBoundErrorProb}. We begin by describing the $L$
ensembles of IC's and how they are transmitted. Then we use
averaging arguments in order to show that there exists a singe
sequence of IC's that attains the optimal DMT.

We begin by considering a sequence of $K_{0}T_{0}$-complex dimensional IC's with multiplexing gain $r=0$, i.e. the transmitter density $\gamma_{tr}=1$ for any $\rho$. In a similar manner to Theorem \ref{Th:UpperBoundErrorProb}, we first consider an ensemble of finite constellations drawn uniformly within $\cube_{K_{0}T_{0}}(b)\subset\mathbb{C}^{K_{0}T_{0}}$. Each code-book contains $\lfloor
\gamma_{\mathrm{tr}}b^{2K_{0}T_{0}}\rfloor=\lfloor
b^{2K_{0}T_{0}}\rfloor$ points, where each point
is drawn uniformly within $cube_{K_{0}T_{0}}(b)$. Let us denote a certain finite constellation in the ensemble by $C_{FC}(\rho,K_{0}T_{0},b)\subset\cube_{K_{0}T_{0}}(b)$. We extend each finite constellation in the ensemble into an IC in a similar manner to \eqref{eq:ExtendingFinteConstoICInTheTransmitter}
\begin{equation}\label{eq:ICEnsembleForTheCaseK0T0}
IC(\rho,K_{0}T_{0})=C_{FC}(\rho,K_{0}T_{0},b)+(b+b^{'})\cdot\mathbb{Z}^{2K_{0}T_{0}}.
\end{equation}
By choosing $b=\sqrt{\frac{K_{0}T_{0}}{\pi e}}\rho^{\frac{K_{0}T_{0}}{2}+2\epsilon}$ and $b^{'}=\sqrt{\frac{K_{0}T_{0}}{\pi e}}\rho^{\frac{K_{0}T_{0}}{2}+\epsilon}$, we get a sequence of ensembles of IC's with multiplexing gain $r=0$. For a certain channel realization $\udl{\eta}\ge 0$ we get in accordance with Theorem \ref{Th:UpperBoundErrorProb}
\begin{equation}
\ol{P_{e}}(\rho,\udl{\eta},K_{0}T_{0})\le
D(K_{0}T_{0})\rho^{-T_{0}K_{0}+\sum_{i=1}^{K_{0}T_{0}}\eta_{i}}
\end{equation}
where $\ol{P_{e}}(\rho,\udl{\eta},K_{0}T_{0})$ is the average
decoding error probability of the $K_{0}T_{0}$-complex dimensional
ensemble of IC's. From Theorem \ref{Th:LowerBoundDiversityOrder} we
know that by transmitting the ensemble of IC's over the transmission
matrix $G_{0}$, and averaging over the channel realizations, we get
diversity order $d_{K_{0}}=MN$. Transmitting over $G_{0}$ gives us a
$K_{0}T_{0}$-complex dimensional ensemble of IC's within
$\mathbb{C}^{MT_{0}}$.

Next we derive from the $K_{0}T_{0}$-complex dimensional ensemble of IC's, another $K_{l}T_{l}$-complex dimensional ensemble of IC's, where $l=1,\dots,L-1$. For each IC, $IC(\rho,K_{0}T_{0})$, in the ensemble we take the first $\lfloor
b^{2K_{l}T_{l}}\rfloor$ points in $C_{FC}(\rho,K_{0}T_{0},b)$. We take the components of these points inside $\cube_{K_{l}T_{l}}(b)$, and denote this
new finite constellation as $C_{FC}(\rho,K_{l}T_{l},b)$. Then we replicate these points in a similar manner to \eqref{eq:ICEnsembleForTheCaseK0T0}. In this case we get
a new $K_{l}T_{l}$-complex dimensional IC
\begin{equation}
IC(\rho,K_{l}T_{l})=C_{FC}(\rho,K_{l}T_{l},b)+(b+b^{'})\cdot\mathbb{Z}^{2K_{l}T_{l}}.
\end{equation}
By doing it to each IC in the ensemble, we get a new $K_{l}T_{l}$-complex dimensional ensemble of IC's. This new ensemble is equivalent to ensemble of IC's generated by drawing uniformly $\lfloor
b^{2K_{l}T_{l}}\rfloor$ points inside $\cube_{K_{l}T_{l}}(b)$, and then replicate these points according to $(b+b^{'})\mathbb{Z}^{2K_{l}T_{l}}$. Each IC sequence in this ensemble has multiplexing gain $r=0$. Since $b>\sqrt{\frac{K_{l}T_{l}}{\pi e}}\rho^{\frac{K_{l}T_{l}}{2}+2\epsilon}$ and $b^{'}>\sqrt{\frac{K_{l}T_{l}}{\pi e}}\rho^{\frac{K_{l}T_{l}}{2}+\epsilon}$, we get in accordance with Theorem \ref{Th:UpperBoundErrorProb} that for a certain channel realization $\udl{\eta}\ge 0$
\begin{equation}
\ol{P_{e}}(\rho,\udl{\eta},K_{l}T_{l})\le
D(K_{l}T_{l})\rho^{-T_{l}K_{l}+\sum_{i=1}^{K_{l}T_{l}}\eta_{i}}
\end{equation}
where $\ol{P_{e}}(\rho,\udl{\eta},K_{l}T_{l})$ is the average
decoding error probability of the $K_{l}T_{l}$-complex dimensional
ensemble of IC's. By transmitting this ensemble of IC's on the
transmission matrix $G_{l}$, and averaging over the channel
realizations, we get diversity order $d_{K_
{l}}=(M-l)(N-l)+l(N+M-2\cdot l-1)$. Transmitting over $G_{l}$ gives
us a $K_{l}T_{l}$-complex dimensional ensemble of IC's within
$\mathbb{C}^{MT_{l}}$.

From the sequential structure of the transmission scheme we get that omitting the $2\cdot l$ rightmost columns of $G_{0}$ yields $G_{l}$. Hence we can derive from the $K_{0}T_{0}$-complex dimensional ensemble of IC's, that attains diversity order $d_{K_{0}}$, another $K_{l}T_{l}$-complex dimensional ensemble of IC's the attains diversity order $d_{K_{l}}$, where $l=1,\dots,L-1$. We attain it by diluting the points of each $K_{0}T_{0}$-complex dimensional IC in the ensemble in the aforementioned manner, and then reducing its dimensionality by dropping the $2\cdot l$ rightmost columns of $G_{0}$.

So far we have shown the connection between the ensembles. Now we would like to show that there exists a certain sequence of $K_{0}T_{0}$-complex dimensional IC's, that gives us the desired diversity orders by diluting its points and adapting its dimensionality. We denote the average decoding error probability of the $K_{l}T_{l}$-complex dimensional ensemble of IC's by $A_{l}(\rho)\rho^{-d_{K_{l}}}$, where $\lim_{\rho\to\infty}\frac{\log(A_{l}(\rho))}{\log(\rho)}=0$. We also define $I_{l,\rho}$ as the event where a $K_{l}T_{l}$-complex dimensional IC in the ensemble has average decoding error probability which is smaller or equal to $(L+1)A_{l}(\rho)\rho^{-d_{K_{l}}}$, where $l=0,\dots,L-1$. From averaging arguments we know that $Pr(I_{l,\rho})\ge\frac{L}{L+1}$. We wish to show that the probability of the event $\{I_{0,\rho}\cap I_{1,\rho}\dots\cap I_{L-1,\rho}\}$ is bounded away from zero. From averaging arguments we know that
$$Pr(I_{0,\rho}\cap I_{1,\rho}\dots\cap I_{L-1,\rho})\ge1-\sum_{i=0}^{L-1}Pr(I_{i,\rho})\ge\frac{1}{L+1}.$$
Hence there must exist a sequence of $K_{0}T_{0}$-complex dimensional IC's that attains diversity order $d_{K_{0}}$ and has multiplexing gain $r=0$, from which we can derive for each $l=1,\dots,L-1$, a sequence of $K_{l}T_{l}$-complex dimensional IC's with multiplexing gain $r=0$ and diversity order $d_{K_{l}}$.

Next we show that these $L$ sequences attain the optimal DMT. Consider a sequence of $K_{l}T_{l}$-complex dimensional IC's, that has multiplexing gain $r=0$ and attains diversity order $d_{K_{l}}$. From Corollary \ref{Cor:SequenceICAttainsTheEntireLine} we know that scaling this sequence by a scalar $\rho^{-\frac{r}{2K_{l}}}$ yields a new sequence of IC's with multiplexing gain $r$ and diversity order
$$d_{K_{l}}(r)=(M-l)(N-l)-(r-l)(N+M-2\cdot l-1)$$
where $0\le r\le K_{l}$ and $l=0,\dots,L-1$. Each of the $L$ straight lines $d_{K_{l}}(r)$, $l=0,\dots,L-1$, coincides with a different segment out of the $L$ segments of the optimal DMT. This concludes the proof.
\end{appendices}

\section*{Acknowledgment}
The authors wish to thank Joseph J. Boutros for interesting
discussions regarding this work, and also to Or Ordentlich for
fruitful discussions on subsection
\ref{subsec:LatticeVsLatticeBasedFC}.

\bibliographystyle{IEEEtran}
\bibliography{IEEEabrv,YairRef}

\end{document}